\documentclass[a4paper, 11pt]{article}
\usepackage[T1]{fontenc}
\usepackage[latin1]{inputenc}
\usepackage{appendix}
\usepackage{amsmath}
\usepackage{amsthm}
\usepackage{amsfonts}
\usepackage{amssymb}
\usepackage{authblk}
\usepackage{bm}
\usepackage{dsfont}

\newtheorem{lemma}{Lemma}
\newtheorem{proposition}{Proposition}
\newtheorem{corollary}{Corollary}
\newtheorem{definition}{Definition}
\newtheorem{theorem}{Theorem}

\newtheorem{example}{Example}

\newtheorem*{remark}{Remark}

\usepackage{natbib}
\usepackage[top=3cm,bottom=3cm,left=3cm,right=3cm]{geometry}
\usepackage{setspace}
\usepackage{xcolor}
\usepackage{graphicx}
\usepackage{hyperref}
\usepackage{comment}
\usepackage{subfig}
\usepackage{tikz}
\usetikzlibrary{graphs}
\usetikzlibrary{arrows.meta}
\usepackage[normalem]{ulem}
\usepackage{enumitem}

\begin{document}

\title{Causal Persuasion\footnote{We are grateful to Murali Agastya, Chiara Aina, Kadir Atalay, Jean-Michel Benkert, Andreas Bjerre-Nielsen, Juan Carlos Carbajal, Martin Dufwenberg, Andrew Ellis, Christian Gillitzer, Benjamin Golub, Alessandro Ispano, Anton Kolotilin, Klaus Kultti, Jay Lu, Joel Sobel, as well as seminar and conference audiences at Copenhagen, Sydney, UNSW, NET, EWMES 2025, AETW 2026 for the many valuable comments. } 
}
\author[1]{Anastasia Burkovskaya}
\author[2]{Egor Starkov\footnote{Burkovskaya: anastasia.burkovskaya@sydney.edu.au; Starkov: egor.starkov@econ.ku.dk.}}
\affil[1]{University of Sydney, School of Economics}
\affil[2]{University of Copenhagen, Department of Economics}

\maketitle

\vspace{-20pt}
\begin{abstract}
	Correlation is not causation... until the right variables are on the table. We propose a model of causal persuasion, in which a sender selectively discloses variables alongside a subjective causal model linking them. Persuasion requires that the model rules out rival explanations, on top of being consistent with the underlying data distribution. To establish a causal link, the sender often needs to disclose at most two well-chosen variables, whereas dispelling a perceived link, every common cause must be disclosed. This highlights a fundamental asymmetry: Establishing causality is often much easier than ruling it out.
	\\
	
	\noindent {\textbf{JEL-codes}}: D83\\
	
	\noindent {\textbf{Keywords}}: causal persuasion, causal models, directed acyclic graphs, strategic communication
\end{abstract}

\section{Introduction}

Media offers \emph{stories}. Politicians tell \emph{stories}. Advertising pushes \emph{stories}. Economists debate \emph{stories}.
``\emph{Why} Investors Are Worried About Japan's Bond Market.''\footnote{Bloomberg, Jul 23, 2025. URL: \url{https://www.bloomberg.com/news/articles/2025-07-23/why-japan-s-bond-market-has-investors-worried}, retrieved Aug 08, 2025.}
``\emph{Because} of Tariffs, our Economy is BOOMING!''\footnote{Donald Trump, Truth Social, Jun 03, 2025. URL: \url{https://truthsocial.com/@realDonaldTrump/posts/114617666432673584}, retrieved Aug 08, 2025.}
``You're Not You \emph{When} You're Hungry''\footnote{\url{https://www.youtube.com/watch?v=OTPJYZLD6L8}, retrieved Aug 08, 2025.}
``Do Financial Concerns \emph{Make} Workers Less Productive?''\footnote{\citet{kaur_financial_2025}}
All of these examples suggest causal relationships that rationalize the observed data. Some stories try their best to offer a good-faith explanation of the data. Some of those fail despite their best efforts.\footnote{``Leaders: Stop Confusing Correlation with Causation.'' Harvard Business Review, Nov 05, 2021. URL: \url{https://hbr.org/2021/11/leaders-stop-confusing-correlation-with-causation}, retrieved Aug 08, 2025.} 
Others come up with a spin that is favorable for them. What makes a causal story persuasive? How can a persuader instill a causal story with a listener? What does it take to debunk a story? These are the questions that we tackle in this paper. 

We develop a model of \emph{causal} persuasion in which a sender tries to persuade a receiver about a particular causal relationship. The receiver may or may not initially have knowledge of some real-world variables and a subjective causal model that links them. The sender can disclose some variables to either persuade the receiver of some causal model or to falsify the receiver's model and replace it with a new one. We provide conditions under which persuasion is possible and provide blueprints for effective persuasion mechanisms. We show that some properties of the true causal model can make it ``easy'' to persuade the receiver of a certain causal connection and/or to debunk the receiver's faulty model. Without these properties, the respective tasks can become ``difficult'' or even impossible. Finally, we show that the sender may be able to instill an incorrect model even if they are not able to lie about the factual data.

Consider the classic question of whether education raises wages. The correlation between the two is not disputed, but various rival stories are compatible with it: that education genuinely improves earnings, that higher-earning careers induce more schooling, or that underlying ability drives both. Persuasion in this context is not about disclosing more favorable correlation numbers, but about establishing the causal relation. Disclosing the right \emph{variables} can achieve this if their observed relationships rule out the rival stories. 

The same logic applies to debunking narratives. For example, the belief that immigration raises crime is popular and may fit the raw correlation, but does not seem to be grounded in reality.\footnote{Causal studies by \citet*{bianchi2012,marie_2024} find no effect of immigration on overall crime.}
Debunking this narrative may then have less to do with minimizing the perceived correlation and more with identifying the right confounding variables---conditioning on them can break the pattern. These are exactly the issues explored in our paper: when can revealing extra variables instill or debunk narratives; when do one or two variables suffice and which ones; and which outcomes can no disclosure ever achieve?

Specifically, we assume that there is a sender who observes the true causal model related to some variables of interest $x$ and $y$, captured by a directed acyclic graph (DAG). The sender can selectively disclose to a receiver some of the variables from this graph and propose a reduced model explaining them, which is accepted only if it is consistent with the data.
We show that a na{\"i}ve receiver, who treats such consistency as sufficient to adopt a model, can always be persuaded that $x$ affects $y$ without disclosing any extra variables, as long as $x$ and $y$ are correlated. In contrast, persuading a sophisticated receiver---who only adopts a narrative regarding $x$ and $y$ if all other narratives are ruled out by the data---is a nontrivial task. 

When the sender's narrative is consistent with the data, revealing only one or two carefully chosen variables in addition to $x$ and $y$ may be enough to persuade a sophisticated receiver. These variables allow the receiver to rule out the other possibilities and conclude that $x$ indeed causes $y$, as the sender suggests. However, if no such appropriate variables exist, then the sender's model may need to be impractically large to be persuasive. Persuasion may be difficult (in terms of required model size) or even impossible in this case. In particular, persuading the receiver of a causal direction opposite to the truth is altogether impossible (Theorem~\ref{thm:no_persuade_defective}). Despite this, the sender \emph{can} sometimes mislead the receiver about correlation arising from confounding variables and present it as causation, since it does not directly contradict the true causal graph. Further, under some conditions, whenever persuasion is possible, it is optimal for the sender to propose a model that directly links $x$ and $y$, as opposed to revealing all the proxy variables that (according to the sender) mediate the causal effect (Theorem~\ref{thm:direct}).

We then consider the case when the receiver has some pre-existing subjective model. The sender needs to debunk this model before they can persuade the receiver of anything else. We show that to be falsifiable, the receiver's model \emph{must} be causally incorrect. Merely omitting variables is not enough. Rather, this omission must lead to the receiver either ``flipping'' some causal connections relative to the true model, or creating new causal links that are not present in the true model.

If the receiver's model \emph{can} be debunked, doing so is qualitatively similar to persuading a sophisticated receiver with a blank mind. In particular, it can be done by revealing at most two appropriate variables if they exist (Theorem~\ref{thm:debunk_simple}) and can be arbitrarily difficult otherwise.
To illustrate the latter case, debunking some causal link (e.g., that immigration causes crime) and persuading the receiver that no such link exists requires identifying and revealing \emph{all} common causes of the two variables. Because there may be arbitrarily many such confounders, this can be very burdensome for both sender and receiver. In practice, disproving a causal link can often be harder than persuading the receiver that the link runs in the opposite direction---meaning deception can be easier than establishing the truth. This stands in contrast to the classical Humean asymmetry for universal claims (e.g., ``all swans are white''), which are difficult to verify but easy to falsify \citep{hume_1748,alnajjar_2014}. The difference arises because, unlike universal claims, ruling out a causal relationship requires eliminating all competing explanations.

\medskip 

We next review relevant literature in Section \ref{sec:literature} and present an illustrative example in Section \ref{sec:example}. The baseline model is set up in Section \ref{sec:model}, although the fully specified game is introduced only in Section \ref{sec:microfound}. The main results are contained in Sections \ref{sec:persuade_blank} and \ref{sec:preexist}. Section \ref{sec:applications} discusses additional implications of the model in the context of the illustrative example and another application. Section \ref{sec:discussion} discusses some of the assumptions behind the model and outlines directions for future research.
Most formal proofs are relegated to Appendix \ref{sec:proofs}. Supplementary Appendix \ref{sec:examples} offers additional examples to illustrate selected concepts from the main text.

\section{Related Literature} \label{sec:literature}

Our paper connects the literature on misspecified causal models with the literature on persuasion, especially model persuasion. For an introduction to causal models and Bayesian networks, see \citet{pearl_causality_2009} for a textbook treatment and \citet*{guo_survey_2021} or \citet*{zanga_survey_2022} for recent surveys of that literature. For a more comprehensive overview of the fast-growing literature on mental models in economics, see an excellent survey by \citet*{ambuehl_mental_2026}.

\paragraph{Misspecified causal models.}
\citet{spiegler_bayesian_2016} pioneered the Bayesian-network approach to causal misperceptions in economics. In this model, an agent fits a subjective DAG to the objective data, which results in a biased belief. This paper and the program that followed (surveyed in \citealp{spiegler2020survey}; see also \citealp{spiegler_can_2020,spiegler2026bci}; and \citealp{ellis_limited_2026}) trace the behavioral consequences of holding a misspecified model. 
We invert this approach by assuming the receiver starts off with no model, and instead ask how causal misperceptions could arise as a result of a strategic interaction with a sender who wants to instill a specific model.
As one example illustrating the importance of this agenda, \citet{spiegler2022reverse} studies the consequences of an exogenously assumed causal reversal. Our Theorem~\ref{thm:no_persuade_defective} provides the conceptual counterpart, showing that no such reversal can be \emph{induced} in a sophisticated receiver.

\paragraph{Model persuasion.}
\citet{schwartzstein_using_2021}, \citet{aina_tailored_2025}, and \citet{ispano_perils_2025} explore models in which the sender offers models capturing the joint distributions of different variables, and the receiver adopts them according to their fit with the observed data sample. \citet{bauch2025} construct a similar model in which the receiver initially faces ambiguity regarding the true model. All of these papers vary the interpretation of some fixed evidence. We instead allow the sender to vary the evidence by choosing which variables the receiver observes.
 
Other recent papers give the sender control over both information and its interpretation. These include \citet{eliaz_news_2024} on news media and \citet*{ba2026worldviews} on jointly designed signals and narratives, and \citet{ichihashi_meng} who study the joint design problem generally. Our sender also controls both instruments, but instead of disclosing variable realizations the sender can only disclose whole variables.
 
The acceptance criterion we use is different from all of the aforementioned papers: instead of evaluating the fit to a finite data sample, our receiver demands consistency with the whole data distribution. This restricts belief movement among models that explain the evidence and limits the set of plausible models the sender can propose. However, we show that deception is still possible with our approach. \citet*{eliaz_cheating_2021} bridge the two approaches by deriving a bound on the correlation distortion that a proposed DAG can generate given the data.

\paragraph{Procedural receivers and evidence.}
Beyond fit-based and consistency-based model evaluation rules, other criteria exist in the literature.
\citet{eliaz_wasonian_2025} study persuasion of a receiver following a non-Bayesian sample-based procedure, in the lineage of \citet{glazer2006,glazer2012,glazer2014}. Their sender proposes a causal DAG to the receiver, who then uses a single random data point to evaluate the proposed model. They characterize the distribution of DAGs that prevails in society as a result of the sender's strategic concerns.\footnote{\citet*{montielolea2022,eliaz_model_2020,eliaz_false_2025}; and \citet{ba_robust_2026} provide alternative characterizations of models that would be predominant in society in an equilibrium, but do so in the absence of a strategic sender and with different model selection/survival criteria.}
\citet*{eliaz2021interpretations} propose another model persuasion setting with boundedly-rational receivers: they let the sender supply a model that interprets the data in a setting where the receiver does not make strategic inferences from the supplied model.

Experimentally, \citet{ambuehl2025} document both fit-based model selection and worst-case reasoning when subjects cannot identify the correct model. \citet{aina_weighting_2026} show that subjects are not Bayesian regarding model fit, and many rely only on the best-fitting model. More generally, \citet{kendall_complexity_2024} show model formation is cognitively costly. In turn, \citet{kendall2022narratives} and \citet{barron2025} show that supplied causal narratives move beliefs, validating the whole model persuasion approach.

\paragraph{Strategic disclosure.}
More broadly, our work relates to the literature on strategic communication that has long explored the question ``what makes communication persuasive?'' See \citet{little_bayesian_2023} for a brief overview of the literature. Specifically, our model is connected to the literature on disclosure of verifiable information, in which the sender can withhold data but cannot produce fake data (\citealp{grossman_disclosure_1980,milgrom_good_1981}; see \citealp{dranove_quality_2010} for an overview and \citealp{di_tillio_strategic_2021} for a recent alternative model). The literature on Bayesian Persuasion considers a problem that is similar to disclosure models, but allows the sender to design information in a rich way (\citealp{kamenica_bayesian_2011}; see \citealp{bergemann_information_2019} and \citealp{kamenica_bayesian_2019} for recent overviews).
Unlike both strands, our model allows the sender to disclose variables rather than their realizations, which has implications for which causal models remain plausible. In particular, the sender in our model can offer an incorrect model to go with factually correct data. 


\section{Illustrative example} \label{sec:example}

Consider the following illustrative example. It is well established that an MBA degree is correlated with higher lifetime earnings.\footnote{``The MBA Premium: What MBAs Earn Over A Lifetime Will Shock You'' Poets and Quants, May 05, 2021. URL: \url{https://poetsandquants.com/2021/05/05/lifetime-earnings-of-mbas/}, retrieved Nov 21, 2025.} 
However, it may not be completely clear which way the causal link goes. Business schools naturally want everyone to believe that their degrees offer substantial value to prospective students. However, an employee weighing whether to leave their job for an MBA may also consider the possibility that the degree has no real value, and that the causal link runs in the opposite direction---high earners would succeed regardless and merely self-select into earning an MBA. An employer is interested in sending the employee to obtain an MBA only if it is valuable.

Suppose, for the sake of argument, that the true causal graph is as presented in Figure~\ref{fig:pic3}(a): education $e$ (MBA degree) neither increases earnings $w$ nor is caused by them. Instead, the two are correlated due to two confounding variables, a person's ability $a$ and social skills $s$, that both affect education $e$ and earnings $w$.\footnote{For example, \citet{tamborini_education_2015} show using US panel data that while the effect of education is far from zero, ``accounting for key covariates reduces the estimated lifetime earnings return to college education by approximately 30\%'' (pp.1396--1397) and the return to graduate degrees by approximately 20\% (Fig. 2).} 
In addition, earnings $w$ are also affected by job experience/tenure $t$. Can the business school persuade the employee that education affects earnings? If so, can the employer then debunk this view and persuade the employee that no actual link exists between $e$ and $w$?

\begin{figure}
	\begin{center}
		
	\begin{minipage}[b]{0.45\textwidth}
		\centering 
		\includegraphics[scale=0.3]{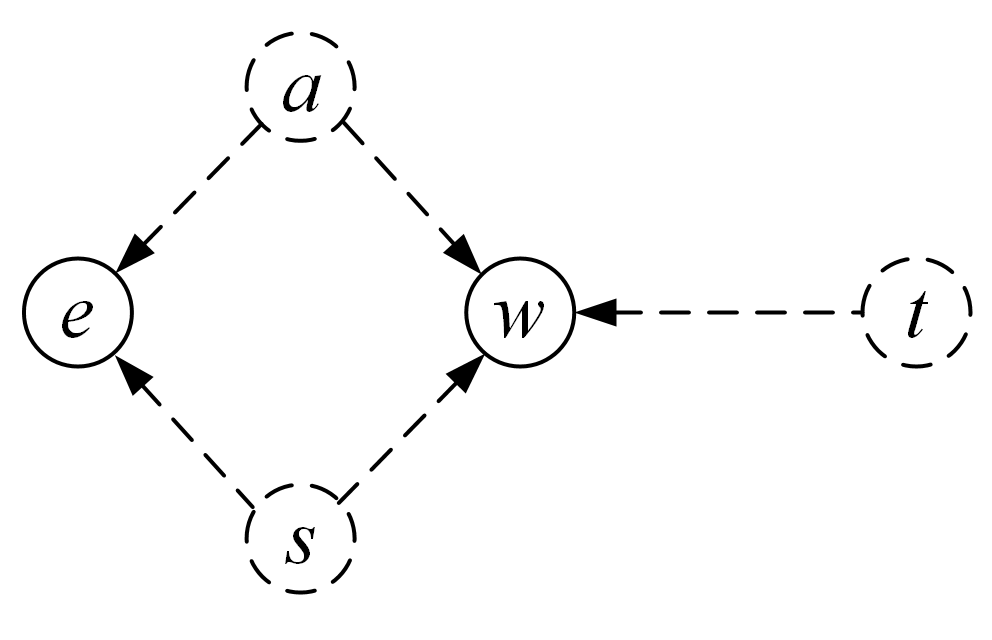}\\
		(a) true model\\
	\end{minipage}
	\begin{minipage}[b]{0.45\textwidth}
		\centering 
		\includegraphics[scale=0.3]{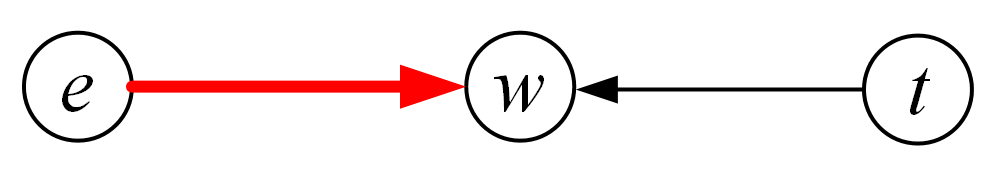}\\
		(b) subjective model\\
	\end{minipage}
	
	\caption{The true DAG and the employee's model after the promo campaign. \label{fig:pic3}}
	\end{center}
	\footnotesize
	\textsc{Notes:} Throughout the paper, the figures use dashed variable outlines and links to represent the items that are not (made) known to the receiver. Bold red arrows represent defective links.
\end{figure}

Suppose the business school launches a promotional campaign around variables $e$, $w$, and $t$ under the slogan ``Experience $+$ MBA $=$ More Pay!'' This campaign makes the employee aware of variable $t$ and how this variable is connected to $e$ and $w$ in the real world. Specifically, they see that $e$ is independent of $t$ unconditionally, but the two are correlated conditional on $w$ (conditional on high earnings $w$, a person is more likely to be experienced---high $t$---and educated---high $e$). The employee infers that this is only possible if $e$ and $t$ jointly determine $w$, as in Figure~\ref{fig:pic3}(b). Indeed, the conditional correlation $(e \not\perp t \mid w)$ suggests that there must be \emph{some} connection between $e$ and $t$. However, if $e \to w \to t$ or $e \leftarrow w \leftarrow t$, or $e \leftarrow w \to t$ (or $e \to t$, or $e \leftarrow t$), then $e$ and $t$ would also have been \emph{unconditionally} correlated, which is not the case in the data. The model in Figure~\ref{fig:pic3}(b) is the only causal model over $e,w,t$ that is consistent with the observed data on these variables. Therefore, after hearing the ad campaign, the employee has no choice but to conclude that the correlation is generated not by selection $e\leftarrow w$, but by a causal link $e\to w$. The promotional campaign successfully persuades the employee by disclosing an additional variable on top of two variables of interest. Notably, persuasion is effective despite the suggested narrative being untruthful.

Suppose the employer, in turn, knows that the degree does not actually increase the employee's value (and, therefore, earnings $w$). When the employee asks the employer to fund their education, the employer seeks to make this lack of causal connection clear---to debunk the model that the business school instilled. Suppose the employer first tells the employee that ability $a$ is a common factor that determines both educational choices $e$ and earnings $w$. The employee can see that $a$ is indeed correlated with both $e$ and $w$, and sees no contradiction in the data to the suggestion that $a$ is a common cause of $e$ and $w$. At the same time, the employee sees no contradiction to education $e$ increasing earnings $w$, since it is clear from the data that $e$ and $w$ are correlated even conditional on $a$. The employee's subjective model after this conversation is presented in Figure~\ref{fig:pic4}(b). Disclosing one common cause is insufficient for the employer to debunk the causal link $e \to w$ that the employee believes in.

\begin{figure}
	\centering
		
	\begin{minipage}[b]{0.45\textwidth}
		\centering 
		\includegraphics[scale=0.3]{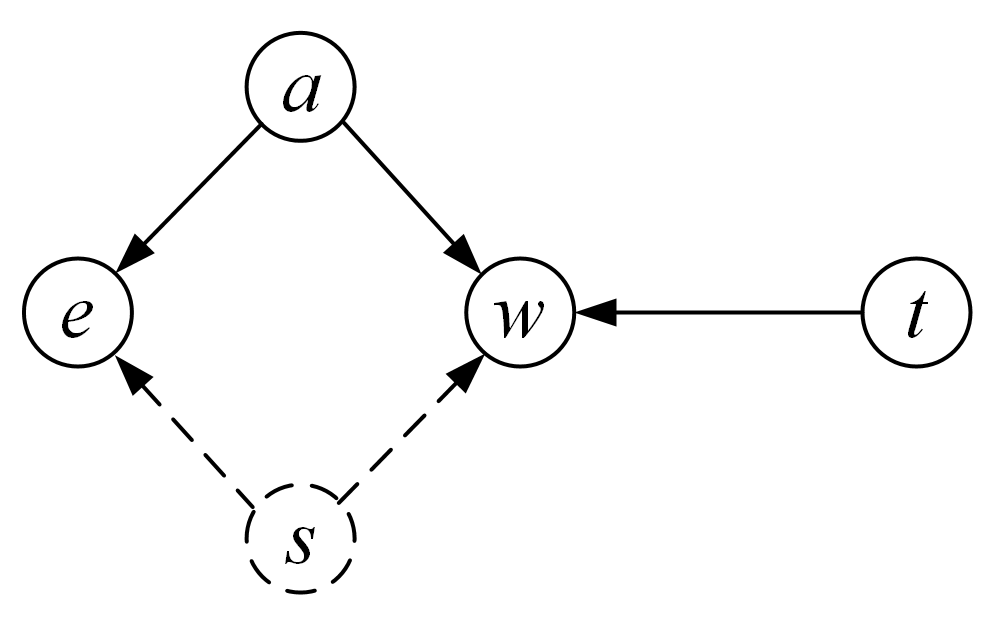}\\
		(a) true model\\
	\end{minipage}
	\begin{minipage}[b]{0.45\textwidth}
		\centering 
		\includegraphics[scale=0.3]{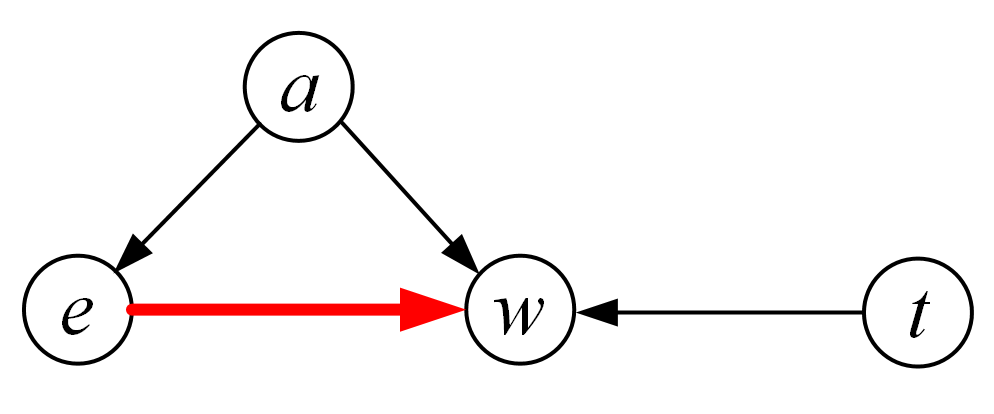}\\
		(b) subjective model\\
	\end{minipage}
	
	\caption{Employee's model after the first chat with the employer. \label{fig:pic4}}
\end{figure}

Indeed, to convince the employee that $e$ does not directly affect $w$, the employer would have to disclose all common factors influencing both $e$ and $w$, which in this example include $a$ and $s$, but in general there could be many more. So long as there is even a single confounding variable not known to the receiver, they see that $e$ and $w$ are correlated conditional on everything else they know, suggesting that $e$ and $w$ must be causally linked. While it was easy for the business school to persuade the employee that the link runs in the opposite direction---requiring disclosure of only one variable---it is far harder for the employer to convince the employee that no direct link exists, as this demands revealing \emph{all} relevant variables.

\section{Model} \label{sec:model}

\subsection{Setup}

\paragraph{The world.}
The world is described by a causal directed acyclic graph (DAG) $M_t=(\Omega_t,C_t)$ called \emph{the true model}, where $\Omega_t$ is a set of \emph{variables}, and $C_t:\Omega_t \to 2^{\Omega_t}$ describes directed links capturing the causal relations between them. For any variable $a \in \Omega_t$, relation $b \in C_t(a)$ is denoted equivalently as $b \to_t a$ and described as ``$b$ \emph{has a causal effect} on $a$'' and ``$b$ is a parent of $a$''. 

Each node $a \in \Omega_t$ represents a random variable distributed according to some p.d.f. $a \sim f_a \left( \cdotp | C_t(a) \right)$. We let $P$ denote the joint density of all variables, so for any vector of realizations of $\Omega_t$: $P(\Omega_t) \equiv \prod_{a \in \Omega_t} f_a \left( a | C_t(a) \right)$. 
For any subset of variables $\Omega \subset \Omega_t$, let $P|\Omega$ denote the marginal density of $\Omega$. In what follows, we refer to the full distribution $P$ and all marginals $P|\Omega$ as the \emph{data}. We also assume that if a player is aware of set $\Omega$ of variables, then they observe the data $P|\Omega$, meaning that they observe the whole true joint distribution of these variables.
This allows us to abstract from statistical issues and focus purely on causal discovery.
Given some $S \subset \Omega_t$ and $a,b \in \Omega_t \setminus S$, we use the notation $(a \perp b \mid S)$ if $a$ and $b$ are statistically independent conditional on $S$ in $P$, and $(a \not\perp b \mid S)$ if the converse is true. 

\paragraph{Subjective models and consistency.}
We define a \emph{(subjective) model} of the world as a DAG $M=(\Omega,C)$ with $\Omega \subseteq \Omega_t$. Given a model $M=(\Omega,C)$, $b \in C(a)$ for some $a,b \in \Omega$ denotes the same parent relation as in the true model. Further, let $\bar{C}(a)$ denote the set of predecessors of any $a \in \Omega$: $b \in \bar{C}(a)$---equivalently denoted as $b \Rightarrow a$ and labeled as $a$ being a \emph{descendant} of $b$---is true if and only if there exists a path $b \to ... \to a$ in $C$. We say that variables $a,b \in \Omega$ are \emph{adjacent} if $a \to b$ or $a \leftarrow b$. We say that $a,b$ are \emph{correlated} if either $a \Rightarrow b$ or $a \Leftarrow b$, or there exists a common ancestor $c \in \bar{C}(a) \cap \bar{C}(b)$.
We call a triplet of variables $a-b-c$ (where $b$ is adjacent to $a$ and $c$, but $a$ and $c$ are not adjacent) a \emph{collider} if $a \to b \leftarrow c$, a \emph{chain} if $a \to b \to c$ or $a \leftarrow b \leftarrow c$, and a \emph{fork} if $a \leftarrow b \to c$. Thus, throughout the main text, ``collider'' refers to what the graphical-models literature calls an unshielded collider.

We now introduce two important definitions that bridge the topology of a model and the data it produces. They specify which graphical properties must be reflected in the data. The causal discovery literature focuses on independencies: two variables must be statistically correlated if and only if they are correlated in the model in the sense defined above. The following definition expands this notion to also apply to \emph{conditionally} (in)dependent variables.

\begin{definition}[d-separation]\label{def:dsep}
	Given a model $M=(\Omega,C)$, let $a,b\in\Omega$ and $S\subseteq\Omega\setminus\{a,b\}$. The set $S$ \emph{d-separates} $a$ and $b$ in $M$ if every path between $a$ and $b$ contains a non-endpoint variable $d$ satisfying at least one of the following:
	\begin{enumerate}[nosep]
		\item the two links of the path adjacent to $d$ both point toward $d$, and neither $d$ nor any descendant of $d$ belongs to $S$ ($d \notin S$ and $\forall e \in S: d \not\Rightarrow e$); or
		\item the two links of the path adjacent to $d$ do not both point toward $d$, and $d\in S$.
	\end{enumerate}
	When the model is not specified, d-separation refers to the true model $M_t$.
\end{definition}

Intuitively, d-separating  $a$ and $b$ means that set $S$ controls for all sources of correlation between them. This means blocking (i.e., having an element on) every directed path $a \Rightarrow b$, as well as all paths that pass through confounders (common correlation sources), $a \Leftarrow c \Rightarrow b$. The collider part of the definition is trickier: the converging-arrow path $a \to c \leftarrow b$ does not produce any correlation between $a$ and $b$ by itself, but $a,b$ will be correlated conditional on $c$. Hence, $S$ must \emph{not} include $c$ or any of its descendants.\footnote{Definition~\ref{def:dsep2} in Appendix \ref{sec:proofs} presents an alternative equivalent characterization of d-separation, which may be slightly easier to parse.}
We now relate d-separation in the model to statistical independence in the data.

\begin{definition}[Model consistency]
	We say that model $M=(\Omega,C)$ is \emph{consistent with data} $P|\Omega$ if the following two conditions hold for any $a,b \in \Omega$ and $S \subseteq \Omega \setminus \{a,b\}$:
	\begin{description}
		\item[(Markov property)] if $S$ d-separates $a,b$ in $M$, then $\left( a \perp b \mid S \right)$;\footnote{Alternatively, the Markov property can be written as a distribution factorization statement, requiring that there exists a collection of probability density functions $\left\{ \hat{f}_a \left( \cdotp | C(a) \right) \right\}_{a \in \Omega}$ such that $P(\Omega) = \prod_{a \in \Omega} \hat{f}_a \left( a | C(a) \right)$.}
	
		\item[(Faithfulness)] if $\left( a \perp b \mid S \right)$, then $S$ d-separates $a$ and $b$ in $M$.
	\end{description}
\end{definition}

The two conditions require that the model and the data agree on which variables are conditionally independent. The Markov property requires that the independencies displayed by the model's topology must hold in the data \citep[Chapter 3]{pearl_probabilistic_1988}. Faithfulness requires the converse: every conditional independence in the data must be reflected in the model's topology, so that no data independence arises ``by accident''---e.g., from two causal channels of opposing signs that cancel out exactly. 

We let $\mathcal{M}(\Omega_s)$ denote the set of models consistent with $P|\Omega_s$.
We assume the true model is consistent with $P$: $M_t \in \mathcal{M}(\Omega_t)$. While the Markov property holds by definition of $P$, faithfulness is nontrivial: it requires that the data contain no conditional independencies beyond those implied by the graph.


\paragraph{Players.}
We consider a setting with two players: a sender and a receiver. To focus on the substantive elements, we present a reduced-form interaction here. A more standard game-theoretic model that serves as a microfoundation for this setup is presented in Section \ref{sec:microfound}.

The receiver initially has no model in mind and is not aware of any variables.
The sender knows the true model $M_t$.
The sender can \emph{propose} a model $M_s=(\Omega_s,C_s)$ with $\Omega_s \subseteq \Omega_t$ to the receiver, which the receiver can then \emph{accept or reject}. In words, the sender can reveal an arbitrary subset of variables and a subjective model that describes the causal connections between them.
The sender can also remain silent, denoted by $M_s = \varnothing$.

We assume that there are two variables of interest for the sender, $x,y \in \Omega_t$, and the sender wants to persuade the receiver that $x \Rightarrow_s y$.\footnote{Throughout the paper, we use subscripts to indicate which model the arrow notation refers to: e.g., ``$a \to_t b$'' means $a \in C_t(b)$ in $M_t$, and ``$a \Rightarrow_s b$'' means $a \in \bar{C}_s(b)$ in $M_s$.} 
We consider two scenarios under which the receiver accepts the proposed model:
\begin{itemize}[noitemsep]
	\item a \emph{na{\"i}ve} receiver accepts model $M_s$ if and only if it is consistent with $P|\Omega_s$, and
	\item a \emph{sophisticated} receiver accepts model $M_s$ with $x \Rightarrow_s y$ if and only if $\mathcal{M}(\Omega_s)$ is nonempty and $x \Rightarrow_M y$ in every $M \in \mathcal{M}(\Omega_s)$.\footnote{For sake of completeness, we can also say that a sophisticated receiver accepts model $M_s$ with $x \Leftarrow_s y$ if and only if $x \Leftarrow_M y$ in every $M \in \mathcal{M}(\Omega_s) \neq \emptyset$, and accepts $M_s$ with $x \not\Rightarrow_s y$ and $x \not\Leftarrow_s y$ if and only if $x \not\Rightarrow_s y$ and $x \not\Leftarrow_s y$ in every $M \in \mathcal{M}(\Omega_s) \neq \emptyset$.}
\end{itemize}
A na{\"i}ve receiver is someone who is relatively uncritical: whenever they are presented with a correlation and a story that does not contradict it, they accept such a story---though they can still identify and reject inconsistent models. A sophisticated receiver, on the other hand, is wary of reverse causality; they accept the sender's causal narrative regarding $x$ and $y$ only if the data explicitly rule out all other possibilities.\footnote{Experimental evidence supports our modeling choices. \citet{kendall_complexity_2024} show that subjects struggle to form and/or articulate their mental model of the world. Hence even if people can see the data, they may not automatically have a subjective model that explains it. At the same time, \citet{ambuehl2025} show that subjects are remarkably good at identifying and discarding models that are \emph{not} consistent with the data, so it is natural to assume that even na{\"i}ve receivers are capable of this.} 
We assume that the receiver's type is commonly known.
Regardless of whether the receiver accepts the model, the game ends.

In addition to the goal of persuading the receiver that $x \Rightarrow_s y$, we assume that the sender's secondary objective is to do this using the smallest number of disclosed variables, i.e., to minimize $|\Omega_s|$. Specifically, one can think that the sender's payoff function is 
\begin{align} \label{eq:sender_util}
	U_S(M_s) = v(M_s) - \kappa\left( |\Omega_s| \right),
\end{align}
where $v(M_s)$ is the benefit of persuading the receiver and is such that $v(M_s) = 1$ if the receiver accepts a model with $x \Rightarrow_s y$ and $v(M_s)=0$ otherwise. The second term, $\kappa\left( |\Omega_s| \right)$, is a strictly increasing function representing the costs of communicating larger models. We normalize $\kappa(0)=0$, so that silence yields $U_S(\emptyset)=0$. This term may capture a direct cost of discovery and procuring evidence, a reduced-form cost of communicating complex models through a lower success rate among receivers---who may find such models confusing and unconvincing---or an instrumental cost to the sender of revealing evidence that they would otherwise prefer to keep hidden.\footnote{The cost $\kappa$ does no substantive work in the analysis beyond selecting among persuasive proposals and silence.} 
For the remainder of the paper, we assume that the benefit of persuasion is always greater than the cost when persuasion is possible.

\medskip 

This concludes the setup of the baseline model; the remainder of Section \ref{sec:model} introduces various supplementary definitions and helpful results.
Before we proceed, however, one point of discussion is in order. The setup above does not describe a ``game'' in the strict game-theoretic sense and is best seen as the sender's decision problem. This is because the receiver is not a strategic payoff-maximizing player, but rather follows a specific choice procedure. The sender is also nonstandard in that they benefit directly from changing the receiver's mind, as opposed to deriving payoffs from the receiver's action. This choice of setup is intentional: given the novelty of the causal discovery machinery to most economists, we strip away all the details that are not strictly necessary for the persuasion aspect, to focus on what drives the receiver to accept or reject a model. For interested readers, Section \ref{sec:microfound} completes the setup and describes a fully specified game, in which the receiver makes some decision and the sender attempts to influence this decision by communicating additional information. The results there establish the equivalence between the two settings.

\subsection{Consistency} \label{sec:consistency}

In this section, we discuss in more detail what consistency means and how to check whether a given model is consistent or not with a given distribution. There are two observations that drive the analysis. First, the Markov property allows one to easily identify adjacent variables: if $\left( a \not\perp b \mid S \right)$ for all $S\subseteq\Omega\setminus\{a,b\}$, then $a$ and $b$ are adjacent. Second, given a triplet of variables $a-b-c$, a collider $a \to_t b \leftarrow_t c$ looks different from chains $a \to_t b \to_t c$ and $a \leftarrow_t b \leftarrow_t c$ or a fork $a \leftarrow_t b \to_t c$.
A collider in the true model generates data such that $a \perp c$ and $(a \not\perp c \mid b)$, while the other three options generate the opposite pattern: $a \not \perp c$ and $(a \perp c \mid b)$. 
We refer to collider patterns in the data as V-structures, defined below.

\begin{definition}[V-structure]
	Variables $a,b,c \in \Omega$ constitute a \emph{V-structure} in the data $P|\Omega$ if there exists $S\subseteq\Omega\setminus\{a,b,c\}$ such that $\left( a \perp c \mid S \right)$ and $\left( a \not\perp c \mid S \cup \{b\} \right)$. We further label $b$ as the \emph{V-center}, $a$ and $c$ as the \emph{V-parents}, and variables in $S$ as \emph{controls}.
\end{definition}
\begin{definition}[Direct V-structure]
	Variables $a,b,c \in \Omega$ constitute a \emph{direct V-structure} in the data $P|\Omega$ if they constitute a V-structure and there is no $S\subseteq\Omega\setminus\{a,b,c\}$ such that $\left( a \perp b \mid S \right)$ or $\left( b \perp c \mid S \right)$.
\end{definition}

\citet{verma_equivalence_1990} show that any two models are observationally equivalent (are consistent with the same data) if they have the same sets of adjacencies and colliders. Conversely, a model that has different adjacencies or different colliders relative to the true model would contradict the data.
They show that as a consequence, the set of models $\mathcal{M}(\Omega)$ consistent with data $P|\Omega$ (if it is nonempty) can be obtained by following the Inductive Causation ($IC$) algorithm described below, which relies on first identifying adjacencies---the skeleton of the model---then identifying colliders $a \rightarrow b \leftarrow c$ from direct V-structures, and finally directing further links using V-structures and acyclicity. We follow the presentation of this algorithm in \citet{pearl_causality_2009}.

\begin{definition}[IC algorithm]
	Consider a set $\Omega$ of variables and data $P|\Omega$. Construct a partially-directed graph $K(\Omega)$ as follows.
	\begin{enumerate}[nosep]
		\item For any $a,b \in \Omega$, link $a$ and $b$ if there is no set $S\subseteq\Omega\setminus\{a,b\}$ such that $\left( a \perp b \mid S \right)$. 
		
		\item For any $a,b,c \in \Omega$ that constitute a direct V-structure, direct the links toward $b$. 
		
		\item For any $a,b \in \Omega$ connected by an undirected link, direct the link from $a$ to $b$ if:
		\begin{enumerate}[nosep]
			\item link $b \to a$ creates a V-structure that would have been identified in Step 2, or
			\item link $b \to a$ creates a cycle.
		\end{enumerate}
	\end{enumerate}
\end{definition}

The IC algorithm outputs a completed partially-directed acyclic graph (CPDAG) $K(\Omega)$ that represents an equivalence class of observationally-equivalent models, i.e., the set $\mathcal{M}(\Omega)$. If $\mathcal{M}(\Omega)$ is nonempty, then any directed acyclic selection $C' \subseteq K(\Omega)$ that contains the same adjacencies and no additional V-structures relative to $K(\Omega)$ produces a consistent model, i.e., any such $(\Omega,C')$ is consistent with $P|\Omega$: $(\Omega,C') \in \mathcal{M}(\Omega)$. Conversely, no model that contradicts $K(\Omega)$ can be consistent. 
We say that any link $a \to b$ that is oriented in $K(\Omega)$ is \emph{compelled} by the data $P|\Omega$, meaning no model $M \in \mathcal{M}(\Omega)$ consistent with the data can have $a \leftarrow b$.

\citet{meek_causal_1995} shows that Step 3 of the IC algorithm is equivalent to the iterative application of four simple orientation rules depicted in Figure~\ref{fig:Meek} in Supplementary Appendix \ref{sec:examples}.
An example application of the IC algorithm is demonstrated in Example~\ref{exp:IC} in Supplementary Appendix \ref{sec:examples}.
If the set of consistent models $\mathcal{M}(\Omega)$ is empty for some $\Omega \subset \Omega_t$, then the IC algorithm will produce an inconsistent model, as shown in Example~\ref{exp:IC_incons} in Supplementary Appendix \ref{sec:examples}.

\subsection{Defective links and models}

As stated above, the IC algorithm does not always \emph{uniquely} identify the true model, even when all variables are observed. While the skeleton of the model can be pinned down, not all links can be oriented from the data alone. Hence, the receiver may be persuaded to accept a model that is consistent with the data yet incorrect. 
We call such models defective and offer examples below.

\begin{definition}[Defective links and models]
	Given a model $M_s$, we call $a \to_s b$ for some $a,b \in \Omega_s$ a \emph{defective link} if $a \not\Rightarrow_t b$. The model $M_s$ is \emph{defective} if it has a defective link.
\end{definition}

\begin{example}[Defective consistent link]
	Suppose the true model is $a \to b$. Running the IC algorithm on $\Omega = \{a,b\}$ yields in Step 1 that $a - b$, but since there are no V-structures in the data, the algorithm stops after Step 1 and leaves the link undirected. Hence, $a \leftarrow b$ is consistent with the data, but such a link is defective.
\end{example}
\begin{example}[Defective consistent models]
	Suppose the true model is $a \to_t b \to_t c$. The IC algorithm identifies in Step 1 that $a - b - c$, but since there are no V-structures in the data, the algorithm stops after Step 1 and leaves the links undirected. Then models $a \leftarrow b \leftarrow c$ and $a \leftarrow b \to c$ are consistent yet defective. In turn, model $a\to b \leftarrow c$ is inconsistent with the data, since it contains a collider but the data have no corresponding V-structure.
\end{example}

\subsection{Simple and rich worlds}

For some of the results, we impose additional restrictions on the true model to ensure it is sufficiently ``nice''. These conditions are stated below. 

\begin{definition}[Simple world]
	The true model $M_t$ is \emph{simple} if for any $a,b,c \in \Omega_t$ that form a collider $a \to_t b \leftarrow_t c$, variables $a$ and $c$ are not correlated.
\end{definition}

\begin{definition}[Rich world]
	The true model $M_t$ is \emph{rich} if $\mathcal{M}(\Omega_t) = \{M_t\}$.
\end{definition}

Simplicity requires that all direct V-structures are ``obvious'' in the data, in the sense that they can be identified without any controls: for any $a,b,c \in \Omega_t$ such that $\left(a\perp c \mid S\right)$ and $\left(a \not\perp c \mid S, b \right)$, these two observations must hold for $S=\emptyset$ in particular. 
We assume that even if the true model is simple, the receiver is unaware of this fact, so we allow the sender's proposed model $M_s$ to be non-simple.

\begin{figure}
	\centering
	
	\begin{minipage}[b]{0.32\textwidth}
		\centering 
		\includegraphics[scale=0.3]{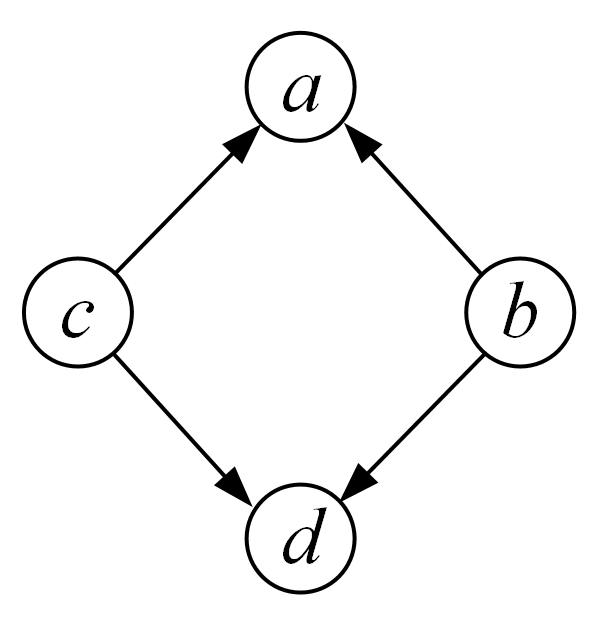}\\
		(a) \\
	\end{minipage}
	\begin{minipage}[b]{0.32\textwidth}
		\centering 
		\includegraphics[scale=0.3]{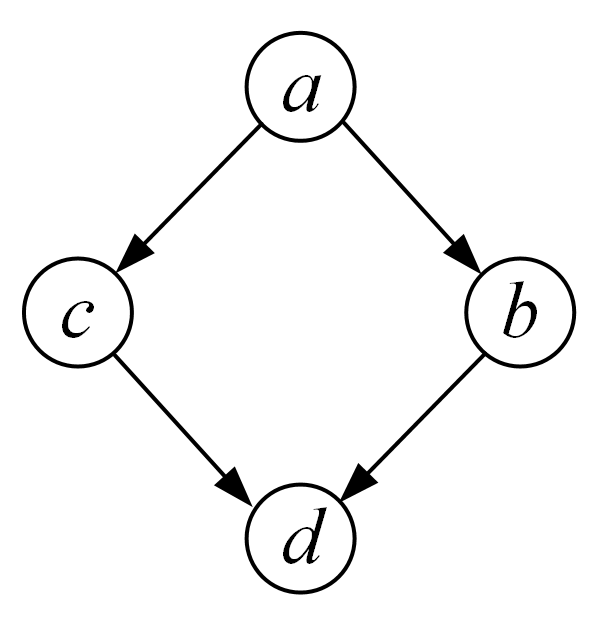}\\
		(b) \\
	\end{minipage}
	\begin{minipage}[b]{0.32\textwidth}
		\centering 
		\includegraphics[scale=0.3]{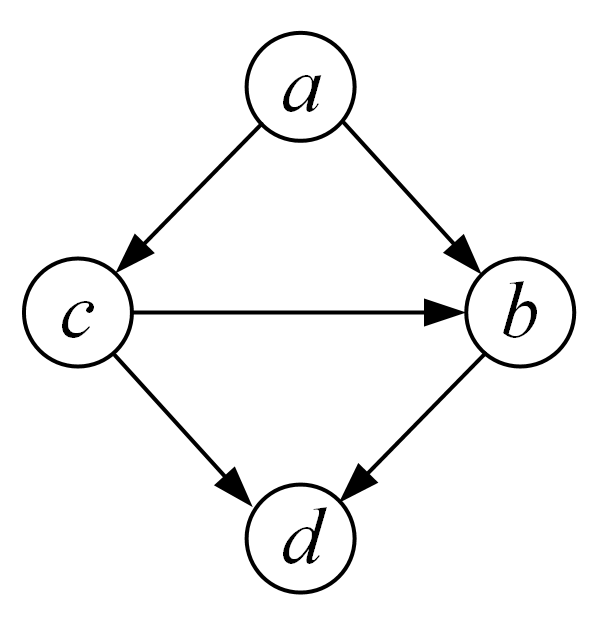}\\
		(c) \\
	\end{minipage}
	
	\caption{Simple (a and c) and non-simple (b) models. \label{fig:rhombi}}
\end{figure}
\begin{example}[Simple models]
	The models in Figures~\ref{fig:rhombi}(a) and (c) are simple, but the model in (b) is not. In panel (a), $c \to a \leftarrow b$ and $c \to d \leftarrow b$ create V-structures, but because $b \perp c$, they can be identified in the data without any controls. In (b), $c \to d \leftarrow b$ creates a V-structure, but identifying it requires controlling for $a$, since $c \not \perp b$ and $(c\perp b \mid a)$. In (c), $c \to d \leftarrow b$ does not create a V-structure because of the link $c \to b$, which leads to $c \not \perp b$ and $(c \not\perp b \mid a)$. Similarly, the true model from the illustrative example in Figure~\ref{fig:pic3}(a) is simple.
\end{example}

In turn, richness requires that every link is compelled by the full distribution $P$, so the true model can be uniquely identified using the IC algorithm. This means that every directed link must either belong to a V-structure, or be orientable by tracing back to a V-structure. The models in Figures~\ref{fig:pic3}(a) and \ref{fig:rhombi}(a) are rich, whereas the models in Figures~\ref{fig:rhombi}(b) and (c) are not.

\section{Causal persuasion} \label{sec:persuade_blank}

\subsection{The power of omitted variables}

The sender's persuasive power in our model comes from two factors: the ability to omit variables and the ability to propose a particular model that ties the revealed variables.
We start by exploring the former. Does variable omission create the potential for causal misperceptions? That is, which causal connections \emph{can} the receiver get wrong due to not observing all of $\Omega_t$ (omitted variable bias)? These questions are partially answered by the following lemma, which follows from the results of \citet{verma_equivalence_1990} and imposes a lower bound on the alignment between the true model and any consistent subjective model.

\begin{lemma} \label{lem:consistent_models}
	If $M_s$ is consistent with $P|\Omega_s$, then the following hold for any $a,b \in \Omega_s$:
	\begin{enumerate}[nosep]
		\item $a$ and $b$ are correlated in $M_s$ if and only if they are correlated in $M_t$;
		
		\item if $a$ and $b$ are adjacent in $M_t$, then they are adjacent in $M_s$.
	\end{enumerate}
\end{lemma}
\begin{proof}
	The first statement follows from the observation that $a,b$ are correlated in a consistent model if and only if $a \not\perp b$ in the data, and that the true model is consistent.
	
	The second statement: since $M_s$ is consistent, it must align with the output of the IC algorithm. Suppose, toward a contradiction, that $a,b$ are not adjacent in $M_s$. From Step 1 of the IC algorithm, there then exists $S \subseteq \Omega_s\setminus\{a,b\}$ such that $(a \perp b \mid S)$. Faithfulness then implies that $a$ and $b$ cannot be adjacent in the true model.
\end{proof}

In words, Lemma~\ref{lem:consistent_models} above says that if two variables are correlated in the true model (i.e., they are connected by a path of chains and forks but no colliders), then they must be correlated as well in any subjective model that only includes a subset of variables and is consistent with the data on these variables. The converse also holds: if $x$ and $y$ are \emph{not} correlated in the true model---implying that $x \perp y$ in the data---then they cannot be correlated in any subjective model consistent with the data. 
Therefore, the sender cannot persuade the receiver that $x \to_s y$ in this case, regardless of whether the receiver is sophisticated or na{\"i}ve.
It also follows immediately from the lemma that revealing random noise or other variables unrelated to $x$ and $y$ neither helps nor harms the sender in their quest to persuade, since such variables become isolated nodes in any consistent model $M$.

\begin{example}[Identifying adjacencies]
	Suppose the true model is $a \to_t b \leftarrow_t c$. If the sender discloses all three variables, they constitute a V-structure, hence the true model is uniquely identified by the IC algorithm. If the sender only discloses $\Omega_s = \{a, c\}$, then the only model consistent with the data are such that $a$ and $c$ are not adjacent, since $a \perp c$ in the data.
\end{example}

One possible misperception not ruled out by the lemma is that a consistent subjective model can draw a link that is nonexistent in the true model, due to omitting other relevant variables. This means that if $x$ and $y$ are correlated in the true model, then Lemma~\ref{lem:consistent_models} does not preclude the existence of a consistent subjective model $M_s$ that includes $x \to_s y$ even if $x \not\Rightarrow_t y$ (making link $x \to_s y$ defective). This is explored in more detail in later sections.

\begin{example}[Non-defective model drawing a nonexistent link]
	Suppose the true model is $a \to_t b \to_t c$, and the sender only discloses $\Omega_s = \{a,c\}$. Then both $a \to_s c$ and $a \leftarrow_s c$ are consistent with the data, since $a \not\perp c$. The latter model, $a \leftarrow_s c$, is defective. The former model, $a \to_s c$, draws a nonexistent link from $a$ to $c$, but is not defective according to our definition, since $a \Rightarrow_t c$, so the implication that $a$ affects $c$ is correct.
\end{example}

\subsection{Persuading a na{\"i}ve receiver}

We proceed by characterizing the benchmark case of persuading a na{\"i}ve receiver. Lemma~\ref{lem:consistent_models} above implies that if $x$ and $y$ are not correlated in the true model (or the data), then no consistent model can include the link $x \to_s y$. Persuasion is therefore impossible in this case. If, however, $x$ and $y$ \emph{are} correlated in the true model, then persuading a na{\"i}ve receiver is easy. This is captured by the following proposition, which is trivial and left without proof.

\begin{proposition} \label{prop:pers_naive}
	If $x,y \in \Omega_t$ are correlated in the true model, then a na{\"i}ve receiver is persuaded by $M_s$ such that $\Omega_s = \{x,y\}$, $x \to_s y$.
	Otherwise a na{\"i}ve receiver cannot be persuaded that $x \to_s y$.
\end{proposition}

A na{\"i}ve receiver is someone who cannot tell correlation from causation. Hence, presenting a pair of correlated variables and claiming that one of them affects the other is sufficient for the sender to persuade such a receiver. No additional variables are required, as is illustrated by the example below and Example~\ref{exp:vaccines} in Supplementary Appendix \ref{sec:examples}. As we see shortly, this is not the case for a sophisticated receiver.

\begin{example}[Labor unions] \label{exp:union1}
	Throughout history, labor unions have attempted to recruit workers, promising higher wages and better working conditions \citep[e.g., ][]{white_collar_unionism_1970,seiu_unions_2018}.
	These promises have occasionally been supported by data suggesting that union membership is correlated with higher wage levels \citep[e.g., ][]{lewis_chapter_1986}. Such a narrative together with the data would suffice to persuade na{\"ive} workers.
\end{example}

\subsection{Persuading a sophisticated receiver}\label{sec:persuade_soph}

We next ask a more interesting question: when can the sender persuade a \emph{sophisticated} receiver that $x \to_s y$, and how? Again, Lemma~\ref{lem:consistent_models} implies that a necessary condition for persuasion to be possible is that $x$ and $y$ are correlated in the true model. However, this is no longer sufficient. In Example~\ref{exp:union1} above, a sophisticated receiver would be left unconvinced that union membership would increase their wage, since this causal link is not \emph{compelled} by the data. They would recognize that correlation does not imply causation and would therefore consider the possibility of self-selection---that high earners may be more likely to join unions.

Our first result for sophisticated receivers is negative in that it identifies a further case (in addition to the one implied by Lemma~\ref{lem:consistent_models}) in which persuasion is impossible.

\begin{theorem}[] \label{thm:no_persuade_defective}
	If $x \Leftarrow_t y$ then a sophisticated receiver cannot be persuaded that $x \Rightarrow_s y$.
\end{theorem}
\begin{proof}
	See Appendix \ref{sec:proofs}.
\end{proof}

Theorem~\ref{thm:no_persuade_defective} argues that the sender can never persuade a sophisticated receiver of reverse causation, i.e., that the causal relation runs in the direction opposite to the truth.
\begin{example}[Labor unions, continued] \label{exp:union2}
	Recent research suggests that the truth may be the exact opposite of the story the labor unions are telling: rather than unionization increasing wages, it is instead the case that unions are more likely to be created in high-wage industries and environments \citep{dinardo_economic_2004}. 
	Theorem~\ref{thm:no_persuade_defective} then implies that if this is indeed the truth, labor unions would not be able to persuade ``sophisticated'' workers to join.
\end{example}

The proof of Theorem~\ref{thm:no_persuade_defective} shows that for every feasible disclosure $\Omega_s$, the set of consistent models contains a DAG that preserves every directed link in the latent projection of the true model. This follows by combining standard properties of latent projections with the loyal-equivalent-graph result of \citet{zhang2005transformational}. Consequently, if $x\Leftarrow_t y$, then regardless of which variables the sender reveals or withholds, there is always a consistent model with $x\Leftarrow y$. The sender can therefore never compel the opposing relation $x\Rightarrow_s y$ by selective disclosure.\footnote{\citet{spirtes_causal_1995} study which causal relations can be reliably inferred when the set of latent variables is given. In our setting, the sender strategically chooses which variables remain latent. The contribution of Theorem~\ref{thm:no_persuade_defective} is to show that this strategic omission cannot compel a causal direction opposing the truth.}

\medskip 
Moving on from Theorem~\ref{thm:no_persuade_defective}, it is reasonable to ask whether it is possible to persuade a sophisticated receiver that $x \to_s y$ if $x \not\Leftarrow_t y$. 
As we show in what follows, under some additional conditions, persuasion is not only possible in this case, but can be done by disclosing at most two variables in addition to $x,y$. If such variables exist, we say that persuasion is ``easy''.
We show later that if these conditions do not hold, then even if persuasion is possible, it may be ``difficult'' in that it may require disclosing arbitrarily many variables.

\begin{definition}[Obvious cause]
	Given $x,y \in \Omega_t$, we call $z \in \Omega_t$ an \emph{obvious cause of $y$ given $x$} if $z \in \bar{C}_t(y)$ and $z$ is not correlated with $x$.
\end{definition}
\begin{definition}[Non-obvious cause]
	Given $x,y \in \Omega_t$ such that $x \Rightarrow_t y$, we call $w \in \Omega_t$ a \emph{non-obvious cause of $y$ given $x$} if $w \in \bar{C}_t(x)$ and $x$ d-separates $w$ and $y$. 
\end{definition}
\begin{definition}[Confounder]
	Given $x,y \in \Omega_t$, we call $c \in \Omega_t$ a \emph{confounding variable} (or a \emph{confounder}) for $(x,y)$ if $C_t$ contains a path $c \Rightarrow_t x$ that does not pass through $y$ and a path $c \Rightarrow_t y$ that does not pass through $x$. We denote the set of such confounders by $\tilde{C}_t(x,y) \subset \Omega_t$.
\end{definition}

\begin{proposition} \label{prop:persuade_sophist}
	Suppose $x,y \in \Omega_t$ are correlated and $x \not\Leftarrow_t y$. 
	\begin{enumerate}
		\item If there exists an obvious cause $z$ of $y$ given $x$, then a sophisticated receiver can be persuaded that $x \to_s y$ by a model $M_s$ with $\Omega_s=\{x,y,z\}$.
		\item If $x \Rightarrow_t y$ and there exist two non-obvious causes $v,w$ of $y$ given $x$ such that $v \perp w$, then a sophisticated receiver can be persuaded that $x \to_s y$ by a model $M_s$ with $\Omega_s=\{v,w,x,y\}$.
	\end{enumerate}
\end{proposition}
\begin{proof}
	\textbf{Case 1.}
	Revealing variables $\Omega_s = \{x,y,z\}$ when $z$ is an obvious cause creates a V-structure $x \to y \leftarrow z$, hence $x \to_s y$ is compelled by $P|\Omega_s$.
	
	\textbf{Case 2.}
	Revealing variables $\Omega_s = \{v,w,x,y\}$ when $v$ and $w$ are non-obvious causes creates a V-structure $v \to x \leftarrow w$ since $v \perp w$, and $x \to_s y$ is then compelled by $P|\Omega_s$ by $(v,x,y)$ \emph{not} being a V-structure since $(v \perp y \mid x)$ (which holds because by definition of $v$, $x$ d-separates $v$ and $y$).
\end{proof}

\begin{figure}
	\centering
		
	\begin{minipage}[b]{0.45\textwidth}
		\centering 
		\includegraphics[scale=0.3]{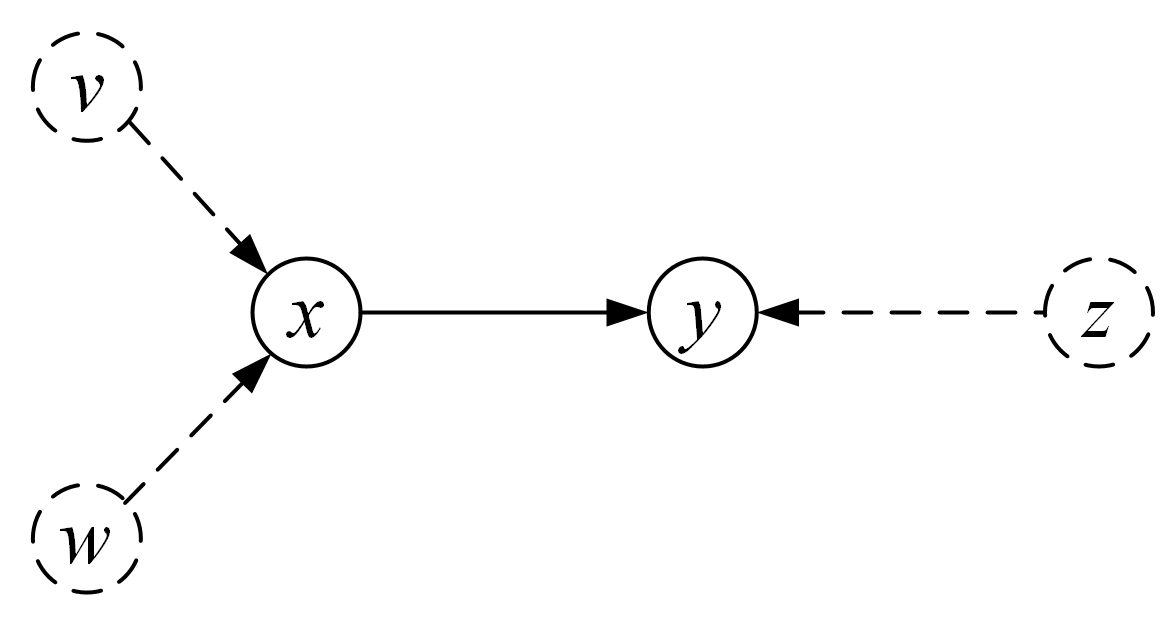}\\
		(a) true link $x\to_t y$\\
	\end{minipage}
	\begin{minipage}[b]{0.45\textwidth}
		\centering 
		\includegraphics[scale=0.3]{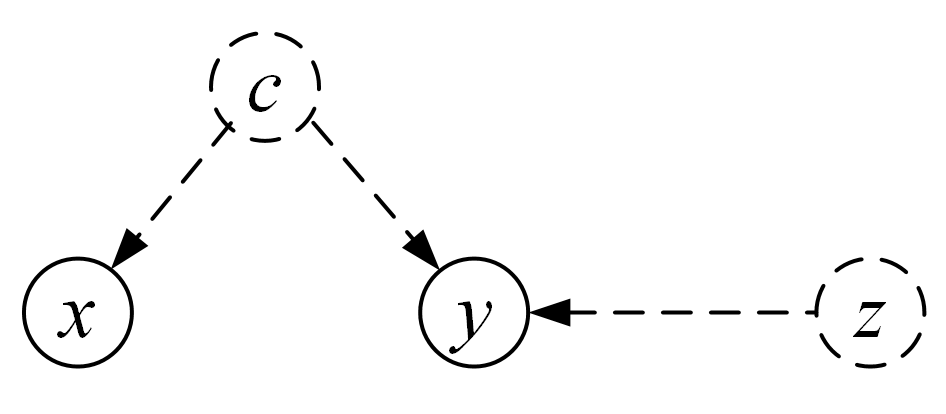}\\
		(b) unobservable confounder\\
	\end{minipage}
	
	\caption{$M_t$ with obvious and non-obvious causes. \label{fig:causes_copy}}
\end{figure}

Obvious and non-obvious causes describe variables that suffice to persuade a sophisticated receiver. If an obvious cause exists, revealing it forms a V-structure together with the two variables of interest, which would be identifiable from the data. In the second case, $v\to x\leftarrow w$ orients the links into $x$. The absence of a V-structure at $(v,x,y)$ then compels $x\to y$.
 Figure~\ref{fig:causes_copy} illustrates the two cases.

Both obvious and non-obvious causes have close relatives in econometric practice, albeit no exact analogs. 
An obvious cause $z$ known to affect $y$ yet independent of $x$ is the familiar shape of an excluded demand shifter in a simultaneous (demand \& supply) system, and verifying $z \perp x$ is a balance test. A non-obvious cause $v$ is qualitatively similar to an instrumental variable, affecting $y$ only through $x$ (i.e., $x$ d-separates $v$ and $y$).
Yet, the goals and the assumptions are different. In econometrics, the causal graph is assumed---both the skeleton and the link directions---and the problem is to estimate the size of the causal effect (causal inference). We consider instead the problem of causal discovery: the object is not the size of an effect but the existence and direction of the link itself. In particular, there is no expectation that the receiver shares the same causal model from the outset; instead, the goal is to prove to the receiver that the targeted causal relation is compelled by the disclosed data.\footnote{The causal-discovery literature contains results that may resemble Proposition~\ref{prop:persuade_sophist}. For example, \citet{mani_theoretical_2012} use the term ``Y-structures'' to describe causal constructs with two non-obvious causes and describe a score-based algorithm to discover them in the data. In a similar vein, \citet{cooper_simple_1997} described an algorithm that used non-obvious causes to mine datasets for pairs of $x,y$ s.t. $x \to y$. The difference between these papers and our work is conceptual: \citet{cooper_simple_1997} and \citet{mani_theoretical_2012} mine the data to discover \emph{any} causal relations, whereas we are seeking ways to prove to a receiver that $x \to_s y$ for specific $x$ and $y$. While the answers may look similar, the questions are qualitatively different.}

\medskip 

What if the premise of Proposition~\ref{prop:persuade_sophist} does not hold? Is it easy to persuade a sophisticated receiver if there are no suitable obvious or non-obvious causes? Example~\ref{exp:nocauses} in Supplementary Appendix \ref{sec:examples} demonstrates that in this case, persuasion may be easy, difficult (in the sense of requiring the sender to reveal arbitrarily many variables in addition to $x,y$), or impossible. 
This reflects the core challenge faced by a large portion of empirical economics, which seeks to establish causal relationships between variables of interest. The sheer size of the literature underscores the difficulty of the task, as these variables are typically jointly determined by a large number of confounding factors.\footnote{A canonical example is the estimation of the returns to education, where ability, family background, and other unobserved characteristics confound the relationship between schooling and earnings; see, for instance, \cite{card99}.} 
Therefore, credible identification (and persuasion) requires conditioning on a rich set of covariates or exploiting exogenous variation. Consequently, persuading a sophisticated receiver in such a setting would require the sender to provide information at the level of an advanced empirical economist.

We conclude the section with a structural result on optimal disclosure in simple worlds. The theorem below claims that if the sender speaks up, then $x$ and $y$ must be adjacent in their proposed model. (Recall that the sender may also remain silent when persuasion is impossible or too costly.) Thus, the sender never benefits from disclosing specific channels that carry the causal effect of $x$ on $y$. 

\begin{theorem} \label{thm:direct}
	If $M_t$ is simple, then, for either receiver type, the sender's optimal proposed model $M_s$ is either $M_s = \emptyset$ or a model in which $x$ and $y$ are adjacent.
\end{theorem}
\begin{proof}
	See Appendix \ref{sec:proofs}.
\end{proof}

The causal discovery literature treats the observed variable set as given and asks which orientations the data compel \citep{verma_equivalence_1990,meek_causal_1995}, or mines a fixed dataset for locally compelled links \citep{cooper_simple_1997,mani_theoretical_2012}; work on intervention design instead minimizes the number of experiments \citep{ChooShiragur2023,ShiragurZhangUhler2023}. Theorem~\ref{thm:direct} concerns a third problem: selecting which variables to observe so that purely observational data compel a targeted ancestral relation. The answer it provides is structural: any minimal such selection compels the relation through a direct adjacency, so certifying $x\Rightarrow y$ through mediating links is never cheaper. The computational complexity of finding a minimal selection is, to our knowledge, an open question.

\medskip 

To conclude Section \ref{sec:persuade_blank}, the feasibility and the difficulty of persuasion depend on both the sophistication of the receiver and the topology of the causal graph. Persuading a na{\"i}ve receiver is trivial, as any correlation can be presented as causal, so long as this correlation exists in the data. In contrast, persuading a sophisticated receiver is feasible in a smaller range of scenarios (only when the true model does not contradict the proposed narrative), and always involves revealing additional variables. In the best case, only one or two suitable (obvious or non-obvious) causes need to be revealed in addition to the variables of interest $x$ and $y$ to create a persuasive narrative. In the absence of such variables, persuasion may become arbitrarily difficult or even impossible. However, if it is possible, then (at least in simple worlds) offering a model that directly connects $x$ and $y$ is optimal.

\section{Debunking of pre-existing models} \label{sec:preexist}

The previous section investigated how to persuade a receiver with a blank mind. However, what if the receiver already has a subjective model in mind by the time the sender enters the picture? Can the sender change the receiver's mind? As we show in this section, debunking a receiver's model is qualitatively similar to persuading a sophisticated receiver, even when the receiver is actually na{\"i}ve. We begin by introducing the modified problem.

\subsection{Modified setup: Receiver with a pre-existing model}

Suppose the receiver is initially aware of a subset $\Omega_r \subset \Omega_t$. This means that they only observe the data $P|\Omega_r$ on variables in $\Omega_r$ but not other variables. The receiver further has some subjective model of the world $M_r=(\Omega_r,C_r)$ that is consistent with $P|\Omega_r$.

As before, the sender knows the true model $M_t$ and can propose a model $M_s$ to the receiver. The sender can propose any model $M_s$ that is consistent with $P|\Omega_s$ and is such that $\Omega_r \subseteq \Omega_s \subseteq \Omega_t$ (the sender knows the set of variables known to the receiver and cannot make the receiver forget any variables). Alternatively, the sender can stay silent, $M_s = \varnothing$.

For the receiver's choice procedure to be well-defined, we must specify when the initial model $M_r$ remains compatible with data involving additional variables. Given $\Omega_r\subseteq\Omega_s$, call a model $M\in\mathcal M(\Omega_s)$ a \emph{consistent extension} of $M_r$ to $\Omega_s$ if it preserves every causal link in $M_r$: i.e., for all $a,b\in\Omega_r$, $a\to_r b$ implies $a\Rightarrow_M b$. We say that the disclosure $\Omega_s$ \emph{debunks} $M_r$ if $M_r$ has no consistent extension to $\Omega_s$.\footnote{We also abuse the notation and say that model $M_s = (\Omega_s,C_s)$ debunks $M_r$ if disclosure $\Omega_s$ does so.} 

Given the sender's proposal $M_s$, we assume that if there exists a consistent extension of $M_t$ to $\Omega_s$, then the receiver adopts it (breaking ties arbitrarily) and rejects the sender's proposed model. If no such extension exists, the receiver discards $M_r$ and evaluates $M_s$ according to the same acceptance rule as before. For instance, take our illustrative example in Section \ref{sec:example}. The employee's subjective model after hearing the business school advertisement is as in Figure~\ref{fig:pic3}(b). If the employer then reveals variable $a$, the employee's subjective model does not contradict the new data, and so their subjective model after this conversation is as in Figure~\ref{fig:pic4}(b).

We again assume that the sender's objective is to \emph{persuade} the receiver to accept a model $M_s$ such that either $x \Rightarrow_s y$ for some $x,y \in \Omega_r$, or $x$ and $y$ are not adjacent (both cases are considered).\footnote{Our results regarding when and how the sender can successfully rule out a link between $x$ and $y$ also apply directly to the persuasion problem considered in Section \ref{sec:persuade_blank}. This case was not covered previously because it is difficult to think of a real-world setting in which the sender would want to reveal two variables that the receiver was not aware of and subsequently try to persuade the receiver that the two are not causally related.} 
We focus on the interesting case where the receiver is initially aware of both variables, $x,y \in \Omega_r$, but their existing model does not satisfy the sender's objective. Further, we maintain the sender's secondary objective to persuade using the smallest number of variables, i.e., to minimize $|\Omega_s|$.

The sender's problem in this modified model is twofold. First, they choose which new variables $\Omega_s \setminus \Omega_r$ to present to the receiver to debunk their initial model $M_r$: the expanded data must admit no consistent extension of $M_r$.  Section \ref{sec:debunk} examines when and how the receiver's model can be debunked.
The second part of the sender's problem is to propose a new consistent model $M_s$ that the receiver of a given type accepts. We reserve \emph{successful persuasion} for a proposal satisfying both requirements. Thus, debunking the incumbent model is necessary but need not be sufficient for persuasion. Sections \ref{sec:persuade} and \ref{sec:breaklink} discuss when and how the sender can successfully persuade the receiver.

\subsection{Debunking the receiver's model} \label{sec:debunk}

One sufficient way to debunk $M_r$ is to rule out one of its causal links. We say that the disclosure $\Omega_s$ \emph{debunks the link} $a\to_r b$ if $a\not\Rightarrow_M b$ for every $M\in\mathcal M(\Omega_s)$. Debunking any one link is sufficient to debunk $M_r$, because no consistent extension can then preserve all of its links. The converse need not hold: the links may each be individually compatible with the new data while not being jointly compatible within a single model (see Example~\ref{ex:debunk_nolink} in Supplementary Appendix~\ref{sec:examples}). Our results below focus on the link-level route to debunking.

Debunking a link for both types of receivers is essentially equivalent to persuading a sophisticated receiver about that link, as characterized in Section \ref{sec:persuade_soph}. This is because debunking a link $x \leftarrow_r y$ requires disclosing data that are incompatible with the link---which means either compelling a link $x \to_s y$ (or, more generally, showing that $x \Rightarrow_M y$ for all $M \in \mathcal{M}(\Omega_s)$), or showing that $x$ and $y$ are correlated purely due to confounders but $x \not\Leftarrow_s y$ and $x \not\Rightarrow_s y$ (this option is explored in Section \ref{sec:breaklink}).
Consequently, our results regarding what a sophisticated receiver can and cannot be persuaded of have direct implications for which models can and cannot be debunked. We start by stating a corollary of Theorem~\ref{thm:no_persuade_defective}, showing that our notion of a defective link (and a defective model) provides a necessary condition for which links can be debunked.

\begin{corollary} \label{cor:no_debunk_truth}
	The sender can never debunk a non-defective link.
\end{corollary}

Corollary \ref{cor:no_debunk_truth} argues that the sender can debunk only defective links in the receiver's model, that is, causal relations that the receiver believes in, $x \leftarrow_r y$, but which do not exist in reality, $x \not\Leftarrow_t y$. In particular, if the receiver's model $x \leftarrow_r y$ is incorrect only in that it omits some proxy variable $a$ (meaning that the true model is such that $x \leftarrow_t a \leftarrow_t y$), then revealing $a$ would only make the receiver expand their model, but not abandon it. 
We can further state the following result (which is not immediately implied by Corollary \ref{cor:no_debunk_truth} due to the point above that debunking a link is not equivalent to debunking a model, but follows immediately from Lemma~\ref{lem:realizability} in the proof of Theorem~\ref{thm:no_persuade_defective}).

\begin{corollary} \label{cor:no_debunk_truth_m}
	The sender can never debunk a non-defective model.
\end{corollary}

Corollary \ref{cor:no_debunk_truth_m} says that if the receiver's model is a simplified version of the truth, in the sense that it only omits some variables but does not contain any defective links, then such a model cannot be debunked by the sender. This is true because a simplified version of the truth always admits a consistent extension, and cannot be debunked. 
It follows that the sender can only hope to successfully persuade the receiver if the receiver's pre-existing model is defective. 

The converse is not automatically true: not every defective model can be debunked. There appears to be no simple necessary-and-sufficient characterization of when debunking is possible. However, richness is a trivial sufficient condition: if the true model is rich, then any defective receiver's model can be debunked by revealing the true model $M_t$. The relevant question then becomes one of cost: is it necessary to reveal the whole true model, or would a small, well-chosen subset of variables suffice?
We next show that when the true model is simple and rich, Proposition~\ref{prop:persuade_sophist} can be extended to debunking models. Specifically, we argue that obvious or non-obvious causes can be used to ``easily'' debunk a defective model.

\begin{theorem}\label{thm:debunk_simple}
	Suppose the receiver's model $M_{r}$ is consistent with $P|\Omega_{r}$ and has a defective link $x \gets_{r} y$. Then:
	\begin{enumerate}
		\item If there exists an obvious cause $z$ of $y$ given $x$, then the link $x \gets_{r} y$, and hence $M_{r}$, can be debunked by revealing one new variable: $\Omega_s = \Omega_r \cup \{z\}$. 
		\item If $M_{t}$ is simple and rich, $x\Rightarrow_t y$, and a non-obvious cause of $y$ given $x$ exists, then the link $x \gets_{r} y$, and hence $M_{r}$, can be debunked by revealing at most two new variables.
	\end{enumerate}
\end{theorem}
\begin{proof}
	See Appendix \ref{sec:proofs}.
\end{proof}

Theorem~\ref{thm:debunk_simple} generalizes Proposition~\ref{prop:persuade_sophist} to the case of the receiver with an arbitrary pre-existing model. In Proposition~\ref{prop:persuade_sophist}, the receiver did not have any model in mind, so it was safe to assume the causes would be adjacent to the variables of interest. Here, by contrast, the receiver already holds a subjective model involving additional variables that might stand between the candidate causes and the variables of interest. Nevertheless, the result continues to hold.
This equivalence between debunking a defective link and persuading a blank-mind sophisticated receiver has no analog in the causal discovery literature, which does not model an agent whose beliefs must be corrected.
Examples \ref{exp:polls}--\ref{exp:inflation} in Supplementary Appendix \ref{sec:examples} illustrate how the insights of Theorem~\ref{thm:debunk_simple} could be applied in different real-world settings.

The proof of Theorem~\ref{thm:debunk_simple}, much like that of Proposition~\ref{prop:persuade_sophist}, proceeds by constructing a V-structure either upstream of $y$ (using an obvious cause), or upstream of $x$ (using a non-obvious cause). With an obvious cause, the construction is straightforward. With a non-obvious cause, we use model richness to show that there must exist another independent non-obvious cause. Simplicity then ensures the projection preserves identifiability, so the two causes orient $x \to_s y$ and debunk the receiver's defective link. Alternatively, part 2 of the theorem can be formulated without relying on richness, in a way it was done in Proposition~\ref{prop:persuade_sophist}. 

As with persuasion, if the premise of Theorem~\ref{thm:debunk_simple} does not hold, it may be difficult to debunk the receiver's defective model. That is, the sender may then have to reveal arbitrarily many new variables. 
For example, suppose that the receiver's model is $x \leftarrow_r y$ and is defective, meaning $x \not\Leftarrow_t y$. This opens up two cases: either $x \Rightarrow_t y$, or $x \not\Rightarrow_t y$. 
In the former case, Example~\ref{exp:nocauses} from Supplementary Appendix \ref{sec:examples} applies, which demonstrates that the sender may need to reveal arbitrarily many variables to debunk the receiver's model $x \leftarrow_r y$. The latter case is discussed in Section \ref{sec:breaklink}.
Further, Theorem~\ref{thm:debunk_simple} does not guarantee that the sender can propose a consistent replacement model using the same disclosure $\Omega_s$. This is illustrated by Example~\ref{exp:decept} in Supplementary Appendix \ref{sec:examples}.
Finally, we do not explore the possibility of debunking the receiver's model without debunking any of its links, which could be possible in principle.

\subsection{Persuading a receiver with a pre-existing model} \label{sec:persuade}

The previous subsection asked when can a defective model be (easily) debunked. We now turn to the second requirement: replacing the debunked model with a model that the receiver accepts.
The sender wants to persuade the receiver that $x \Rightarrow_s y$, and we assume the receiver's initial model is such that $x \not\Rightarrow_r y$.
We consider different cases corresponding to whether the respective link is present in the true model or not.

\paragraph{Truth is on the sender's side.}
If the true model $M_t$ is such that $x \Rightarrow_t y$ and the world is rich, then the sender can always persuade the receiver that $x \Rightarrow_s y$. By definition of richness, this can be done by revealing the true model: $M_s = M_t$. In other words, while persuasion may be difficult (due to $M_t$ containing many variables and links), it is possible in principle. Further, the debunking certificates identified in Theorem~\ref{thm:debunk_simple} can also support an accepted replacement. Hence, if an appropriate cause of $y$ exists, then successful persuasion requires at most two additional variables. This applies regardless of whether the receiver is na{\"i}ve or sophisticated. The result is summarized by the following corollary.

\begin{corollary}
	If the true model is simple and rich and such that $x \Rightarrow_t y$ and there exists an (obvious or non-obvious) cause of $y$ given $x$, then the sender can successfully persuade the receiver that $x \Rightarrow_s y$ with at most two additional variables.\footnote{Note specifically that the assumption $x \Rightarrow_t y$ rules out the situation described in Figure~\ref{fig:decept} and Example~\ref{exp:decept} from Supplementary Appendix \ref{sec:examples}, where revealing an obvious cause produces an inconsistent model.}
\end{corollary}

\paragraph{Spurious correlation.}
If the true model is such that $x$ and $y$ are not adjacent and $x \not\Rightarrow_t y$, $x \not\Leftarrow_t y$, but rather there exist some confounders $c \in \tilde{C}_t(x,y)$ such that $c \Rightarrow_t x,y$, then two cases are possible.
If $x$ and $y$ are nonadjacent in $M_r$, additional disclosure cannot make them adjacent in any consistent expansion: a d-separating set available in $\Omega_r$ remains available after disclosure. This rules out persuasion through a direct $x$-$y$ link, but the sender may still be able to flip the direction of $x$-$c$ or $c$-$y$ links. If, on the other hand, $x \leftarrow_r y$ (due to the receiver being unaware of some of the confounders $c \in \tilde{C}_t(x,y)$), then the sender can flip this link and persuade the receiver that $x \rightarrow_s y$, as is illustrated by Example~\ref{exp:decept2} below. This is possible even though there is no actual link between $x$ and $y$ in the true model, meaning the sender debunks the receiver's defective model and replaces it with another defective model. This is also possible regardless of whether the receiver is na{\"i}ve or sophisticated.
However, such deception is only possible if the true model allows it and if the receiver is unaware of any causes of $x$ given $y$, as illustrated by Example~\ref{exp:decept} in Supplementary Appendix \ref{sec:examples}.

\begin{example}[Flipping a spurious link] \label{exp:decept2}
	Suppose the true model is as depicted in Figure~\ref{fig:decept2}(a), and the receiver's model is $x \leftarrow_r y$, as in Figure~\ref{fig:decept2}(b). Then the sender can successfully persuade the receiver that $x \to_s y$ by revealing variable $z$ (the proposed model is shown in Figure~\ref{fig:decept2}(c)), in which case $x,y,z$ form a V-structure in the data.
\end{example}
\begin{figure}
	\centering
	
	\begin{minipage}[b]{0.32\textwidth}
		\centering
		
		\includegraphics[scale=0.3]{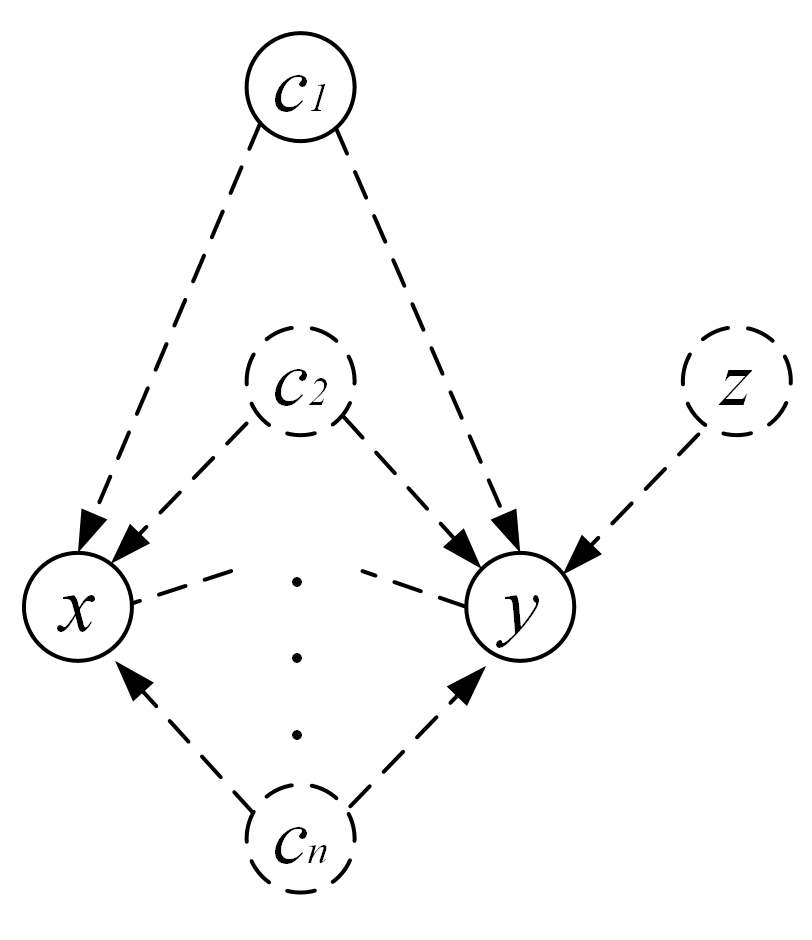}\\
		(a) $M_t$\\
	\end{minipage}
	\hfill
	\begin{minipage}[b]{0.32\textwidth}
		\centering
		
		\includegraphics[scale=0.3]{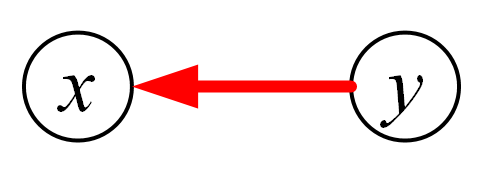}\\
		(b) $M_r$\\
	\end{minipage}
	\hfill
	\begin{minipage}[b]{0.32\textwidth}
		\centering
		
		\includegraphics[scale=0.3]{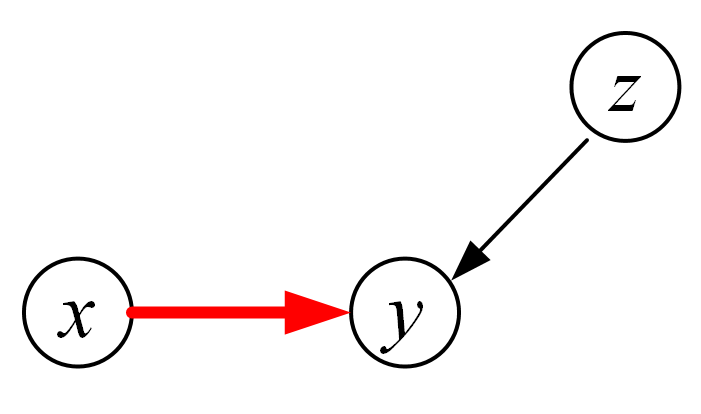}\\
		(c) $M_s$\\
	\end{minipage}
	
	\caption{Replacing a defective model with a defective model. \label{fig:decept2}}
\end{figure}

\paragraph{Truth is against the sender.}
If the true model is such that $x \Leftarrow_t y$, then Corollary \ref{cor:no_debunk_truth} shows that the sender can never debunk \emph{this} link and present data for which only $x \Rightarrow_s y$ is consistent with it. This implies that a sophisticated receiver cannot be persuaded in this case.
However, persuading a na{\"i}ve receiver may be possible if $M_r$ contains some other defective link. If so, the sender can ``nitpick'' the na{\"i}ve receiver's model: debunk some defective link that is not directly related to the targeted link, and then use this opportunity to deceive the receiver. The following example demonstrates how such ``persuasion by nitpicking'' can work.

\begin{figure}
	\centering
		
	\begin{minipage}[b]{0.32\textwidth}
		\centering 
		\includegraphics[scale=0.27]{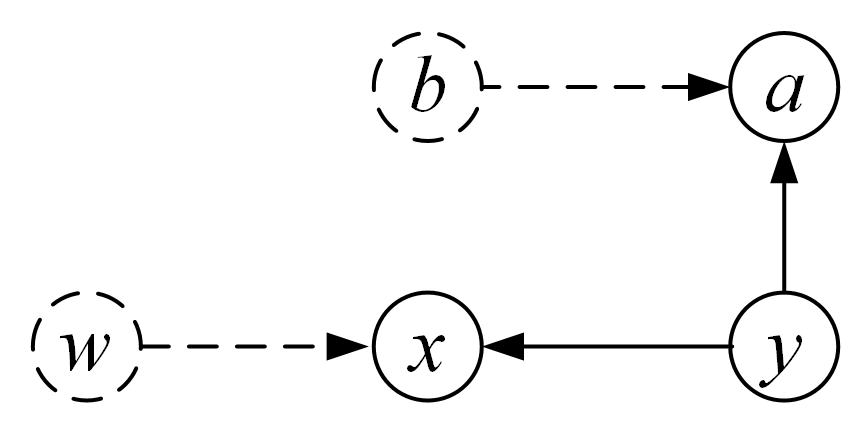}\\
		(a) $M_t$\\
	\end{minipage}
	\begin{minipage}[b]{0.32\textwidth}
		\centering 
		\includegraphics[scale=0.27]{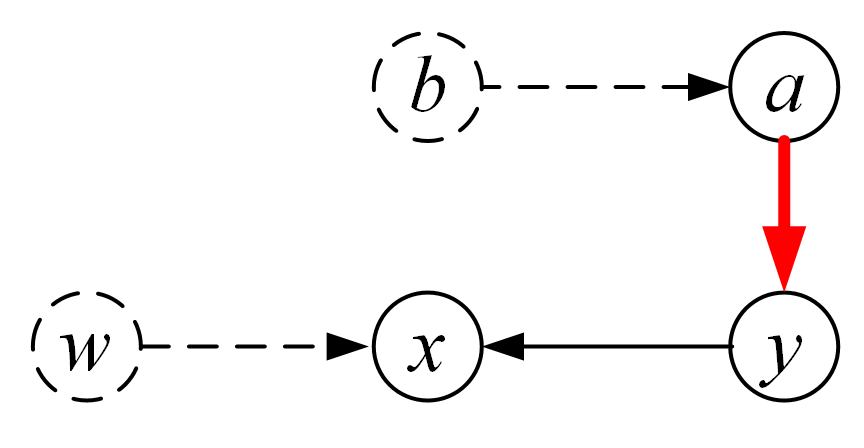}\\
		(b) $M_r$\\
	\end{minipage}
	\begin{minipage}[b]{0.32\textwidth}
		\centering 
		\includegraphics[scale=0.27]{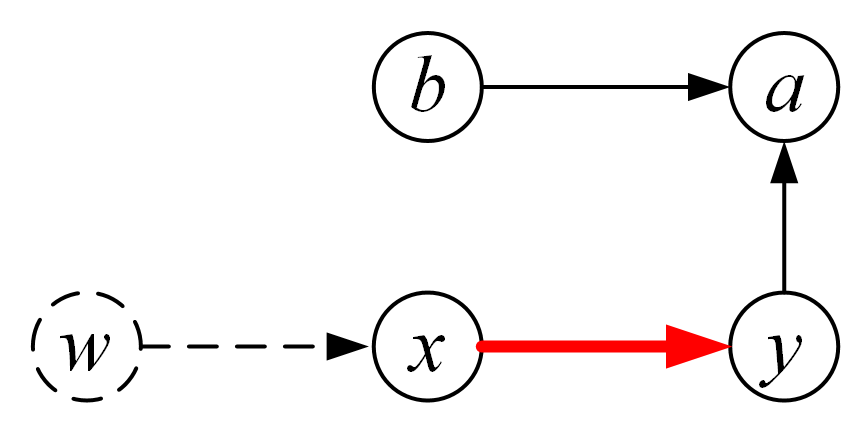}\\
		(c) $M_s$\\
	\end{minipage}
	
	\caption{Persuasion by nitpicking. \label{fig:nitpick}}
\end{figure}
\begin{example}[Persuasion by nitpicking] \label{exp:nitpick}
	Suppose the true world is as shown in Figure~\ref{fig:nitpick}(a). The na{\"i}ve receiver observes variables $a$, $x$ and $y$ and thinks the model is as shown in Figure~\ref{fig:nitpick}(b). Specifically, the receiver correctly believes that $x \leftarrow_r y$, but also mistakenly believes that $a\to_r y$. The sender would like to deceive the receiver and persuade them that $x \to_s y$. Because $x \leftarrow_t y$ is the truth, the sender cannot debunk this link directly. Instead, the sender can debunk the defective link $a\to_r y$. In this example, the sender is able to reveal variable $b$, which is the obvious cause of $a$ given $y$. The sender will then offer the model shown in Figure~\ref{fig:nitpick}(c), which will be accepted by the na{\"i}ve receiver.
\end{example}

As before, successful deception along the lines of Example~\ref{exp:nitpick} requires that the link of interest is not compelled by $P|\Omega_r$: models with both $x \Rightarrow y$ and $x \Leftarrow y$ must be consistent with the data available to the receiver. Otherwise, if the na{\"i}ve receiver can unambiguously infer from the data $P|\Omega_r$ that $x \Leftarrow_r y$, then no nitpicking can help the sender persuade them that $x \Rightarrow_s y$.
Further, the scope for such deception also depends on the relationship between the link of interest and the defective link. If both links trace back to the same V-structure, deception may not be possible, as demonstrated by Example~\ref{exp:nitpick_no} in Supplementary Appendix \ref{sec:examples}.

\subsection{Dissuading a receiver with a pre-existing model} \label{sec:breaklink}

What if the sender wants to rule out a link between $x$ and $y$? In other words, the receiver believes that $x \Leftarrow_r y$, and the sender wants to persuade them that there is no link between $x$ and $y$ (i.e., propose a model such that $x \not\Leftarrow_s y$ and $x \not\Rightarrow_s y$). This could, for example, be the case when the sender tries to debunk a defective model fueled by spurious correlation, which has previously been imprinted on the receiver. 

Ruling out a direct link between two correlated variables $x$ and $y$ requires the sender to demonstrate the proxy and/or confounding variables that capture all sources of correlation between $x$ and $y$. Without such complete disclosure, $x$ and $y$ would still be correlated from the receiver's perspective conditional on any set of other variables known to them, and hence the receiver would still draw a direct causal link between $x$ and $y$. In other words, the sender needs to reveal a set of variables that d-separates $x$ and $y$ (Definition~\ref{def:dsep}). The statement is formalized in the following proposition, which is trivial and is presented without proof.

\begin{proposition} \label{prop:lat_dsep}
	Suppose the receiver's model contains a link $x \leftarrow_r y$. This link is debunked by a sender's model $M_s$ in which $x$ is not adjacent to $y$ only if $\Omega_s$ contains a set of variables $S \subset \Omega_s$ that d-separates $x$ and $y$. Further,
	\begin{enumerate}[nosep]
		\item if $x \leftarrow_t y$ or $x \to_t y$, then no such set $S \subset \Omega_t$ exists;
		\item if $x$ and $y$ are not adjacent in $M_t$, then such a set $S \subset \Omega_t$ necessarily exists.
	\end{enumerate}
\end{proposition}

In some situations, it is trivial that no such d-separating set exists. Specifically, if $x \leftarrow_t y$ or $x \to_t y$, then Lemma~\ref{lem:consistent_models} implies that $x$ and $y$ must be adjacent in any consistent subjective model as well. It follows that it would not be possible to dissuade the receiver of a link between $x$ and $y$ in this case, regardless of the receiver's type. 
The remaining cases are explored below.

\paragraph{Truth is on the sender's side.}
Suppose $x$ and $y$ are not adjacent (but correlated) in the true model. Then revealing a d-separating set is necessary to rule out a direct link between the two variables in the receiver's model. This can be as simple as revealing a single hidden confounder (as in the case of polio, see Example~\ref{exp:polio} in Supplementary Appendix \ref{sec:examples}).
However, while finding a minimal d-separating set is a well-studied problem solvable in polynomial time \citep{TianPazPearl1998, vanderZanderLiskiewicz2020}, even the minimal set can be arbitrarily large, making the associated disclosure cost prohibitively high. For example, in the $M_t$ depicted in Figure~\ref{fig:decept2}(a), the only set that d-separates $x$ and $y$ is $S \equiv \{c_1, \dots, c_n\}$, meaning the sender would have to reveal $n$ variables to break a link between $x$ and $y$.

\begin{example}[Immigration and crime]
	Consider immigration $x$ and crime $y$. The two are correlated in aggregate data---largely because immigrants settle in urban, economically distressed areas---but causal studies find no effect of immigration on overall crime \citep[e.g.,][]{bianchi2012,marie_2024}. Yet the causal belief persists and fearmongering continues.\footnote{Donald Trump stated in his presidential campaign announcement speech (June 16, 2015) that immigrants are ``bringing drugs, they're bringing crime, they're rapists.'' Full transcript: \url{https://time.com/3923128/donald-trump-announcement-speech/}, retrieved Apr 03, 2026.} 
	Our model explains why: persuading a receiver that a link is \emph{absent} requires disclosing the d-separating set of the two variables---settlement patterns, local economic conditions, demographics, and so on. An instrument, by contrast, works around the confounders rather than revealing them. So long as any confounder stays hidden, the receiver sees that immigration and crime remain correlated conditional on everything they know, and continues to draw a causal link. 
\end{example}

Further, revealing a d-separating set is necessary but not necessarily sufficient to debunk the receiver's model. Revealing $S$ breaks the direct link between $x$ and $y$, but may not help identify the directions of the links between $x,y$ and all the d-separating variables. For instance, if the true model is $x \leftarrow_t c \to_t y$ and the receiver starts with a model $x \leftarrow_r y$, then revealing $c$ does not by itself debunk the receiver's model. In particular, the receiver would regard the data on $x,y,c$ as consistent with an extended model $x \leftarrow_r c \leftarrow_r y$ and reject the sender's model claiming $x \leftarrow_s c \rightarrow_s y$.


We conclude that breaking a link between $x$ and $y$ is a difficult task. Therefore, if the sender's goal is merely to debunk the receiver's original model $x \leftarrow_r y$ (without regard to which alternative model is put in place), then truthful persuasion may not be the best course of action when $x\not\Leftarrow_t y$ and $x\not\Rightarrow_t y$. Indeed, in Example~\ref{exp:decept2}, a receiver's defective model can be debunked by another defective model via revealing just one new variable, while debunking it by a non-defective model requires revealing arbitrarily many new variables $S = \{c_1, ..., c_n\}$. Persuading with a defective model may thus be easier than proving the truth, regardless of the receiver's type.

\paragraph{Truth is against the sender.}
We now move on to the case in which the sender wants to persuade the receiver that $x\not\Leftarrow_s y$ and $x\not\Rightarrow_s y$ but the truth is against the sender in that either $x \Leftarrow_t y$ or $x \Rightarrow_t y$. In this case, Theorem~\ref{thm:no_persuade_defective} suggests that the sender can never prove conclusively that $x$ and $y$ are correlated purely due to confounders. This then implies that a sophisticated receiver cannot be persuaded. However, if the receiver is na{\"i}ve, then just like in Example~\ref{exp:nitpick} above, the sender may still be able to persuade them that $x$ and $y$ do not affect one another. If the receiver believes $x \leftarrow_r y$ and this link is defective, meaning $x \Rightarrow_t y$, then the sender may be able to debunk this link and offer a fitting consistent model with confounders. Otherwise, if $x \Leftarrow_t y$ and $x \leftarrow_r y$, then  ``persuasion by nitpicking'' may still be possible if $M_r$ has some other defective link. Examples \ref{exp:dissuade_2} and \ref{exp:dissuade_1} in Supplementary Appendix \ref{sec:examples} demonstrate these possibilities in the two cases (when $x \Rightarrow_t y$ and $x \Leftarrow_t y$, respectively, all while $x \leftarrow_r y$).

\section{Applications} \label{sec:applications}

The results so far characterize a single interaction between one sender and one receiver. To illustrate their dynamic implications, we now consider an exogenous sequence of senders with conflicting objectives. We do not model sender arrival or solve a dynamic equilibrium; instead, each interaction follows Section~\ref{sec:preexist}: each sender observes the receiver's current awareness set and model, and can disclose additional variables together with a consistent proposed model. Disclosure is cumulative. The conclusions below should therefore be understood as conditional on the sequence of senders and disclosures. Throughout the applications, the disputed link between $x$ and $y$ is the only defective link in the incumbent model, ruling out persuasion by nitpicking (Example~\ref{exp:nitpick}).

We show that this dynamic environment is strongly path-dependent, and this path dependence is built on a fundamental asymmetry: a causal claim can \emph{enter} the receiver's model while it is merely possible (for na{\"i}ve receivers), but it \emph{exits} only when it is refuted.
We also contrast the dynamic implications of our model with the dynamic version of \citet{schwartzstein_using_2021}.

\subsection{Labor unions} \label{sec:app_unions}

We now expand on the labor union story from Examples \ref{exp:union1} and \ref{exp:union2}, highlighting other aspects of this setting and how our results apply.

For as long as labor unions have existed, they have promised higher wages for their members. Both unionization rates and union membership $u$ (for simplicity, we do not distinguish between the two) are correlated with wages $w$, see \citet{lewis_chapter_1986}. 
The unions have spun this correlation as causal evidence that $u \to_s w$ \citep[e.g., ][]{white_collar_unionism_1970,seiu_unions_2018}.
On the other side of the debate is the selection argument: ``most students of unionism, I think, would argue that union status selection is not random'' \citep[p. 1143]{lewis_chapter_1986}, suggesting that $u \leftarrow_s w$ is the correct model. Specifically, \citet{zerzan_unions_1976} argued that ``it is generally the relatively more prosperous wage fields that become unionized.'' 
\citet{bockerman_erosion_2006} offer a specific selection channel: high earners assign higher value to unemployment insurance; if it is mainly offered by labor unions then high earners are more likely to select into joining a union.

Causal evidence suggests that the latter narrative is closer to the truth. Under U.S. law, if a majority of workers vote in favor of the union, management must recognize and bargain with the union. \citet{dinardo_economic_2004} use this in a regression discontinuity design, comparing outcomes in firms where the vote barely passed with those where it barely failed, arguing that such firms must be sufficiently similar in all other regards. They show that outcomes in the two groups are statistically indistinguishable, suggesting that unions do not have a causal effect on wages, $u \not\Rightarrow_t w$. 

The design of \cite{dinardo_economic_2004} can be represented in our language as $z \to_t u \leftarrow_t w$, where $z$ is random noise in the voting outcomes, which is independent of $w$ but affects whether a union is formed, $u$. Adopting this as the true model, we can reconstruct the following (heavily simplified and abridged) version of the debate around labor unions using our model.
\begin{enumerate}[noitemsep]
	\item First, unions use descriptive statistics to propose the narrative $u \to_s w$ \citep{white_collar_unionism_1970}. This persuades na{\"i}ve workers to join but not the sophisticated ones.
	
	\item Then, the opponents such as \citet{zerzan_unions_1976} argue that $u \leftarrow_s w$ is the correct model instead. This is best described as screaming into the void, since this neither sways the na{\"i}ve workers---who have adopted the unions' model $u \to_r w$, which is not debunked by $u \leftarrow_s w$---nor persuades the sophisticated workers, who remain unconvinced by either model.
	
	\item Finally, causal evidence from \citet{dinardo_economic_2004} reveals $z \to_t u \leftarrow_t w$ as the true model. This persuades both the na{\"i}ve and the sophisticated workers that $u \leftarrow w$ and shuts down further debate. Indeed, while some unions do still advertise higher wages for union members, many others do not, restricting themselves to broader promises of job protection, opportunity to bargain with employers, and positive change.\footnote{``How would a union benefit me?'' Office \& Professional Employees International Union F.A.Q. URL: \url{https://www.opeiu.org/NeedAUnion/FrequentlyAskedQuestions.aspx}, retrieved Jul 15, 2026.}$^,$\footnote{``Why should I join?'' National Tertiary Education Union F.A.Q. URL: \url{https://www.nteu.au/NTEU/FAQs/Library?hkey=8f2e9fd5-5bb5-48c8-8c86-e9e4fcebd651}, retrieved Jul 15, 2026.}
\end{enumerate}

Had the evidence of \citet{dinardo_economic_2004} been publicly known before the unions' campaign, the data on $\{u, w, z\}$ would have identified $z \to u \leftarrow w$ from the outset. Workers of either type would have adopted the selection story, and the unions' claim $u \to w$ would have no bite. An early disclosure of $z$ would have \emph{inoculated} the receivers of either type against the unions' narrative. 

We also use this opportunity to highlight the differences in the dynamics of narratives enabled by our model and that of \citet{schwartzstein_using_2021}.\footnote{While \citet{schwartzstein_using_2021} do not explicitly set up a dynamic version with multiple senders, we consider a natural model that results from running a series of their static models across a number of periods. This effectively restricts senders to be myopic/non-strategic about future senders' behavior, but suffices for illustrative purposes, since the same restrictions apply in our model as well. \label{foot:SS_extension}}  
In our model, there is a first-mover advantage: the initial narrative instilled by labor unions in the first stage becomes difficult to debunk. In their model, each new sender can offer a marginally better-fitting model and effectively persuade the receiver. The sequence above would therefore play out as follows in the model of \citet{schwartzstein_using_2021}:
\begin{enumerate}[noitemsep]
	\item First, unions use descriptive statistics to propose the narrative $u \to_s w$, which aligns with the observed signals and persuades the workers.
	\item Then, an opponent like \citet{zerzan_unions_1976} could use recent anecdotes and selected evidence (which is indeed his primary mode of persuasion in the referenced article) to argue that model $u \leftarrow_s w$ fits these observed signals better. Since the fit is indeed better, the workers are persuaded by these arguments and become disenchanted with unions.
	\item The unions could respond by presenting other anecdotes which make the model $u \to_s w$ fit the data better. The workers would be persuaded again.
	\item Causal evidence from \citet{dinardo_economic_2004} would reveal a whole new variable, enlarging the dataset and making the model with $u \leftarrow_s w$ (and $z \to_s u$) fit the data better again.
	\item The unions could again respond by revealing either more anecdotal evidence, or additional variables, together with a model describing conditional distributions of all these variables in a way that makes the model with $u \to_s w$ fit the data better again.
\end{enumerate}
In the model of \citet{schwartzstein_using_2021}, each new sender can offer a distribution of signals conditional on states of the world that renders their preferred model a better fit for the limited observed data on signal realizations. The authors argue that the process stops with the model that claims the observed sequence of signals was inevitable regardless of the state (and thus leaves the receiver with no information about the state of the world apart from their prior belief). However, if the sender can reveal new evidence or, like in our setting, additional variables (which is not allowed in their model), then the process could arguably run indefinitely. Each competing sender would offer new data discrediting the old model's fit together with a better-fitting model that supports their preferred narrative. The receiver would then flip between the competing narratives indefinitely.
In contrast, in our setting, the receiver only flips to a competing narrative if their old model is debunked. Debunking is irreversible, meaning that they will then never adopt this narrative again. This is because identification is discrete, whereas fit in \citet{schwartzstein_using_2021} is continuous. If, for every remaining defective link in our model, some sender eventually finds it profitable to disclose a debunking certificate, cumulative disclosure eliminates all false narratives and the receiver converges to the true model.


\subsection{MBA: competing obvious causes}\label{sec:MBA-dynamic}

\begin{figure}
	\centering
	\includegraphics[scale=0.3]{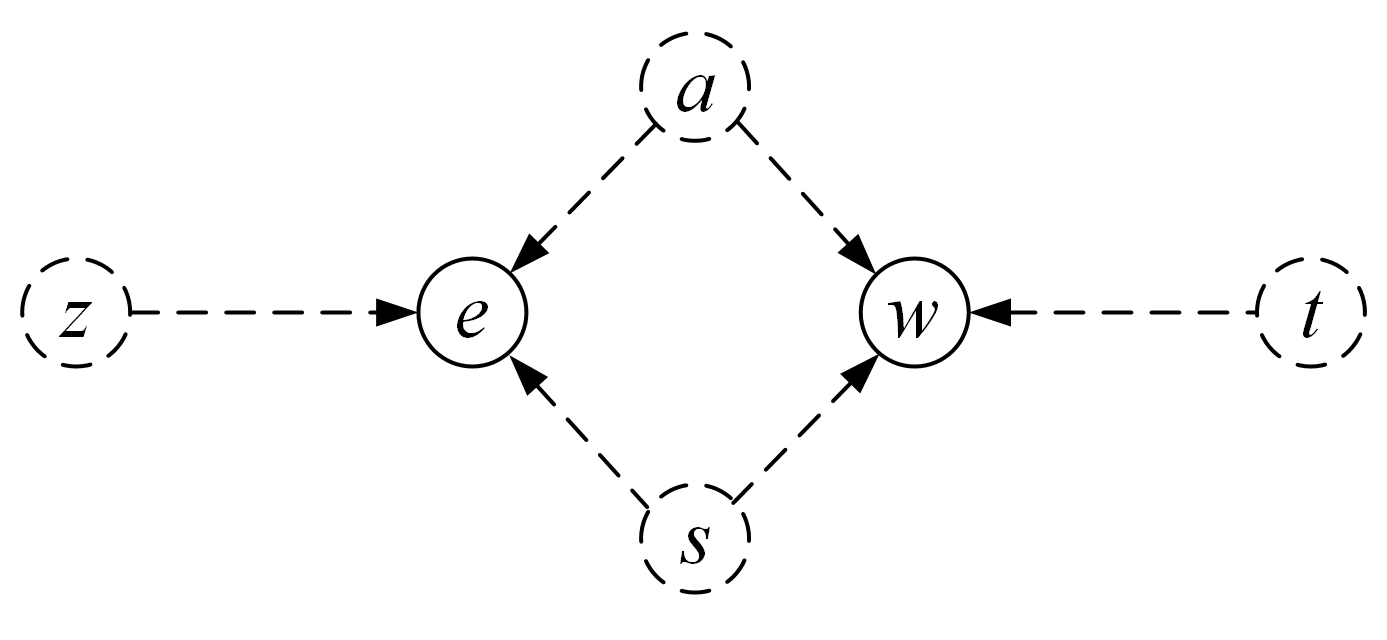}
	\caption{MBA example with competing obvious causes $z$ and $t$. \label{fig:mba_new}}
\end{figure}
We now return to the MBA example of Section~\ref{sec:example} and place it in the sequential environment of this section. To make the competition between senders interesting, we enrich the true model of Figure~\ref{fig:pic3}(a) with one additional variable. Let $z$ denote an \emph{admission or scholarship lottery}, with $z$ being either idiosyncratic admissions noise or the lottery outcome.\footnote{``Niche MBA Scholarship'', URL: \url{https://www.niche.com/colleges/scholarships/mba-scholarship/}, retrieved Jul 20, 2026.}$^,$\footnote{``GMAT Club MBA Scholarship'', URL: \url{https://gmatclub.com/forum/gmat-club-mba-scholarship-2025-2026-10-000-application-open-451640.html/}, retrieved Jul 20, 2026.} 
The true model becomes as in Figure~\ref{fig:mba_new}. As before, we assume the correlation between education $e$ and earnings $w$ is generated entirely by the confounders, $a$ (ability) and $s$ (social skills).

The model contains two obvious causes. Tenure $t$ is an obvious cause of $w$ given $e$, while lottery $z$ is an obvious cause of $e$ given $w$. By Proposition~\ref{prop:persuade_sophist}, each variable suffices to persuade even a sophisticated receiver of one of two \emph{opposing} claims:
\begin{equation}\label{eq:mba-margins}
    \Omega_s = \{e, w, t\}
    \;\Longrightarrow\;
    e \to_s w \leftarrow_s t,
    \qquad\qquad
    \Omega_q = \{e, w, z\}
    \;\Longrightarrow\;
    z \to_q e \leftarrow_q w.
\end{equation}
In the left model, $M_s$, link $e \to_s w$ is compelled by the data. In the right model, $M_q$, link $e \leftarrow_q w$ is similarly compelled. Both identified stories are defective: the true model contains neither $e \to_t w$ nor $e \leftarrow_t w$. 

Further, the two models are incompatible: if all variables $\Omega \equiv \Omega_s \cup \Omega_q$ are revealed, there is no model $M$ that is consistent with the data on all of them (similar to Example~\ref{exp:decept} in Supplementary Appendix \ref{sec:examples}). 
Indeed, as long as at least one confounder is unobserved, $e$ and $w$ remain adjacent in every candidate model. The data then demand a collider at $w$ (from $t$) and a collider at $e$ (from $z$) on the same edge. Only models containing \emph{both} $a$ and $s$ are consistent once $t$ and $z$ have been disclosed, since then $(e \perp w \mid a, s)$ breaks the adjacency and both V-structures can coexist. 

Because disclosure is cumulative, the above implies that whichever obvious cause is disclosed \emph{first} renders the opposing obvious cause powerless as a stand-alone response. We now show how persuasion unfolds depending on the order in which the senders make their claims.

\paragraph{1. The business school moves first.}
The school launches its campaign around $\{e, w, t\}$, as in Section~\ref{sec:example}. The data then identify $e \to w \leftarrow t$, and both na{\"i}ve and sophisticated employees adopt the MBA-premium story.
Consider now the employer, who wants to dissuade the employee (the receiver) from undertaking education. After the school's campaign, they can no longer simply reveal $z$, since there is no consistent model to accompany it. Instead, the employer can only reveal $\{a,s\}$, which proves the absence of a causal link between $e$ and $w$ and refutes the school's claim.
The school's first move \emph{inoculates} the employee against the ``simple'' potential counterargument and raises the price of persuasion for the employer from one variable ($z$) to two ($a$ and $s$). 

\paragraph{2. The employer moves first.} 
Suppose instead that the thought of pursuing an MBA crosses the mind of an employee independently of any business school's ad campaign, and they go to discuss it with their employer. The employer is thus the first to propose a model explaining the correlation between $e$ and $w$ observed by the employee.

The employer can then talk about the outcomes of scholarship lotteries, revealing variable $z$ and arguing that model $z \to_q e \leftarrow_q w$ is uniquely identified given these data.\footnote{E.g., \cite{schwerdt2012effects} study a randomized experiment in which vouchers for adult education were assigned by lottery across the general Swiss adult population. They document increased adult-education participation but no statistically significant average earnings effect one year later.}
The employee then believes that it is the high earners who select into education, $e \leftarrow w$. This selection story is itself defective, since the true model contains no link between $e$ and $w$. However, it completely inoculates the employee against the business school's narrative. The school's ad campaign revealing $t$ produces no consistent models, so it has no model to propose to the employee. The school cannot offer any other variables that would persuade the receiver that $e \to w$; the most it can do is debunk the link $e \leftarrow_r w$ by revealing $\{a,s\}$, but this would not convince the employee to enroll in an MBA program. Notably, inoculation against the business school's defective narrative $e \to w$ is done via another defective narrative $e \leftarrow w$.

\paragraph{3. A truth-oriented sender moves first.} 
Finally, suppose the first mover is a regulator or a government agency concerned first and foremost with accuracy. Such a sender can disclose the common causes $\{a,s\}$ directly. In any model with $\{e, w, a, s\} \subseteq \Omega$, the data reveal
the lack of a direct link between education and earnings. Both directional claims, $e \to w$ and $w \to e$, are refuted once and for all. A later disclosure of $t$ cannot revive the school's campaign, since the collider argument $e \to w \leftarrow t$ requires $e$ and $w$ to be adjacent, which has been ruled out.

This route to defeating the school's narrative is more truthful than the employer's scholarship gambit but also more expensive: in this example, it requires two variables instead of one. More generally, it can require many more as the number of confounders grows, possibly becoming infeasible or not worthwhile.

\bigskip 
Comparing the scenarios, we see that unless the business school moves first, it is locked out of instilling its preferred narrative. Its ad campaign would therefore work only among receivers with no (conflicting) pre-existing models---though equally well for na{\"i}ve and sophisticated receivers.
The employer moving second faces a subtler predicament: their original claim $e \leftarrow w$ is likewise unattainable, but they can still achieve their goal by disclosing $\{a,s\}$. The latter option is more truthful, but also more costly. 
Overall, our model predicts a clear \emph{first-mover advantage}. The sender who moves first (i) has a higher chance of instilling their preferred narrative, (ii) incurs a lower persuasion cost, and (iii) might inoculate the receiver against conflicting narratives presented by future senders.\footnote{We clarify what inoculation means formally in Appendix \ref{sec:proofs}, in the remark following Lemma~\ref{lem:cert_noc}.} 
As in the previous example, this contrasts sharply with the results of \citet{schwartzstein_using_2021}, which suggest a last-mover advantage. 

Furthermore, this example refines our conclusions regarding learning outcomes. The previous example suggested that the true narrative will prevail. In this example, the school and the employer only care whether the employee believes that $e \to w$ or not. In this sense the truth prevails: the employee's final model correctly identifies that $e$ does not causally affect $w$ (if the disclosure costs are low enough for the employer to counter the MBA campaign by revealing $\{a,s\}$). However, the employee need not end up with the \emph{correct} model regarding $e$ and $w$: if the employer moves first, then the employee ultimately believes that $e \leftarrow w$, which is incorrect. We conclude that the receiver only ends up learning the true model to the extent that this is in the senders' interests.\footnote{\citet{gentzkow_competition_2017} obtain a somewhat similar conclusion with competing senders in the context of Bayesian Persuasion.}

\section{Model microfoundation} \label{sec:microfound}

This section sets up a fully specified game that we show is equivalent to our baseline setup in Section \ref{sec:model}. This is meant to ground our baseline model in a more standard economic setting, in which the sender wants to persuade the receiver to take a certain action, rather than merely leaving the receiver with some causal model. To capture the receiver's uncertainty over the true model of the world, we use the variational representation of ambiguity-averse preferences due to \citet{maccheroni_ambiguity_2006}, hereinafter MMR.

In particular, we show that a na{\"i}ve receiver behaves like a subjective expected utility maximizer given a consistent sender's model $M_s$. In turn, a sophisticated receiver behaves like an ambiguity-averse agent who considers the worst-case scenario among all models consistent with the data.
The latter kind of behavior is natural: \citet{ambuehl2025} show that when experimental subjects are unable or unwilling to rule out some of the models of the world, they choose cautiously, evaluating options according to the worst-case scenario.\footnote{\citet{bauch2025} offer a different model of narrative communication under ambiguity.}

\subsection{Model setup}
Fix some true causal model $M_t = (\Omega_t,C_t)$.
Consider a game between a sender and a receiver that proceeds in two stages: first the sender proposes a model, and then the receiver makes a decision. At the second stage, the receiver's decision is $x \in \{0,1\}$, where we interpret $x=1$ as ``taking an action'' and $x=0$ as ``the status quo''.
Variable $y$ represents an outcome that is positively correlated with $x$ in the data, but it is unclear whether the two are causally related. The players' terminal payoffs are given by
\begin{align*}
	v_R(x,y) &= x \cdot y - x,
	\\
	v_S(x,y) &= x.
\end{align*}
The sender thus always wants the receiver to choose $x=1$. For specificity, suppose that $\mathbb{E}[y]=0$ (where the expectation is taken with respect to the data $P$). Therefore, if the receiver's model has $x \not\Rightarrow y$, then the receiver prefers action $x=0$. Suppose further that if $x\Rightarrow y$, the receiver prefers $x=1$. To state this condition in terms of primitives, denote the gross benefit of acting according to model $M=(\Omega,C)$ as $b(M) \equiv \mathbb{E} \left[ y \mid M, do(x)=1 \right]$.\footnote{The $do(x)$ operator denotes an intervention that externally sets variable $x$, replacing its usual data-generating mechanism while leaving the rest of the causal system unchanged. See \citet[Ch. 3.2]{pearl_causality_2009}.}  
Then the nontriviality condition is equivalent to requiring $b(M) > 1$ for all models $M=(\Omega,C)$ consistent with $P|\Omega$ and such that $x \Rightarrow_M y$.

The receiver is initially unaware of any variables, apart from $x$ and $y$.
The sender knows the true model $M_t$ and in the game's first stage they propose a model $M_s = (\Omega_s,C_s)$, i.e., they choose a subset $\Omega_s \subseteq \Omega_t$ to disclose and a subjective model $C_s$ that is consistent with $P|\Omega_s$. The sender may also remain silent, $M_s = \varnothing$.

After observing $M_s$, the receiver forms a set $\mathcal{M}(\Omega_s)$ of models that are consistent with the data $P | \Omega_s$.
The receiver then chooses the action $x$ to maximize the following value function:
\begin{align} \label{eq:receiver_util_MMR}
	V_R(x \mid M_s) 
	&\equiv \min_{M \in \mathcal{M}(\Omega_s)} \left\{ \mathbb{E}_y \left[ v_R(x,y) \mid M, do(x) \right] + \frac{1}{\theta} \psi (M \mid M_s) \right\}
	\\ \nonumber
	&= \min_{M \in \mathcal{M}(\Omega_s)} \left\{ x \left(b(M)-1\right) + \frac{1}{\theta} \psi (M \mid M_s) \right\}.
\end{align}
Here, $\psi (M \mid M_s)$ is a model index function that represents, intuitively, whether the receiver considers model $M$ plausible after the sender proposed model $M_s$ (higher $\psi$ means lower plausibility). We assume this function is such that $\psi (M_s \mid M_s) = 0$ and $\psi (M \mid M_s) \in \left[ \underline{\psi}, \bar{\psi} \right]$ with $\underline{\psi} > 0$ for any $M \neq M_s$. Parameter $\theta \in \mathbb{R}_+$ measures ambiguity aversion. Expression \eqref{eq:receiver_util_MMR} then selects the worst causal model for the receiver given $x$, adjusting for how plausible they consider different models to be.
If the sender stays silent, $M_s = \varnothing$ (or proposes an inconsistent $M_s$), we assume the receiver chooses $x^*(M_s)=0$.\footnote{The receiver is only aware of $x$ and $y$ initially. Since they are correlated, then both models $x \leftarrow y$ and $x \rightarrow y$ are consistent, and we assume the receiver treats the former as the focal model, which renders inaction $x=0$ optimal.}

The sender chooses a consistent $M_s$ to maximize 
\begin{align} \label{eq:sender_util_MMR}
	V_S(M_s) \equiv x^*(M_s) - \kappa(|\Omega_s|),
\end{align}
where $x^*(M_s)$ denotes the receiver's best response to hearing the sender's proposed model $M_s$.
By equilibrium we mean a subgame-perfect Nash equilibrium (where the receiver adopts the sender's proposed model $M_s$ as a default model in their problem \eqref{eq:receiver_util_MMR}).

\paragraph{Connection to MMR.}
The model set up above is equivalent to the variational preferences model of \citet*{maccheroni_ambiguity_2006}. This is a broad family of preferences that captures subjective expected utility preferences, maxmin preferences \citep{gilboa_maxmin_1989}, and multiplier preferences \citep{hansen_robust_2001,hansen_robust_1999} as special cases.

The connection between our model and \citet{maccheroni_ambiguity_2006} may not be immediate, since our receiver operates with a set of causal models, as opposed to a set of prior beliefs over states of the world. This issue, however, is easy to address by treating our models as states of the world and extending the receiver's problem to allow for beliefs over models. For a given $M_s$, extend the receiver's cost function to $\varDelta(\mathcal{M}(\Omega_s))$ affinely as $\hat{\psi}(p) \equiv \mathbb{E}_p[\psi(\cdot \mid M_s)]$ for all $p \in \varDelta(\mathcal{M}(\Omega_s))$. Then $\hat{\psi}$ is weakly convex, lower semicontinuous, and grounded ($\hat{\psi}(\delta_{M_s}) = 0$). Further, the MMR criterion
\begin{align*}
	\min_{p \in \varDelta(\mathcal{M}(\Omega_s))} \left\{ \mathbb{E}_M \Big[ \mathbb{E}_y \left[ v_R(x,y) \mid M, do(x) \right] \mid p \Big] + \frac{1}{\theta} \hat{\psi} (p) \right\}
\end{align*}
(where the outer expectation is taken over models $M$ according to $p$) is exactly equal to \eqref{eq:receiver_util_MMR}, since an affine objective on a simplex attains its minimum at a vertex.
It follows that function \eqref{eq:receiver_util_MMR} is indeed a member of the variational preference family.

\subsection{Equivalence to the baseline model}

In what follows, we refer to the setup above as ``the game'' and to the setup in Section \ref{sec:model} as ``the baseline model''.

If $\theta$ is high enough, the receiver behaves as a maxmin agent, maximizing the worst-case payoff among the models they find consistent. Such an agent is ambiguity averse and exhibits the status quo bias: they will choose $x=1$ if and only if all consistent models have $x \Rightarrow y$, and otherwise they default to the status quo decision $x=0$.\footnote{Ambiguity aversion thus produces the status quo bias in this setting. \citet{ortoleva_status_2010} explores the link between the two more generally.}
In contrast, for $\theta$ low enough, the ambiguity term $\frac{1}{\theta} \psi (M \mid M_s)$ rules out all models except $M_s$, making the receiver credulous: they will choose $x=1$ if and only if $x \Rightarrow_s y$ according to the sender's proposed model $M_s = (\Omega_s, C_s)$.
We introduce the labels and summarize this in a subsequent proposition.
\begin{definition}
	The receiver is \emph{credulous} if $\theta < \underline{\psi}$. The receiver is \emph{skeptical} if $\theta > \bar{\psi}$.
\end{definition}
\begin{proposition} \label{prop:ambig_optactions}
	If the receiver is credulous, then $x^*(M_s)=1$ if and only if $x \Rightarrow_s y$ according to $M_s$.
	\\
	If the receiver is skeptical, then $x^*(M_s)=1$ if and only if $x \Rightarrow_M y$ for all $M \in \mathcal{M}(\Omega_s)$.
\end{proposition}
\begin{proof}
	See Appendix \ref{sec:proofs}.
\end{proof}

We proceed to show that from the sender's point of view, the baseline model with a sophisticated receiver is equivalent to the game with a skeptical receiver. The sender's goal in this case is to persuade the receiver that \emph{any} consistent model implies $x \Rightarrow y$. Similarly, the baseline model with a na{\"i}ve receiver is equivalent to the game with a credulous receiver. In this case, the sender only needs to find \emph{some} consistent model $M_s$ with $x \Rightarrow_s y$ to achieve a high payoff.
The formal result is captured by the following proposition.
\begin{proposition} \label{prop:ambig_equiv_str}
	The sender's equilibrium strategy in the game with a skeptical (resp., credulous) receiver coincides with the sender's optimal strategy in the baseline model with a sophisticated (resp., na{\"i}ve) receiver.
\end{proposition}
\begin{proof}
	See Appendix \ref{sec:proofs}.
\end{proof}

Proposition~\ref{prop:ambig_equiv_str} also lets us import the results of Section~\ref{sec:persuade_blank} directly into the game. In particular, by Theorem~\ref{thm:direct}, if $M_t$ is simple, then the sender's optimal proposal in the game with either receiver type is either silence or a model in which $x$ and $y$ are adjacent (see Corollary~\ref{cor:thm2-game} in Appendix~\ref{sec:ambig_pf}).

In the following section, we briefly outline how this game can be expanded to accommodate the setting with a pre-existing receiver's model set up in Section \ref{sec:preexist}.

\subsection{Game with a pre-existing model}
We conclude by sketching a natural extension of the game to a receiver with a pre-existing model, which anchors the index on the receiver's current model: $V_R(x \mid M_r) = \min_{M \in \mathcal{M}(\Omega_r)} \left\{x \left(b(M)-1\right) + \frac{1}{\theta} \psi(M \mid M_r) \right\}$. If the sender stays silent, $M_s=\varnothing$, or attempts to propose an inconsistent model, then the receiver chooses the decision $x$ that maximizes the criterion above.

The receiver would then only replace the anchor with $M_s$ if their initial model $M_r$ is debunked (in line with the AGM paradigm discussion in Section \ref{sec:discussion}). Conditional on this, a credulous receiver would behave na{\"i}vely, choosing the decision $x$ that maximizes their expected utility as if $M_s$ were the true model. A skeptical receiver would only choose action, $x=1$, if all models consistent with $P|\Omega_s$ support this, as a sophisticated agent would in Section \ref{sec:preexist}. In either case, the sender would first need to debunk the receiver's pre-existing model $M_r$, hence the analysis of Section \ref{sec:preexist} applies.

\section{Discussion} \label{sec:discussion}

In this section, we discuss various assumptions behind the model and the potential consequences of relaxing these assumptions.

\subsection{Assumptions regarding the receiver}

\paragraph{Receiver is unaware of their unawareness.}
We assume that our receiver is unaware of their own unawareness and does not even consider the existence of variables they do not observe. One could relax this assumption and make the receiver aware of the fact that there may exist other variables they do not observe. This approach raises a conceptual problem, since unobserved confounders can explain \emph{any} data. For example, any data can be explained by assuming the existence of some confounding variable $c$ that affects all other observed variables: $c \to_t x$ for all $x \in \Omega_r \setminus \{c\}$ (``God's will'' is one example of such a $c$; various conspiracy theories also rely on the existence of such unobserved---and unobservable---confounding factors). It is then unclear how the receiver should choose between multiple consistent models that rely on confounders. 

However, some conclusions can be reached by assuming the receiver allows for the existence of unobserved variables in principle, but applies Occam's razor by not relying on them unless absolutely necessary.\footnote{\citet{gottesman_limitation_2025} argues, in a different but somewhat related setting, that it is optimal for the receiver to not include unrelated variables in their model if the sender does not suggest them. This is true even if the receiver makes strategic inferences from the sender's messages given the sender's strategic concerns.} 
Using this approach, the receiver can still recover the set of causal models that are consistent with the data by using the $IC^*$ algorithm \citep[see][ch.2]{pearl_causality_2009} instead of the IC algorithm this paper applies. Using the $IC^*$ algorithm, the receiver is able to find a consistent model for each set of observable variables, and, in some cases, discover confounders on their own. 
For example, in the situation discussed in Section \ref{sec:MBA-dynamic}, if the receiver observes both $z$ and $t$, they would realize that the model produced by the IC algorithm is not consistent with the data, so there must be a confounding variable affecting $e$ and $w$. The receiver would not observe its distribution, but would become aware of its existence and would even infer the correct causal structure. Further assumptions would then need to be made about how the receiver accepts or rejects models in this scenario. 
Moreover, the $IC^*$ algorithm is only one of many approaches in the literature on causal discovery, and there is no clear consensus regarding the best way to work with situations in which some relevant variables may be unobservable \citep[see][ch. 9.4.1]{peters_elements_2017}.

\paragraph{Receiver is certain of their model.}
We assume in the baseline model that the receiver always holds at most one subjective model in mind. If their initial model becomes irreconcilable with the data, the receiver replaces it with another (singular) model. Given that the receiver's model is very likely to be incorrect---in its incompleteness at the very least---this approach violates the rational expectations assumption that is standard in economic theory. A more standard approach would be to assume the receiver is uncertain about which model of the world is correct, assigning  subjective probabilities to the possible models and updating these probabilities via Bayes' rule after observing new information from the sender. This approach is adopted by the ``narrative persuasion'' literature \citep[cf.][]{schwartzstein_using_2021,aina_tailored_2025,ispano_perils_2025}. We show in Section \ref{sec:microfound} that our results can also be seen through the lens of such model uncertainty.

When taken literally, our approach in the baseline model is instead closer to what is known in epistemology as the AGM paradigm \citep{peppas_constructive_1995}. The AGM paradigm prescribes that belief changes should follow the principle of minimal change. In line with this principle, \emph{expansion} of a subjective model of the world is the simplest change and should be employed when new data are compatible with the current set of beliefs. In our model, once the receiver has a subjective model, they first try to expand their pre-existing model $M_r$ to incorporate new data. However, when new data are incompatible with the pre-existing model, the model is \emph{revised} in some way based on the sophistication of the receiver. 

Experimental evidence suggests that our approach is closer than Bayesian updating to the thought process people employ in the real world. In particular, \cite{aina_weighting_2026} run a lab experiment in which participants must assign subjective probabilities to different models of the world that rationalize some observed evidence. They find that most participants assign subjective probability of one to the model with the best fit to the data, while only a small share of responses are consistent with Bayesian updating.\footnote{The selection criterion in their data are fit-based, in line with \citet{schwartzstein_using_2021}, rather than identification-based. However, the finding that subjects commit fully to a single model rather than mix over candidates is closer to our modeling approach. Whether people select models based on fit or identification is not explored in existing literature and is an interesting question for future work.}

\paragraph{Receiver discards the debunked model.}
In our model, the receiver's response to evidence that is incompatible with their subjective model of the world is to completely discard this model. An alternative assumption, in line with the principle of minimal revision mentioned above, could be that the receiver only discards the debunked links. They could then attempt to construct a new model that is consistent with the new evidence while remaining ``as close as possible'' in some metric to their original model. This approach would make persuasion more difficult for the sender. Our results could then be understood as an upper bound on the feasibility of persuasion in the best case for the sender. 
This approach could also resolve a gap concerning which model a receiver would have if the sender's data debunked their old model, but the sender did not offer a new consistent model to go with it. Our model takes no stance on this issue, but resolving it may be relevant for future work.

\subsection{Assumptions regarding the data}

\paragraph{Perfect data.}
Our model assumes that both players can perfectly observe the joint distribution of all variables (that they are aware of), and that this distribution is sufficiently nice (rich) in the sense that every link in the true model is identifiable in the data. This abstracts from decades of econometric work and centuries of statistical work, which together deal with an immense variety of issues that arise in statistical estimation of such joint distributions. This abstraction is intentional, since we explore the issues that are specific to persuasion and selective disclosure. However, the usual problems of statistical estimation remain in the real world. Incorporating them in our model would lead to agents having less confidence in subjective models. On the one hand, this may make persuasion more difficult, especially with sophisticated receivers, since the sender would have a harder time debunking pre-existing models with limited and imperfect data, and may struggle to find enough evidence supporting their proposed model. But on the other hand, the receivers' pre-existing models would also be based on limited and imperfect data, and as a consequence, they may be easier to debunk in this case than in our benchmark---including, e.g., by offering new data on the variables already known to the receiver.

\paragraph{No interventions.}
We consider learning from observational data, where the receiver tries to recover the true causal graph by merely looking at existing data, and the sender can only provide more of the existing data (other variables). Specifically, we assume that neither the sender nor the receiver can run experiments, directly manipulating some of the variables and thereby generating additional data points. Allowing for such interventions could allow the sender to more easily demonstrate some truthful links, as well as potentially ease the receiver's causal discovery problem. The emerging causal-discovery literature explores optimal design of interventions to learn the causal graph (see \citealp[Section 5]{zanga_survey_2022}; \citealp{ChooShiragur2023,ShiragurZhangUhler2023}), but communication in such a setting is yet to be explored.

On the flip side, allowing the \emph{receiver} to run experiments could also mitigate the risk of being deceived by the sender. The issue is that if the whole distribution is formed by the receiver's choices and the realized consequences, underexperimentation could lead to poor data and learning traps (for some examples, see \citealp{eliaz_model_2020,gratton_bad_2025,spiegler2026bci}). Overall, we view the issues related to optimal learning and persuasion with interventions as an appealing avenue for future research.

\section{Conclusion}

In this paper, we propose a model of causal persuasion. Our main contribution is to develop a novel tractable model of strategic communication of causal models that explicitly deals with how causation can be established and debunked.
We show that unless the receiver is na{\"i}ve, persuading them of a causal link requires disclosing additional variables. The task of persuading a sophisticated receiver with no pre-existing model is also closely related to debunking a pre-existing model that a (na{\"i}ve or sophisticated) receiver may have.

Our model emphasizes the asymmetry in communicating causal models. We show that it is easier for the sender to persuade the receiver that a particular causal link exists than to persuade them that it does not. Persuasion difficulty in this case is measured by the number of causal variables that need to be revealed by the sender to debunk the receiver's subjective causal model. 

Our model is designed to make persuasion as simple as possible for the sender, while restricting them to disclosure of only truthful data. We show that deception nonetheless remains possible in this case, with the sender being able to spin the correlation generated by confounding variables as causation, even if the receiver is sophisticated. Further, the sender may be able to debunk the receiver's defective (incorrect) model and replace it with another defective model. However, if a sender wants to persuade a sophisticated receiver that no causal link exists between two variables and any correlation between them is spurious, then this is possible only when the true model also contains no such relation. Even then, the absence of a link may be difficult to prove---more difficult than flipping the perceived direction of a link.

The theory of causal persuasion we propose in this paper opens up a vast array of questions for future research. The discussion in Section \ref{sec:discussion} outlines some of the assumptions that could be relaxed or reinterpreted in future work. These include allowing the sender to generate new data (whether genuine---generated through interventions---or fabricated), allowing the receiver's prior belief over models to be nondegenerate, or allowing the receiver to retain part of their subjective model when faced with conflicting evidence. Additionally, while we touch on competition between senders in the applications in Section \ref{sec:applications}, exploring this topic in more detail appears to be an interesting direction for future research.

\appendix

\section{Proofs} \label{sec:proofs}

\subsection{Technical graph terminology}

We begin by introducing a few more helpful definitions.

A \emph{path} is a sequence of distinct adjacent vertices, while a \emph{walk} is allowed to repeat vertices and edges. Given a path or walk, an interior occurrence of a vertex is a \emph{collider occurrence} if the two incident edges both have arrowheads pointing into that occurrence; every other interior occurrence is a \emph{non-collider occurrence}.
For example, the DAG in Figure~\ref{fig:pic3}(a) contains a walk $e \leftarrow a \to w \leftarrow s \to w \leftarrow t$, where both incidences of $w$ are collider occurrences, and $a,s$ are non-collider occurrences.
To distinguish this path-based terminology from the main text, we use ``collider occurrence'' whenever discussing a particular path or walk. ``Collider'' without qualification continues to mean an unshielded collider, as in the main text.

A path is \emph{active} given a conditioning set $S$ if every non-collider occurrence lies outside $S$ and every collider occurrence is an ancestor of some variable in $S$, possibly itself. Definition~\ref{def:dsep} states that the set $S$ \emph{d-separates} two variables if $S$ blocks every path between them---i.e., there is no path between them that is active given $S$. We say that $a,b$ are \emph{d-connected} given $S$ if $S$ does not d-separate them. 
In what follows, it is more convenient to use the following equivalent definition of d-separation stated in terms of walks instead of paths (given by, e.g., \citealp[\S2]{van_der_zander_separators_2019}).

\begin{definition}[d-connectedness]\label{def:dsep2}
	Let $M$ be a DAG over $\Omega$, and let $a,b\in\Omega$ and $S\subseteq\Omega\setminus\{a,b\}$ be disjoint. Then $a$ and $b$ are \emph{d-connected} given $S$ in $M$ if and only if there exists a walk between $a$ and $b$ in $M$ on which every collider occurrence belongs to $S$ and every non-collider occurrence lies outside $S$.
\end{definition}

Note that as a consequence of Definition~\ref{def:dsep2}, the existence of a \emph{collider-free} path between $a$ and $b$ establishes that $a$ and $b$ are d-connected given $S=\emptyset$. Further, every such collider-free path has the \emph{fork shape} $a \Leftarrow s \Rightarrow b$ for some summit vertex $s$ (possibly coinciding with $a$ or $b$).

Next, we introduce maximal ancestral graphs (MAGs) that represent projections of the true model $M_t$ on the set of observed variables $\Omega_s$; see \citet{richardson_ancestral_2002,zhang_completeness_2008} for further discussion of these and related concepts.

\begin{definition}[MAG]
	A \emph{maximal ancestral graph (MAG) $L$} is a mixed graph (consisting of undirected $a-b$, directed $a\to b$, and bidirected $a \leftrightarrow b$ links) over variables $\Omega$ such that $\forall a,b \in \Omega$:
	\begin{enumerate}[nosep]
		\item there are no directed cycles (consisting of directed links only),
		\item if $a \leftrightarrow b$ then there does not exist a path $a \Rightarrow b$ (consisting of directed links only), 
		\item if $a-b$ then $\nexists c \in \Omega$ s.t. $c \to a$ or $c \leftrightarrow a$, and
		\item if $a,b$ are not adjacent (i.e., not $a-b$ and not $a \to b$ and not $a \leftrightarrow b$) then there exists a set that m-separates them.\footnote{M-separation generalizes Definition~\ref{def:dsep}/Definition~\ref{def:dsep2} to mixed graphs: A path is active given $S$ if every noncollider occurrence lies outside $S$ and every collider occurrence is an ancestor of a vertex in $S$; otherwise the path is blocked. Variables $a,b$ are m-separated by $S$ if $S$ blocks every path between $a,b$.}
	\end{enumerate}
\end{definition}

In general, the undirected edges in MAGs are used to deal with selection on unobserved variables. In our setting, such selection is never present, hence all MAGs contain only directed and bidirected edges, but no undirected edges. (\citet{zhang2005transformational} refer to such graphs as Directed MAGs or DMAGs.)

\begin{definition}[Latent projection]\label{def:proj}
	\emph{Latent projection} $L(\Omega_s) \equiv \mathrm{proj}(M_t, \Omega_s)$ of $M_t$ on $\Omega_s$ is a MAG such that $\forall a,b \in \Omega_s$:
	\begin{enumerate}[nosep]
		\item $a,b$ are adjacent in $L(\Omega_s)$ if and only if there is an inducing path relative to $\Omega_t \setminus \Omega_s$ between them in $M_t$---where a path between $a,b$ is called an \emph{inducing path relative to $S \subset \Omega$} if every non-endpoint variable in this path is either in $S$, or a collider occurrence, and every collider occurrence on this path is an ancestor of either $a$ or $b$ \citep{zhang_completeness_2008}; and
		\item if $a,b$ are adjacent in $L(\Omega_s)$, then: $a \to_{L} b$ iff $a \Rightarrow_t b$, and $a \leftrightarrow_{L} b$ iff $a \not\Rightarrow_t b$ and $a \not\Leftarrow_t b$ (but $a,b$ are correlated in $M_t$).
	\end{enumerate}
\end{definition}
The inducing-path criterion above is related to, but distinct from the walk-based d-separation criterion in Defition \ref{def:dsep2}: the latter characterizes conditional d-connection in a DAG, whereas the former characterizes adjacency after latent variables are marginalized.

\subsection{Proof of Theorem~\ref{thm:no_persuade_defective}}

As before, $\mathcal{M}(\Omega_s)$ denotes the set of models (DAGs) consistent with $P|\Omega_s$. Recall that a model $M$ is consistent if and only if for any $a,b \in \Omega_s$ and $S \subset \Omega_s \setminus \{a,b\}$, $S$ d-separates $a$ and $b$ in $M$ if and only if $(a \perp b \mid S)$.
The following is a known result establishing that latent projections are also consistent with the observed data; see \citet[Theorems 4.18 and 6.4]{richardson_ancestral_2002} and \citet[Proposition 3.6]{ali_markov_2009}.

\begin{lemma} \label{lem:membership}
	The latent projection $L(\Omega_s)$ of $M_t$ on $\Omega_s$ faithfully represents the data: for any $a,b \in \Omega_s$ and $S \subset \Omega_s \setminus \{a,b\}$, $S$ m-separates $a$ and $b$ in $L(\Omega_s)$ if and only if $(a \perp b \mid S)$.
	
	If $\mathcal{M}(\Omega_s) \neq \emptyset$, then $L(\Omega_s)$ is Markov-equivalent to every $M \in \mathcal{M}(\Omega_s)$: for any $a,b \in \Omega_s$ and $S \subset \Omega_s \setminus \{a,b\}$, $S$ d-separates $a$ and $b$ in $M$ if and only if $S$ m-separates $a$ and $b$ in $L(\Omega_s)$. In particular, $L(\Omega_s)$ and $M$ have the same sets of adjacencies and (unshielded) colliders.
\end{lemma}

We also introduce the CPDAG as the standard way to describe the set of consistent models $\mathcal{M}(\Omega_s)$ \citep[cf. ][Definitions 2.18--2.23]{zanga_survey_2022}. Any DAG $M$ that agrees with (the skeleton and the oriented links of) the CPDAG $K(\Omega_s)$ is consistent with the observed data $P|\Omega_s$. 

\begin{definition}[CPDAG]
	A \emph{completed partially-directed acyclic graph} (CPDAG) $K(\Omega_s)$ is a graph over variables $\Omega_s$ s.t. $\forall a,b \in \Omega_s$:
	\begin{enumerate}[nosep]
		\item $a$ and $b$ are adjacent in $K(\Omega_s)$ if and only if they are adjacent in every $M \in \mathcal{M}(\Omega_s)$;
		\item $a \to b$ in $K(\Omega_s)$ if and only if $a \to_M b$ for every $M \in \mathcal{M}(\Omega_s)$;
		\item if $K(\Omega_s)$ is such that $a,b$ are adjacent and not $a \to b$ and not $a \leftarrow b$, then $a,b$ are connected by an undirected link.
	\end{enumerate}
\end{definition}

\begin{figure}
	\centering
	
	\begin{minipage}[b]{0.32\textwidth}
		\centering
		
		\includegraphics[scale=0.3]{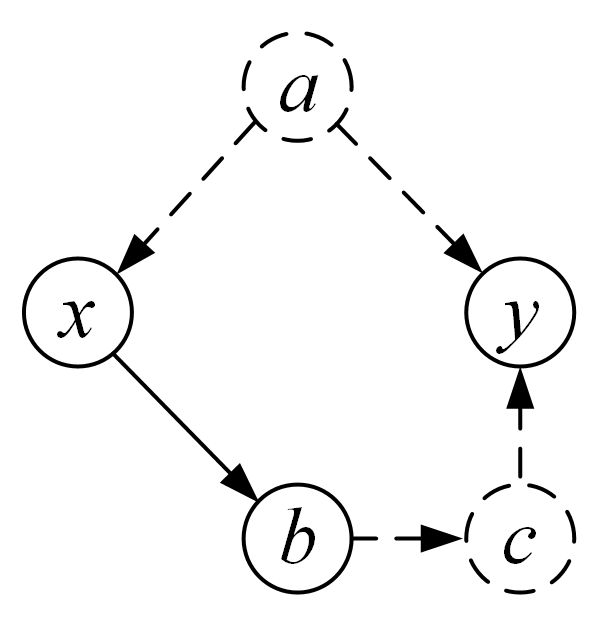}\\
		(a) True model $M_t$\\
	\end{minipage}
	\hfill
	\begin{minipage}[b]{0.32\textwidth}
		\centering
		
		\includegraphics[scale=0.3]{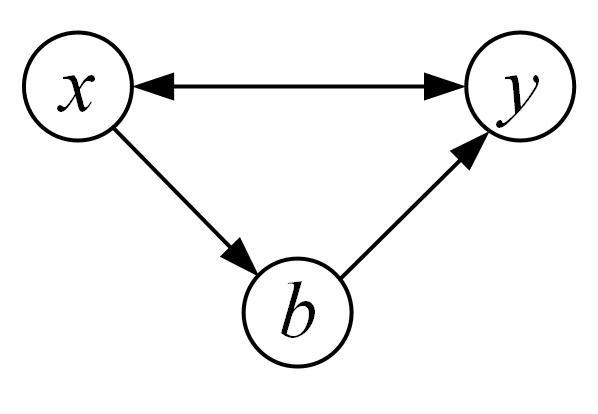}\\
		(b) Latent projection $L(\Omega_s)$\\
	\end{minipage}
	\hfill
	\begin{minipage}[b]{0.32\textwidth}
		\centering
		
		\includegraphics[scale=0.3]{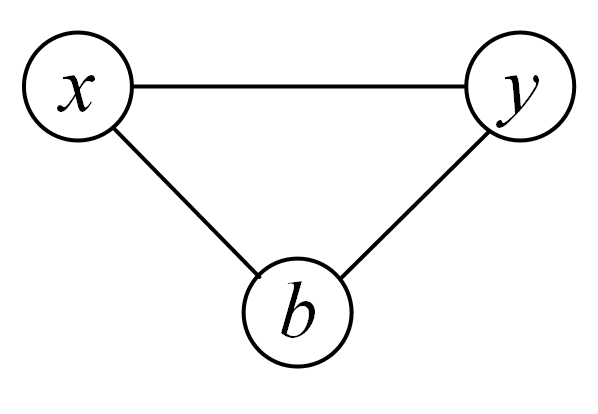}\\
		(c) CPDAG $K(\Omega_s)$\\
	\end{minipage}
	
	\caption{Latent projections and CPDAGs \label{fig:MAG_CPDAG}}
\end{figure}

To summarize, the true model $M_t$ generates a latent projection (MAG) $L(\Omega_s)$ on the set of observed variables $\Omega_s$. The true model also generates the data $P|\Omega_s$ from which we can recover the set of models consistent with the data, $\mathcal{M}(\Omega_s)$. This set can be represented as a CPDAG $K(\Omega_s)$. 
MAG $L(\Omega_s)$ tells us how the causal graph actually looks once we account for missing variables, and CPDAG $K(\Omega_s)$ tells us how much of the causal graph we can recover from the data.
An example is presented in Figure~\ref{fig:MAG_CPDAG}: MAG $L(\Omega_s)$ preserves directed ancestry through omitted proxy variables and represents omitted confounders by bidirected links; CPDAG $K(\Omega_s)$ orients links using V-structures, the requirement not to create additional unshielded colliders, and acyclicity; in this example, it is therefore completely unoriented.

Our goal now is to show that $L(\Omega_s)$, $K(\Omega_s)$ and $M_t$ do not contradict each other. 
The following lemma is an immediate consequence of the loyal-equivalent-graph result of \citet{zhang2005transformational}, which is key to proving Theorem~\ref{thm:no_persuade_defective}. It argues that the latent projection must be represented among the consistent models: there is an $M \in \mathcal{M}(\Omega_s)$ that agrees with every directed link in $L(\Omega_s)$.

\begin{lemma}\label{lem:realizability}
	If $\mathcal M(\Omega_s)\neq\emptyset$, there exists $M_q\in\mathcal M(\Omega_s)$ preserving every directed link of the latent projection $L(\Omega_s)=\operatorname{proj}(M_t,\Omega_s)$: for any $a,b\in\Omega_s$, if $a\to_L b$, then $a\to_q b$.
\end{lemma}
\begin{proof}
   By Lemma~\ref{lem:membership}, the latent projection $L(\Omega_s)$ is Markov-equivalent to every $M\in\mathcal M(\Omega_s)$. Proposition~2 of \citet{zhang2005transformational} establishes that there exists a ``loyal equivalent graph'' of $L(\Omega_s)$, denoted as $H$. This $H$ is a directed MAG (so it has no undirected links), preserves every directed edge of $L(\Omega_s)$, and all bi-directed links in $H$ are invariant, meaning they are present in all MAGs Markov-equivalent to $L(\Omega_s)$. 
	
	Because $\mathcal M(\Omega_s)\neq\emptyset$, the Markov-equivalence class of $L(\Omega_s)$ contains a DAG. No bidirected edge can therefore be invariant in this class, since a DAG cannot have any bidirected edges. It follows that $H$ has no bidirected edges and is itself a DAG. Since $H$ is Markov equivalent to $L(\Omega_s)$, Lemma~\ref{lem:membership} implies that $H\in\mathcal M(\Omega_s)$. Taking $M_q=H$ proves the result. 
\end{proof}

\begin{lemma} \label{lem:exists_modconswtrue}
	If $x \Leftarrow_t y$ and $\mathcal{M}(\Omega_s)\neq\emptyset$, then there exists a model $M^*\in\mathcal{M}(\Omega_s)$ such that $x \Leftarrow_{M^*} y$.
\end{lemma}
\begin{proof}
	If the true model contains a path $x \Leftarrow_t y$ and $x,y \in \Omega_s$, then the latent projection $L(\Omega_s)$ also contains a path $x\Leftarrow_{L} y$ \citep{richardson_ancestral_2002}. By Lemma~\ref{lem:realizability}, if $\mathcal{M}(\Omega_s)\neq\emptyset$ then there exists a consistent model $M^*\in\mathcal{M}(\Omega_s)$ that includes this path.
\end{proof}

Theorem~\ref{thm:no_persuade_defective} is then a direct corollary of the lemma above.
\begin{proof}[Proof of Theorem~\ref{thm:no_persuade_defective}]
	Persuading a sophisticated receiver that $x \Rightarrow_s y$ requires presenting a model $M_s$ such that $x \Rightarrow_M y$ for all $M \in \mathcal{M}(\Omega_s)$. Lemma~\ref{lem:exists_modconswtrue} implies that for every feasible disclosure $\Omega_s$---that is, every $\Omega_s$ with $\mathcal M(\Omega_s)\neq \emptyset$---there exists $M^*\in\mathcal M(\Omega_s)$ with $x\Leftarrow_{M^*}y$.
\end{proof}

\subsection{Proof of Theorem~\ref{thm:direct}}

As in the main text, for a DAG $D$ over $\Omega$ and $a\in\Omega$, we write $C_D(a)$ for the parents of $a$, $C_D^{-1}(a)$ for its children, and $\bar{C}_D(a)$ for its strict ancestors. Additionally, define
\[
\mathrm{An}_D(a)\equiv \bar{C}_D(a)\cup\{a\},
\qquad
\mathrm{De}_D(a)\equiv \{b\in\Omega : a\Rightarrow_D b\}\cup\{a\}
\]
for the \emph{self-inclusive} ancestor and descendant sets. Further, $\mathrm{ne}_D(a)$ denotes the set of variables adjacent to $a$ in $D$ (neighbors), a \emph{sink} is a variable without children, a \emph{source} is a variable without parents, and $D-z$ is the subgraph of $D$ induced on $\Omega\setminus\{z\}$.
Finally, we use $D \simeq D'$ to denote observationally equivalent DAGs, i.e., $D \simeq D'$ if and only if $D,D' \in \mathcal{M}(\Omega)$ for some $\Omega$.

\begin{definition}[Positive set]\label{def:positive}
	A disclosed set $\Omega$ with $\{x,y\}\subseteq\Omega\subseteq\Omega_t$ is:
	\begin{itemize}[nosep]
		\item \emph{positive} if $\mathcal{M}(\Omega)\neq\emptyset$ and $x\Rightarrow_M y$ for every $M\in\mathcal{M}(\Omega)$;
		\item \emph{minimal} if $|\Omega| \leq |\Omega'|$ for any other positive $\Omega' \subseteq \Omega_t$;
		\item \emph{direct} if $x \in ne_M(y)$ for some (equivalently, every) $M \in \mathcal{M}(\Omega)$;
	\end{itemize}
\end{definition}

Positivity of a disclosure is simply a shorthand for the sophisticated receiver's acceptance criterion for a proposal with $x\Rightarrow_s y$. (By Proposition~\ref{prop:ambig_optactions}, this is equivalent to inducing an action from a skeptical receiver in the game of Section \ref{sec:microfound}.)
Using this terminology, the sender's problem with a sophisticated receiver is to find a minimal positive disclosure. Our goal in this proof is to show that such minimal $\Omega^*$ is necessarily direct. The proof proceeds in the following steps:
\begin{enumerate}[nosep]
	\item Section~\ref{app:simple} shows that simplicity of $M_t$ implies that no $M \in \mathcal{M}(\Omega)$ for any $\Omega \subset \Omega_t$ can have a particular structure: a cycle of length at least four that has exactly one collider and no chord (no link between nonconsecutive vertices). Lemmas~\ref{lem:tri_ancestors}--\ref{lem:one-collider} establish this restriction.
	
	\item Section~\ref{app:peel} demonstrates that removing descendants of $x$ and $y$ from $\Omega$ does not impact positivity, hence a minimal $\Omega^*$ does not include them, and has $y$ as the only sink (Lemma~\ref{lem:min-sink}). 
	
	\item Section~\ref{app:min} shows that in simple $M_t$, all proxies $v$ s.t. $x \Rightarrow v \to y$ can also be removed without violating positivity, since otherwise $v$ must be a part of the structure ruled out in Section~\ref{app:simple}.
	
	\item Section~\ref{app:proof2} collates the results and completes the proof.
\end{enumerate}

\subsubsection{Structures prohibited in simple worlds}\label{app:simple}

\begin{lemma} \label{lem:tri_ancestors}
	If $M_t=(\Omega_t,C_t)$ is simple then for any distinct $y,a,b \in \Omega_t$ such that $a,b \in C_t(y)$: either $a \perp b$, or $a$ is adjacent to $b$ in $M_t$.
	
	Further, for any distinct $y,a,b \in \Omega_t$ such that $a,b \in \bar C_t(y)$: either $a \perp b$, or $a \in \mathrm{An}_t(b)$, or $b \in \mathrm{An}_t(a)$.
\end{lemma}
\begin{proof}
	Begin with the first part. If $a,b$ are not adjacent, then $a \to_t y \leftarrow_t b$ is an unshielded collider which forms a V-structure in the data. Then simplicity of $M_t$ requires that $a \perp b$. 
	
	To show the second part, suppose $a \not\perp b$. Then we need to show that either $a \Rightarrow_t b$, or $a \Leftarrow_t b$.
	Suppose not---then $a,b$ are non-adjacent and there is no directed path between them, but $a \not\perp b$ implies (by the Markov property) that there exists a collider-free path between them.
	
	Consider the shortest walk \emph{with a collider} between $a$ and $b$: $a \Rightarrow_t c \Leftarrow_t b$. (We know at least one such walk exists because $a,b \in \bar{C}_t(y) \iff a \Rightarrow_t y \Leftarrow_t b$.) We claim this walk must have length two, i.e., $c$ must be adjacent to both $a$ and $b$. Toward a contradiction, suppose there exist variables $d$ distinct from $a,b,c$ and $e$ distinct from $a,c$ such that $a \Rightarrow_t d \to_t c \leftarrow_t e \Leftarrow_t b$. Then the collider-free path between $a$ and $b$ implies there exists a collider-free path between $d$ and $e$, meaning $d \not\perp e$. Then by part 1 of this lemma, $d$ and $e$ must be adjacent (either $d \to_t e$, or $d \leftarrow_t e$), which contradicts $a \Rightarrow_t c \Leftarrow_t b$ being the shortest walk with a collider. 
	
	We conclude that there exists $c$ such that $a \to_t c \leftarrow_t b$. Then part 1 of this lemma implies that since $a \not\perp b$, $a$ and $b$ must be adjacent---a contradiction.
\end{proof}

The following lemma is needed to draw a chord in a cycle with one sink in the projection.

\begin{lemma}[Monotone marginal adjacency]\label{lem:marg-adj}
	Suppose $M_t$ is simple, and let $O\subseteq\Omega_t$ and $a,r,b\in O$ satisfy $a\Rightarrow_t r\Rightarrow_t b$. If $a$ and $b$ are adjacent in the latent projection $L(O)$, then so are $a$ and $r$.
\end{lemma}

\begin{proof}
If $r\in\{a,b\}$ the conclusion is immediate, so let $r\notin\{a,b\}$. Our goal is to show that there exists an inducing path between $a$ and $r$.

By Definition~\ref{def:proj}, the adjacency of $a$ and $b$ in $L(O)$ is witnessed by an inducing path $\pi$ between them relative to $\Omega_t\setminus O$: every interior vertex of $\pi$ that is not a collider occurrence lies outside $O$, and every collider occurrence on $\pi$ is an ancestor of $a$ or $b$. If $r$ occurs as an interior vertex of $\pi$, it must be a collider occurrence, otherwise $a$ and $b$ could not be adjacent. Call a collider occurrence of $\pi$ \emph{bad} if it does not belong to $\mathrm{An}_t(r)$, and set
\[
e \equiv \text{the first bad collider occurrence along $\pi$ from $a$,
or } e\equiv b \text{ if there is none.}
\]

Intuitively, a collider occurrence matters only through its downstream effect on an endpoint; a bad one affects only $b$, and so becomes irrelevant once we seek a source of correlation between $a$ and $r$ rather than between $a$ and $b$.

Suppose first that $r$ occurs on $\pi$ strictly before $e$. Then the $a$--$r$ part of $\pi$ is itself an inducing path for the pair $(a,r)$: its interior non-collider-occurrences lie outside $O$, and its collider occurrences precede $e$, hence lie in $\mathrm{An}_t(r)$ and are in particular ancestors of the endpoint $r$. Definition~\ref{def:proj} then makes $a,r$ adjacent in $L(O)$, as required. For the remainder of the proof, suppose $r$ does not occur on $\pi$ before $e$.

\medskip
\noindent\emph{Claim.} Suppose $h\Rightarrow_t r$ and $r\Rightarrow_t e$, and let $P$ be a directed path from $h$ to $e$ that does not pass through $r$. Then some vertex $p\in P\setminus\{e\}$ satisfies $p\to_t r$.

\medskip
\noindent\emph{Proof of the Claim.} The mechanism: simplicity repeatedly
shields the collider occurrence at the vertex where a directed path from
$r$ first meets $P$, pushing an arrow from $P$ backwards down the
$r$-path until it lands on $r$ itself (Figure
\ref{fig:monotone-claim}).

Among directed paths $Q: r\Rightarrow_t e$, choose one whose first intersection $m$ with $P$ occurs as early as possible along $P$, and let $Q$ follow $P$ after $m$; an intersection exists because $e\in P\cap Q$. By acyclicity, $m\neq h$,  and $m\neq r$ because $r\notin P$. Let $p$ and $q$ be the ancestors of $m$ on $P$ and $Q$; they are distinct, for a common ancestor would be an earlier intersection. Thus $p\to_t m\leftarrow_t q$. The vertex $h$ is an ancestor of both $p$ and $q$, so they are correlated, and by Lemma~\ref{lem:tri_ancestors} they are adjacent. If $q\to_t p$, then paths $Q$ and $P$ meet strictly earlier than $m$, a contradiction. Hence $p\to_t q$.

Denote the vertices preceding $q$ on $Q$ as $r=q_0\to_t\cdots\to_t q_k=q$ and propagate the argument backward. If $p\to_t q_i$ with $i\geq 1$, the triple $p\to_t q_i\leftarrow_t q_{i-1}$ has correlated parents, because $h$ is an ancestor of both. Lemma~\ref{lem:tri_ancestors} makes $p,q_{i-1}$ adjacent, and by the above reasoning we get  $p\to_t q_{i-1}$. Induction terminates with $p\to_t q_0=r$. Finally, $p\neq e$ because $p$ has the successor $m$ on $P$. \hfill$\square$

\begin{figure}[t]
	\centering
	\begin{tikzpicture}[>=stealth,
			every node/.style={circle, draw, inner sep=1.5pt, minimum size=5.5mm}]
		\node (h) at (0,1.6)   {$h$};
		\node (p) at (2.6,0.8) {$p$};
		\node (m) at (4.4,0.8) {$m$};
		\node (e) at (6.2,0.8) {$e$};
		\node (r) at (0,-0.7)  {$r$};
		\node (q) at (2.6,-0.7){$q$};
		\draw[->, double] (h) -- (p);
		\draw[->]         (p) -- (m);
		\draw[->, double] (m) -- (e);
		\draw[->, double] (h) -- (r);
		\draw[->, double] (r) -- (q);
		\draw[->]         (q) -- (m);
		\draw[->, dashed] (p) -- (q);
		\draw[->, dashed] (p) to[bend left=18] (r);
	\end{tikzpicture}
	\caption{The configuration in the Claim within the proof of Lemma~\ref{lem:marg-adj}. Double arrows denote directed paths, single arrows denote links. Simplicity shields the collider occurrence at $m$, forcing $p\to_t q$ (dashed); iterating down the prefix $r\Rightarrow_t q$ yields $p\to_t r$ (dashed).}
	\label{fig:monotone-claim}
\end{figure}

Returning to the main argument, let $\gamma$ be the collider occurrence of $\pi$ immediately preceding $e$, or $\gamma \equiv a$ if there is none. The $\gamma$--$e$ segment of $\pi$ contains no interior collider occurrences, so it is fork-shaped, $\gamma\Leftarrow_t h\Rightarrow_t e$, for some summit $h$ (with possibly degenerate branches). Let $P: h\Rightarrow_t e$ denote its $e$-branch, a subpath of $\pi$. We verify the hypotheses of the Claim.

First, $h\Rightarrow_t r$. If $\gamma=a$, then $h\Rightarrow_t a\Rightarrow_t r$. If $\gamma\neq a$, then $\gamma$
precedes $e$ and is therefore not bad, so $\gamma\in\mathrm{An}_t(r)$; since $r$ does not occur before $e$, $\gamma\neq r$, hence $\gamma\Rightarrow_t r$ and $h\Rightarrow_t\gamma\Rightarrow_t r$.

Second, $r\Rightarrow_t e$. If $e=b$, this is the hypothesis $r\Rightarrow_t b$. Otherwise $e$ is a bad collider occurrence of the inducing path $\pi$, hence $e\Rightarrow_t b$. Now $h$ is a common ancestor of $r$ and $e$, so the
two are correlated, and both lie in $\mathrm{An}_t(b)$. Lemma~\ref{lem:tri_ancestors} makes them comparable, while badness excludes both $e=r$ and $e\Rightarrow_t r$; hence $r\Rightarrow_t e$.

Third, $r\notin P$ and $r\neq h$. The interior vertices of the $e$-branch are non-collider occurrences of $\pi$ and thus lie outside $O\ni r$; the endpoint $e$ differs from $r$ as just shown. As for the summit, there are three candidates: an interior non-collider occurrence of $\pi$ (outside $O$, hence not $r$); the endpoint $a$, when $\gamma=a$ and the $\gamma$-branch is empty (and $a\neq r$); or, only when $e=b$, the endpoint $b$ itself---excluded because $h=b$ together with $h\Rightarrow_t r$ and $r\Rightarrow_t b$ would contradict acyclicity.

\begin{figure}[t]
	\centering
	\begin{tikzpicture}[
		obs/.style={circle, draw, inner sep=1.5pt, minimum size=5.5mm},
		lat/.style={obs, dashed},
		spc/.style={circle, draw, dashed, inner sep=0pt, minimum size=2.8mm},
		lnk/.style={->, >=stealth},
		dpath/.style={->, >=stealth, double, double distance=1.2pt},
		splice/.style={->, >=stealth, line width=1.2pt}
		]
		\node[obs] (a)  at (0,0)      {$a$};
		\node[spc] (s1) at (1.1,0)    {};
		\node[lat] (g)  at (2.2,0)    {$\gamma$};
		\node[spc] (s2) at (3.3,0)    {};
		\node[lat] (h)  at (4.4,0.6)  {$h$};
		\node[lat] (p)  at (5.5,0)    {$p$};
		\node[spc] (s3) at (6.6,0)    {};
		\node[lat] (e)  at (7.7,0)    {$e$};
		\node[obs] (b)  at (8.8,0)    {$b$};
		\node[obs] (r)  at (4.4,-1.4) {$r$};
		\node[font=\scriptsize] at (2.2,0.55) {good};
		\node[font=\scriptsize] at (7.7,0.55) {bad};
		\draw[lnk] (s1) -- (a);
		\draw[lnk] (s1) -- (g);
		\draw[lnk] (s2) -- (g);
		\draw[lnk] (h)  -- (s2);
		\draw[lnk] (h)  -- (p);
		\draw[lnk] (p)  -- (s3);
		\draw[lnk] (s3) -- (e);
		\draw[lnk, gray] (b) -- (e);
		\node[font=\scriptsize, gray] at (8.25,-0.45) {rest of $\pi$};
		\draw[dpath] (h) -- (r);
		\draw[dpath] (g) to[bend right=30] (r);
		\draw[dpath] (r) to[bend right=25] (e);
		\draw[splice] (p) -- (r);
		\draw[dashed, gray] (6.05,0.35) -- (6.05,-0.35);
		\node[font=\scriptsize, gray] at (6.05,-0.6) {cut};
	\end{tikzpicture}
	
	\smallskip
	(a) before: the inducing $a$--$b$ path $\pi$
	
	\bigskip
	\begin{tikzpicture}[
		obs/.style={circle, draw, inner sep=1.5pt, minimum size=5.5mm},
		lat/.style={obs, dashed},
		spc/.style={circle, draw, dashed, inner sep=0pt, minimum size=2.8mm},
		lnk/.style={->, >=stealth},
		dpath/.style={->, >=stealth, double, double distance=1.2pt},
		splice/.style={->, >=stealth, line width=1.2pt}
		]
		\node[obs] (a)  at (0,0)     {$a$};
		\node[spc] (s1) at (1.1,0)   {};
		\node[lat] (g)  at (2.2,0)   {$\gamma$};
		\node[spc] (s2) at (3.3,0)   {};
		\node[lat] (h)  at (4.4,0.6) {$h$};
		\node[lat] (p)  at (5.5,0)   {$p$};
		\node[obs] (r)  at (6.6,0)   {$r$};
		\draw[lnk] (s1) -- (a);
		\draw[lnk] (s1) -- (g);
		\draw[lnk] (s2) -- (g);
		\draw[lnk] (h)  -- (s2);
		\draw[lnk] (h)  -- (p);
		\draw[splice] (p) -- (r);
		\draw[dpath] (g) to[bend right=35] (r);
		\node[gray, font=\scriptsize] at (8.1,-0.5) {discarded: everything past $p$};
	\end{tikzpicture}
	
	\smallskip
	(b) after: the spliced inducing $a$--$r$ path $\pi'$
	\caption{The splice in the proof of Lemma~\ref{lem:marg-adj}. In panel (a), $p$ is the splice point and the suffix of $\pi$ after $p$ is removed. Panel (b) replaces that suffix with the link $p \to_t r$, producing the inducing $a$--$r$ path $\pi'$. Plain arrows denote links of $\pi$, double arrows denote directed paths in $M_t$, and dashed outlines mark variables outside $O$. The portion of $\pi$ beyond $e$ is schematic.}
	\label{fig:splice}
\end{figure}

Hence, there exists $p\in P\setminus\{e\}$ with $p\to_t r$. Replace the portion of $\pi$ from $p$ onward by this single link, obtaining a path $\pi'$ from $a$ to $r$; as illustrated in Figure~\ref{fig:splice}. The status of $p$ is unchanged: the old link out of $p$ along $P$ and the new link $p\to_t r$ both carry tails at $p$. Every retained interior non-collider occurrence lies outside $O$ (for $p$ itself, if interior: it is a non-collider occurrence of $\pi$ on the $e$-branch, hence outside $O$, unless $p=h=a$, in which case $\pi'$ is the single link $a\to_t r$); every retained collider occurrence precedes $e$ and is therefore an ancestor of the endpoint $r$; and $r$ does not occur among the retained vertices, since it does not occur on $\pi$ before $e$. Thus $\pi'$ is an inducing path between $a$ and $r$ relative to $\Omega_t\setminus O$, so $a,r$ are adjacent in $L(O)$.
\end{proof}

\begin{lemma}[One-collider cycles are excluded]\label{lem:one-collider}
	Suppose $M_t$ is simple and $\mathcal{M}(O)\neq\emptyset$ for some $O\subseteq\Omega_t$. Then in every consistent model $M\in\mathcal{M}(O)$, every cycle of length four or more with exactly one collider occurrence has a chord.
\end{lemma}

\begin{proof}
	Suppose such a cycle exists in some $M\in\mathcal{M}(O)$; denote its unique collider occurrence by $z$ and the cycle-neighbours of $z$ by $p$ and $q$, so that $p\to_M z\leftarrow_M q$, and let $R$ be the remainder of the cycle: the chordless $p$--$q$ path avoiding $z$, all of whose interior vertices are non-colliders in the cycle.
	
	Let $L \equiv L(O)$. By Lemma~\ref{lem:membership}, $L$ and $M$ have the same adjacencies and the same colliders. Because the cycle is chordless and has length at least four, $p$ and $q$ are nonadjacent, so $p\to z\leftarrow q$ is a collider of $M$ and hence of $L$. Every other cycle vertex is the center of an unshielded non-collider triple of $M$, and the same status transfers to $L$.
	
	We claim that in fact
	\begin{equation}\label{eq:tails}
		p\to_L z\leftarrow_L q ,
	\end{equation}
	that is, the $p$- and $q$-ends carry tails. Suppose the $p$-end of the $p$--$z$ link carried an arrowhead, implying that $p\not\Rightarrow_t z$. Since $L$ contains only directed and bidirected links, the non-collider triple centered at $p$ forces a tail at $p$ on its other cycle link: $p\to_L w_1$, where $w_1$ is $p$'s neighbour on $R$. The triple centered at $w_1$ then has an arrowhead into $w_1$ and, being a non-collider, forces $w_1\to_L w_2$; the argument propagates a directed $L$-path $p\to_L\cdots\to_L q$ along $R$, and the non-collider triple at $q$ finally forces $q\to_L z$. Hence, $p\Rightarrow_t z$---a contradiction. The argument at $q$ is symmetric, proving \eqref{eq:tails}.
	
	Hence we have $p,q\in\mathrm{An}_t(z)$. The path $R$ is free of collider occurrences, so $p\not\perp q$ and the two are correlated in $M_t$. Lemma~\ref{lem:tri_ancestors} makes them comparable; without loss of generality, $p\Rightarrow_t q\Rightarrow_t z$. Since $p$ and $z$ are adjacent in $M$ and hence in $L$, by Lemma~\ref{lem:marg-adj} we get $p$ and $q$ adjacent in $L$ and hence in $M$---a chord of the cycle, a contradiction.
\end{proof}

\subsubsection{Removing other sinks}\label{app:peel}

\begin{lemma}[Sink attachment]\label{lem:attach}
	Let $D$ be a DAG over $\Omega$ with a sink $z$.
	\begin{enumerate}
		\item If $Q\simeq D-z$ and $\tilde{Q}$ is obtained from $Q$ by attaching $z$ as a sink with parent set $C_D(z)$, then $\tilde{Q}\simeq D$.
		\item If $Q\simeq D$ and $z$ is a sink of both $Q$ and $D$, then $Q-z\simeq D-z$.
		\item If $D\in\mathcal{M}(\Omega)$, then $D-z\in\mathcal{M}(\Omega\setminus\{z\})\neq\emptyset$.
		\item (Confinement is preserved.) Let $G$ be a DAG over $\Omega$, $x\in\Omega$, and $T\subseteq\Omega$ with $\mathrm{De}_G(x)\subseteq T$. If $\tilde{G}$ is obtained from $G$ by attaching a new vertex $z\notin \Omega$ as a sink with parent set $P\subseteq\Omega$ such that $P\cap T=\emptyset$, then $\mathrm{De}_{\tilde{G}}(x)\subseteq T$; in particular, $z\notin\mathrm{De}_{\tilde{G}}(x)$.
	\end{enumerate}
\end{lemma}

\begin{proof}
	(1) A sink has no outgoing links, so $z$ appears in colliders of $D$/$\tilde{Q}$ only as the center; colliders centered at $z$ are determined only by the shared skeleton and the shared parent set. Colliders among $\Omega\setminus\{z\}$: each attached link has its tail at its old endpoint and so creates no arrowhead there; hence these are the colliders of $Q$ and of $D-z$, respectively, which agree since $Q\simeq D-z$. 
	
	(2) Deleting a vertex changes neither adjacencies nor orientations among the remaining vertices, so the colliders of $D-z/Q-z$ are exactly the colliders of $D/Q$ not centered at $z$.
	
	(3) For any $a,b\neq z$ and $S\subseteq\Omega\setminus\{a,b,z\}$, a path between $a$ and $b$ that passes through the sink $z$ is blocked there: $z$ is an unconditioned collider and has no descendants. Hence $S$ d-separates $a$ and $b$ in $D$ if and only if it does so in $D-z$. Marginalizing $z$ out of $P|\Omega$ leaves all conditional-independence relations among the remaining variables unchanged. Consistency of $D$ therefore implies $D-z\in\mathcal M(\Omega\setminus\{z\})$, so $\mathcal M(\Omega\setminus\{z\})\neq\emptyset$.
	
	(4) Since $z\notin\Omega$ is attached with no outgoing links, no directed path in $\tilde{G}$ can pass \emph{through} $z$: it can only terminate there. Thus any directed path from $x$ either avoids $z$ entirely, staying within $G$, or ends at $z$ with its penultimate vertex among $z$'s parents $P$. Since $\mathrm{De}_G(x)\subseteq T$ and $P\cap T=\emptyset$, $x$ reaches no member of $P$ within $G$, ruling out the second case; hence $x$ does not reach $z$, and its reachable set among the vertices of $\Omega$ is unchanged from $G$. Thus $\mathrm{De}_{\tilde{G}}(x)=\mathrm{De}_G(x)\subseteq T$.
\end{proof}

\begin{lemma}[Sink deletion]\label{lem:min-sink}
	Let $\Omega$ be positive. If some $D\in\mathcal{M}(\Omega)$ has a sink $z\notin\{x,y\}$, then $\Omega\setminus\{z\}$ is positive. 
\end{lemma}
\begin{proof}
	By part 3 of Lemma~\ref{lem:attach}, $\mathcal{M}(\Omega\setminus\{z\})$ is nonempty and consists exactly of the DAGs $Q\simeq D-z$. For any such $Q$, part 1 of Lemma~\ref{lem:attach} yields $\tilde{Q}\simeq D$, so $\tilde{Q}\in\mathcal{M}(\Omega)$ and, by positivity, $x\Rightarrow_{\tilde{Q}}y$. The directed path avoids the sink $z$ and therefore lies in $Q$. Since $Q$ was arbitrary, $\Omega\setminus\{z\}$ is positive.
\end{proof}

\begin{corollary}\label{cor:min-sink}
    If $\Omega^*$ is minimal, then in every $M \in \mathcal{M}(\Omega^*)$: $y$ is the unique sink, and every link to $y$ is compelled.
\end{corollary}

\subsubsection{Removing proxies}\label{app:min}

\begin{lemma}[One-child peel]\label{lem:peel}
	Suppose $\Omega^*$ is minimal and not direct, and fix some $D\in\mathcal{M}(\Omega^*)$. Then there exists a strict descendant $v$ of $x$ with $v\neq y$ and 
	\begin{equation}\label{eq:onechild}
		C_D^{-1}(v)=\{y\}.
	\end{equation}
	Define
	\begin{equation}\label{eq:ABC}
		A \equiv C_D(v),
		\qquad 
		B \equiv C_D(y)\setminus\{v\},
		\qquad
		A_0 \equiv A\setminus B,
		\qquad 
		B_0 \equiv B\setminus A,
	\end{equation}
	and consider DAG $H$ over $\Omega^* \setminus\{v\}$ given by
	\begin{equation}\label{eq:H}
		H \equiv D-v+\{\forall a\in A_0: a\to y\}.
	\end{equation}
	Then $H\in\mathcal{M}(\Omega^* \setminus\{v\})$.
\end{lemma}
\begin{proof}
	Positivity and non-directness imply that some directed $x$--$y$ path in $D$ has an interior vertex, which is a strict descendant of $x$ other than $y$. Let $V$ be the set of strict descendants of $x$ that are also ancestors of $y$ in $D$. Choose $v$ maximal in $V$; i.e., choose $v$ such that there is no $w\in V\setminus\{v\}$ with $v\Rightarrow_D w$. Every child of $v$ is again a strict descendant of $x$; a child other than $y$ would contradict maximality, and $v$ has at least one child because $y$ is the unique sink (Corollary \ref{cor:min-sink}). This gives \eqref{eq:onechild}.
	
	We now show that removing $v$ produces exactly the graph $H$ of \eqref{eq:H}: the directed links $a\to y$ for $a\in A_0$, and nothing else---in particular, no bidirected links.
	
	Two distinct things could in principle go wrong when removing a vertex: (1) if $v$ had two or more children, removing $v$ would force a bidirected link between them (but $v$ has only one child, so this is excluded); and (2) removing $v$ might create a new adjacency between two of $v$'s parents, say $a_1,a_2\in A$: we need to ensure this does not happen. 
	
	To verify no unexpected links in general, consider any inducing path relative to $\{v\}$. It either avoids $v$ entirely---and is then already a link of $D-v$---or passes through $v$ as an interior vertex, in which case every other interior vertex on it must be a collider occurrence that is an ancestor of one of its endpoints, which are its own parents. We show that the second case can only produce the directed path $a\to v\to y$, which becomes the projected link $a\to y$.
	
	\begin{description}[noitemsep]
		\item[Case (i):] \emph{$v$ occurs as a non-collider on the path.} Its two incident links are then $a\to v\to y$ for some $a\in A$, so $y$ lies on the path and is an endpoint. Nor can the path extend beyond $a$: the link $a\to v$ has a tail at $a$, so an interior $a$ would be a non-collider occurrence distinct from $v$, which is impossible by definition of inducing path. Thus the path is exactly $a\to v\to y$---precisely the link the construction \eqref{eq:H} adds.
		
		\item[Case (ii):] \emph{$v$ occurs as a collider on the path.} This is the scenario that would connect two parents of $v$ to each other. If two parents of $v$ get adjacent, it means there is an inducing path from one to another through $v$. Because $v$ is a collider, it has to be an ancestor of one of the endpoints (own parents), which contradicts acyclicity. 
	\end{description}
	
	Hence the projection creates no bidirected links and exactly the adjacencies $a$--$y$ for $a\in A$, oriented $a\to y$ since $a\Rightarrow_D y$; the links not already present are those indexed by $A_0$. So $\mathrm{proj}(D,\Omega^* \setminus\{v\})=H$, an ordinary DAG, and by Lemma~\ref{lem:membership}  we have $H\in\mathcal{M}(\Omega^* \setminus\{v\})$.
\end{proof}

\begin{lemma}[Descendant-cut splice]\label{lem:splice}
	Consider distinct $H,Q \in \mathcal{M}(\Omega)$ for some $\Omega$. Fix $x \in \Omega$ and denote $S \equiv \mathrm{De}_Q(x)$. Form $\bar{Q}$ by taking orientations from $Q$ on every link with at least one endpoint in $S$, and orientations from $H$ on every link with both endpoints outside $S$.
	Then $\bar{Q}\simeq H$ and $\mathrm{De}_{\bar{Q}}(x)\subseteq S$.
\end{lemma}
\begin{proof}
	Every link between $S^c$ and $S$ is directed into $S$ in $Q$ and $\bar{Q}$ retains these orientations. The subgraphs of $\bar{Q}$ induced on $S$ and on $S^c$ are subgraphs of $Q$ and $H$ respectively, hence acyclic, and no directed cycle can cross a one-way cut, so $\bar{Q}$ is a DAG. Since no link leaves $S$, a directed path from $x$ cannot exit it, proving $\mathrm{De}_{\bar{Q}}(x)\subseteq S$.
	
	Next, every unshielded triple keeps its status: A triple centered in $S$ uses only links meeting $S$, all $Q$-oriented, so its status is as in $Q$ and hence as in $H$. A triple centered in $S^c$ with both neighbours in $S^c$ uses $H$-orientations. Finally, a triple centered at $c\in S^c$ with a neighbour in $S$ has its crossing link oriented out of $c$, a tail at the center, making it a non-collider in $\bar{Q}$; the same link carries the same tail in $Q$, so the triple is a non-collider in $Q$ and therefore in $H$ as well.
\end{proof}

\begin{lemma}[Peel trace]\label{lem:trace}
	In the setting of Lemma~\ref{lem:peel}, suppose that $\Omega^* \setminus\{v\}$ is not positive. Then
	\begin{equation}\label{eq:xinA0}
		x\in A_0
		\qquad\text{and}\qquad
		B_0 \neq\emptyset.
	\end{equation}
	Moreover, for every $c\in B_0 $, the graph $D$ contains the chordless four-cycle
	\begin{equation}\label{eq:fourcycle}
		x-v-y-c-x,
	\end{equation}
	whose unique collider is $v\to_D y\leftarrow_D c$.
\end{lemma}
\begin{proof}
	Because $\Omega^* \setminus\{v\}$ is not positive, there exists $Q\in\mathcal{M}(\Omega^* \setminus\{v\})$ such that $x\not\Rightarrow_Q y$. Let $S$ be the set of vertices reachable from $x$ in $Q$, including $x$. Thus $y\notin S$.
	
	The proof has three parts. We first show that removing $v$ must create
	a new $x$--$y$ link. We then rule out the saturated case $B_0 =\emptyset$. Finally, every $c\in B_0 $ yields the four-cycle.

	\paragraph{(a) Removing $v$ creates the $x$--$y$ link.}
	Suppose, toward a contradiction, that $x\notin A_0$.
	
	\textbf{Step a1:} We first claim that no vertex reachable from $x$ in $Q$ is adjacent to $y$ in $H$.
	Suppose otherwise that some $u\in S$ is adjacent to $y$. Because $Q\simeq H$, the same link is present in $Q$. Since $y\notin S$, it must be oriented $y\to_Q u$. Choose a directed path $x=s_0\to_Q s_1\to_Q\cdots\to_Q s_m=u$. We work backwards along this path. Suppose that $y\to_Q s_i$. If $y$ were not adjacent to $s_{i-1}$, then $y\to_Q s_i\leftarrow_Q s_{i-1}$ would be an unshielded collider in $Q$. In $H$, however, $y$ is a sink, so the $y$--$s_i$ link points from $s_i$ into $y$. The same unshielded triple would therefore be a non-collider in $H$, contradicting $Q\simeq H$. Thus $y$ is adjacent to $s_{i-1}$. Moreover, the link must be oriented $y\to_Q s_{i-1}$, since the reverse orientation would make $y$ reachable from $x$. Continuing backwards, we conclude that $x$ and $y$ are adjacent in $H$. They are not adjacent in $D$, because $\Omega^* $ is non-direct. The only new links created by removing $v$ are the links $a\to_H y$ for $a\in A_0$. Hence the presence of the $x$--$y$ link in $H$ implies $x\in A_0$. Hence, no such $u\in S$ exists.

	\textbf{Step a2:} We now use this claim to lift $Q$ back to $\Omega^* $. Since $Q\simeq H$, no member of $S$ is adjacent to $y$ in $H$. Since $Q\simeq H$, Markov equivalence is preserved when $y$ is deleted, and therefore $Q-y\simeq H-y=D-\{v,y\}$. Starting from $Q-y$, attach $v$ as a sink with parent set $A$. Then attach $y$ as a sink with parent set $\{v\}\cup B$. By two applications of Lemma~\ref{lem:attach}, the resulting graph $\tilde{M}$ satisfies $\tilde{M}\simeq D\Rightarrow \tilde{M}\in\mathcal{M}(\Omega^* )$.
	
	Every member of $A\cup B$ is adjacent to $y$ in $H$, since $A\cup B=A_0\cup B$. Step a1 therefore implies $S\cap(A\cup B)=\emptyset$. Deleting $y$ from $Q$ cannot enlarge the set reachable from $x$. When $v$ is attached, all its parents lie outside $S$, so $v$ does not become reachable from $x$. When $y$ is subsequently attached, all its parents, namely $\{v\}\cup B$, also lie outside the set reachable from $x$. By part 4 of Lemma~\ref{lem:attach},
	$x\not\Rightarrow_{\tilde{M}}y$. This contradicts the positivity of $\Omega^* $, implying $x\in A_0$.

	\paragraph{(b) The saturated case is impossible.}
	Suppose, toward a contradiction, that $B_0 =\emptyset$. We first show that every bad model must let $x$ reach a member of $B$, and then construct a bad model in which this does not happen.
	
	Because $B_0 =B\setminus A=\emptyset$, we have $B\subseteq A$, and hence $C_H(y)=A_0\cup B=A$.
	Moreover, $y$ has no children in $H$, while $v$ has no children in $D-y$. Therefore, after relabeling $y$ as $v$, the graph $H$ is identical to $D-y$. In other words, once $v$ is removed, $y$ occupies exactly the position previously occupied by $v$. We call this the \emph{twin property}.
	
	\textbf{Step b1:} Call a graph $Q'\in\mathcal{M}(\Omega^* \setminus\{v\})$ \emph{bad} if $x\not\Rightarrow_{Q'}y$. We claim that every bad graph lets $x$ reach at least one member of $B$.
	
	Suppose instead that some bad graph $Q'$ satisfies $\operatorname{De}_{Q'}(x)\cap B=\emptyset$. Relabel $y$ as $v$ in $Q'$, and call the resulting graph $J$. By the twin property, $J\simeq D-y$. Because $Q'$ is bad, $x$ does not reach $v$ in $J$. By assumption, it also reaches no member of $B$. Thus $\operatorname{De}_J(x)\cap(\{v\}\cup B)=\emptyset$. Attach $y$ to $J$ as a sink with parent set $\{v\}\cup B$, exactly as in $D$. By Lemma~\ref{lem:attach}, the resulting graph $\tilde{M}\simeq D$. Hence $\tilde{M}\in\mathcal{M}(\Omega^* )$. But none of the parents of $y$ is reachable from $x$, and therefore $x\not\Rightarrow_{\tilde{M}}y$, contradicting the positivity of $\Omega^* $. We have proved that ``every bad graph lets $x$ reach at least one member of $B$.''
	
	\textbf{Step b2:} We next show that the equivalence class nevertheless contains a bad graph in which $x$ reaches no member of $B$.
	Since $x\in A_0$ by part (a), the graph $H$ contains $x\to_H y$. A bad graph exists, so the direction of the $x$--$y$ link must be undirected in the CPDAG $K(\Omega^* \setminus\{v\})$. Let $T$ be the reversible block containing $x$ and $y$, namely the set of vertices connected to them by chains of undirected CPDAG links. Define $B_T \equiv B\cap T$ and $R \equiv \{x,y\}\cup B_T$. The members of $R$ are pairwise adjacent: the links $x$--$y$ and $b$--$y$, for $b\in B_T$, are already present. If some $b\in B_T$ were not adjacent to $x$, then $x\to_H y\leftarrow_H b$ would be an unshielded collider, which is impossible. Similarly, if two members $b,b'\in B_T$ were not adjacent, then $b\to_H y\leftarrow_H b'$ would be an unshielded collider, impossible for the same reversible block. 
	
	We now use a standard ordering property of CPDAGs. Within a reversible block, any pairwise-adjacent set of vertices can be placed first in a causal ordering, in any prescribed order, while remaining in the same Markov equivalence class.\footnote{Equivalently, CPDAG chain components are chordal, and any clique can be placed last in a perfect elimination ordering. Orienting the links in the reverse order creates neither a directed cycle nor a new unshielded collider; see \citet{AnderssonMadiganPerlman1997}.}
	Apply this property to $R$. Choose a member $Q^*$ of the equivalence class whose ordering within $T$ begins with $B_T\succ y\succ x$, followed by the remaining vertices of $T$. Orient every undirected link within $T$ from the earlier to the later vertex in this ordering. It follows that $x$ cannot reach either $y$ or any member of $B_T$ along a directed path contained in $T$. Nor can such a path leave $T$ and later return. The directions between reversible blocks are fixed and form an acyclic ordering of those blocks. Therefore $x\not\Rightarrow_{Q^*}y$ and $\operatorname{De}_{Q^*}(x)\cap B_T=\emptyset$.
	
	Finally, take $b\in B\setminus B_T$. The link $b\to y$ runs from the reversible block containing $b$ into $T$ and is compelled. If $x$ could reach $b$, there would be a directed path from $T$ to $b$'s block, followed by the compelled link from that block back into $T$.
	This would create a directed cycle among the reversible blocks. Therefore $x$ cannot reach any member of $B\setminus B_T$ either. We have constructed a bad graph $Q^*$ satisfying $\operatorname{De}_{Q^*}(x)\cap B=\emptyset$, so we conclude that $B_0 \neq\emptyset$.

	\paragraph{(c) The four-cycle.} Fix $c\in B_0 =B\setminus A$. Since $x\in A_0$ and $c\in B$, the graph $H$ contains
	$x\to_H y\leftarrow_H c$. Return to the bad graph $Q$ fixed at the beginning of the proof. Because $Q$ and $H$ have the same skeleton, $x$ and $y$ are adjacent in $Q$, hence $y\to_Q x$.
	
	It follows that $x$ and $c$ must be adjacent. Otherwise $x\to_H y\leftarrow_H c$ would be an unshielded collider in $H$, contradicting $y\to_Q x$. Because removing $v$ creates only links incident to $y$, the $x$--$c$ adjacency is already present in $D$.
	
	Finally, $x\to_D v\to_D y\leftarrow_D c$. There is no $x$--$y$ link because $\Omega^* $ is non-direct. There is also no $v$--$c$ link: the orientation $c\to_D v$ would imply $c\in A$, contrary to $c\in B_0 $, while $v\to_D c$ would contradict the choice of $v$, whose only child is $y$. Thus $x-v-y-c-x$ is chordless with its unique collider being $v\to_D y\leftarrow_D c$.
\end{proof}

\subsubsection{Proof of Theorem~\ref{thm:direct}}\label{app:proof2}

\begin{proof}[Proof of Theorem~\ref{thm:direct}]
	If $x\perp y$, then by Lemma~\ref{lem:consistent_models} no consistent model on any $\Omega_s$ contains $x\Rightarrow_s y$; every nonempty proposal fails and is strictly dominated by silence, so $M_s=\emptyset$ for either receiver type. Assume for the rest of the proof that $x\not\perp y$.
	
	With a na{\"i}ve receiver, Proposition~\ref{prop:pers_naive} implies that among the successful proposals, the minimum-cost proposal discloses only $\Omega_s = \{x,y\}$ and claims $x\to_s y$. The sender's overall optimum is therefore either this proposal with $x,y$ adjacent, or silence.
	
	With a sophisticated receiver, nonempty unsuccessful proposals are strictly dominated by silence. Otherwise, the sophisticated receiver accepts a proposal $M_s=(\Omega_s,C_s)$ with $x\Rightarrow_s y$ if and only if $\Omega_s$ is positive. Conversely, if $\Omega_s$ is positive, the sender may propose \emph{any} member $M_s\in\mathcal{M}(\Omega_s)$: it is consistent and contains a compelled directed path $x \to_s y$, hence it is accepted. The sender's payoff from disclosing $\Omega_s$ is, therefore, $\mathbf{1}\{\Omega_s\text{ positive}\}-\kappa(|\Omega_s|)$. Thus, if any positive disclosure is available, the optimal $\Omega_s$ is minimal. 
	
	To see that the minimal $\Omega_s$ is direct, proceed by contradiction. If $x,y$ are not adjacent in some $M_s \in \mathcal{M}(\Omega_s)$, there exists a proxy variable $v \in \Omega_s$ with $x \Rightarrow_s v \to_s y$, and since $\Omega_s$ is minimal, $\Omega_s \setminus\{v\}$ is not positive. But then Lemma~\ref{lem:trace} implies that $M_s$ contains a chordless four-cycle with exactly one collider occurrence, contradicting the simplicity restriction in Lemma~\ref{lem:one-collider}. 
\end{proof}


\subsection{Proof of Theorem~\ref{thm:debunk_simple}}

Most of the work below is devoted to proving the second part of Theorem~\ref{thm:debunk_simple}. The key step is to show that under the stated conditions, the existence of one non-obvious cause implies the existence of another independent non-obvious cause. The two together can then be used in the same way as in Proposition~\ref{prop:persuade_sophist}, although we follow a different approach to avoid any interference from the variables that may already be known to the receiver.

We next show that if the true model is rich and there is no obvious cause of $y$ given $x$, then there must exist a collider upstream from $x$. To show this, we need only the following result, which is trivial but convenient to have as a lemma.

\begin{lemma} \label{lem:flip_src}
	Let $M=(\Omega,C)$ be a model consistent with $P|\Omega$, $u \in \Omega$ be such that $C(u)=\emptyset$, and $b \in \Omega$ be such that $C(b)=\{u\}$. Let $M'$ be $M$ with the edge $u\to_{M} b$ replaced by $u \leftarrow_{M'} b$. Then $M'$ is consistent with $P|\Omega$.
\end{lemma}
\begin{proof}
	Flipping the link between a source $u$ and its child $b$ that has no other parents does not create or destroy any colliders. Therefore, $M'$ is a DAG with the same skeleton and the same colliders as $M$. Then the results of \citet{verma_equivalence_1990} imply $M'$ is observationally equivalent to $M$.
\end{proof}

\begin{lemma} \label{lem:Vstr_upstream_direct}
	Suppose $M_{t}$ is simple and rich, and $x \Rightarrow_t y$. If there is no obvious cause of $y$ given $x$, then there exists a collider $a \to_{t} b \leftarrow_{t} c$ with $a,c \in \bar C_t(x)$, $a \perp c$, and $b \in \bar C_t(x) \cup \{x\}$.
\end{lemma}
\begin{proof}
	Denote $\bar{A}\equiv\bar C_t(x)\cup\{x\}$ and note that the parents of any member of $\bar{A}$ belong to $\bar C_t(x)$. In particular, if a collider is centered in $\bar{A}$, its parents automatically lie in $\bar C_t(x)$, and simplicity of $M_t$ implies that $a \perp c$. It is thus enough to prove the existence of a collider centered in $\bar A$.
	
	\emph{Step 1: $\bar C_t(x)\neq\emptyset$.} Suppose not, so $\bar C_t(x) = \emptyset$ and $x$ is ancestral in $M_{t}$. Let $B_{xy}\equiv C_t^{-1}(x)\cap\bigl( \bar C_t(y) \cup \{y\} \bigr)$ denote the set containing the first vertex of any directed path from $x$ to $y$. Since $x\Rightarrow_{t} y$, $B_{xy}$ is nonempty. Choose $b\in B_{xy}$ such that $B_{xy} \cap C_t(b) = \emptyset$. We claim that $C_t(b)=\{x\}$. Toward a contradiction, suppose there exists $d \in C_t(b) \setminus \{x\}$. Since $b\in\bar C_t(y)\cup\{y\}$ and $d\to_t b$, we have $d\in\bar C_t(y)$. 
	If $d$ and $x$ are nonadjacent, then $d\to_t b\leftarrow_{t} x$ is a collider and the three variables form a V-structure in the data. By simplicity, it has no controls, hence $d \perp x$ in the data---but then $d$ is an obvious cause of $y$ given $x$, contradicting the premise of the lemma. 
	If $d$ and $x$ are adjacent, $C_t(x)=\emptyset$ implies $x \to_t d \to_t b \Rightarrow_t y$, contradicting the definition of $b$. 
	Hence $C_t(b)=\{x\}$, and by Lemma~\ref{lem:flip_src}, reversing $x\to b$ yields a model $M'\neq M_{t}$ consistent with $P$, which contradicts richness. We conclude that $\bar C_t(x)\neq\emptyset$.
	
	\emph{Step 2: the collider exists.} Suppose, by way of contradiction, that no collider is centered in $\bar{A}$. Pick any $a\in\bar C_t(x)$ and follow iteratively the parental relations. Since parents of ancestors of $x$ are ancestors of $x$ and $M_{t}$ is finite and acyclic, the process terminates at some ancestral node $u\in\bar C_t(x)$ with $C_t(u)=\emptyset$. As above, define $B_{ux} \equiv C_t^{-1}(u)\cap\bar{A}$ and choose $b\in B_{ux}$ such that $B_{ux} \cap C_t(b) = \emptyset$. By the same logic, we obtain that $C_t(b)=\{u\}$. Then by Lemma~\ref{lem:flip_src}, the reversal of $u\to b$ yields $M'\neq M_{t}$ consistent with $P$, contradicting richness. We conclude that the collider exists.
\end{proof}

The following two lemmas show that the parents of the collider we identified in Lemma~\ref{lem:Vstr_upstream_direct} will indeed constitute two independent non-obvious causes under the assumptions of the second part of Theorem~\ref{thm:debunk_simple}.

\begin{lemma}[In unconfounded worlds, all ancestors of $x$ are non-obvious causes]\label{lem:unconf_to_noc}
	If $x \Rightarrow_{t} y$ and $\tilde{C}_{t}(x,y)=\emptyset$, then for every $w\in\bar C_t(x)$, $x$ d-separates $w$ and $y$ in $M_{t}$. Consequently, $(w \perp y \mid x)$.
\end{lemma}
\begin{proof}
	If $\tilde{C}_{t}(x,y)=\emptyset$, then every $w\in\bar C_t(x)$ admits a directed path to $x$ avoiding $y$.
	
	Suppose, toward a contradiction, that $x$ does not d-separate $y$ and some $w \in \bar C_t(x)$. By Definition~\ref{def:dsep2}, there exists a walk $W$ between $w$ and $y$ on which the only possible collider occurrence is $x$.
	
	\emph{Case (i): $x$ does not occur on $W$.} Then $W$ is collider-free and has the fork shape $w \Leftarrow_t s \Rightarrow_t y$ for some $s \neq x$. 
	\begin{itemize}[nosep]
		\item If $s=w$, then $W$ is a directed path $w \Rightarrow_t y$ avoiding $x$, implying $w \in \tilde{C}_{t}(x,y)$, a contradiction to $\tilde{C}_t(x,y) = \emptyset$.
		\item If $s=y$, then $W$ is a directed path $w \Leftarrow_t y$, which forms a cycle with $w \Rightarrow_t x\Rightarrow_t y$.
		\item If $s$ is interior, then $s \Rightarrow_t y$ via a path avoiding $x$ and $s \Rightarrow_t w \Rightarrow_t x$ via a path avoiding $y$ (otherwise there is a cycle, as above), so $s\in\tilde{C}_{t}(x,y)$, a contradiction to $\tilde{C}_t(x,y) = \emptyset$.
	\end{itemize}

\emph{Case (ii): $x$ occurs on $W$.} Let $W'$ be the suffix of walk $W$ beginning at the last occurrence of $x$. Every collider occurrence of $W$ is an occurrence of $x$, so $W'$ is collider-free. Moreover, because this last occurrence of $x$ is a collider in $W$, the first edge of $W'$ has an arrowhead at $x$. Hence $W'$ has the fork form $x\Leftarrow_t s\Rightarrow_t y$ for some $s\neq x$. If $s=y$, then $y\Rightarrow_t x\Rightarrow_t y$ is a cycle. Otherwise $s$ is interion in $W'$, which implies that there exists a path $s \Rightarrow_t y$ avoiding $x$ (by the last-occurrence choice of $W'$) and a path $s \Rightarrow_t x$ avoding $y$ (by acyclicity), meaning $s\in\tilde C_t(x,y)$, a contradiction to $\tilde{C}_t(x,y) = \emptyset$.
\end{proof}

\begin{lemma}[Non-obvious causes exist only in unconfounded worlds]\label{lem:noc_to_unconf}
	Let $x \Rightarrow_{t} y$. A non-obvious cause of $y$ given $x$ exists if and only if $\tilde{C}_{t}(x,y)=\emptyset$ and $\bar C_t(x)\neq\emptyset$.
\end{lemma}

\begin{proof}
	($\Longleftarrow$) The ``if'' part is immediate from Lemma~\ref{lem:unconf_to_noc}: if $\tilde{C}_{t}(x,y)=\emptyset$, every element of $\bar C_t(x)$ is a non-obvious cause.
	
	($\Longrightarrow$) To show the ``only if'' part, let $w\in\bar C_t(x)$ be a non-obvious cause and suppose, by way of contradiction, that some $c\in\tilde{C}_{t}(x,y)$ exists. Then there exists a walk $w \Rightarrow_t x \Leftarrow_t c \Rightarrow_t y$.
	Its unique collider occurrence is $x$. Hence by Definition~\ref{def:dsep2}, $w$ and $y$ are d-connected given $x$ in $M_{t}$, meaning $w$ is not a non-obvious cause---a contradiction.
\end{proof}

Finally, the following two lemmas illustrate which data properties debunk a particular link.

\begin{lemma}[Collider certificate at $y$]\label{lem:cert_oc}
	Let $x,y,z\in\Omega$ be distinct. If the data satisfy $x \perp z$ and $(x \not\perp z \mid y)$ then $x \not\Leftarrow_M y$ for any model $M$ consistent with $P|\Omega$.
\end{lemma}
\begin{proof}
	From $(x \not\perp z \mid y)$ and the Markov property, $x$ and $z$ are d-connected given $y$ in $M$. By Definition~\ref{def:dsep2}, there exists a walk between $x$ and $z$ on which the only possible collider occurrence is $y$. Since $x \perp z$, faithfulness implies $x$ and $z$ are d-separated in $M$ (by $S=\emptyset$), so every \emph{path} between the two contains a collider. Combining the two, there exists a walk $x \Leftarrow_M s_1 \Rightarrow_M y \Leftarrow_M s_2 \Rightarrow_M z$ with $s_1,s_2$ distinct from $y$.
	
	If there exists a path $x \Leftarrow_M y$, there would also exist a walk $x \Leftarrow_M y \Leftarrow s_2 \Rightarrow_M z$ between $x$ and $z$ without a collider, which contradicts $x\perp z$. Hence $x \not\Leftarrow_M y$.
\end{proof}

\begin{lemma}[Collider certificate at $x$]\label{lem:cert_noc}
	Let $v,w,x,y\in\Omega$ be distinct. If the data satisfy all of the following, then $x \not\Leftarrow_M y$ for any model $M$ consistent with $P|\Omega$:
	\[
	\emph{(a)}\;\; v \perp w,
	\qquad
	\emph{(b)}\;\; v \not\perp w \mid x,
	\qquad
	\emph{(c)}\;\; v \perp y \mid x.
	\]
\end{lemma}
\begin{proof}
	By the same argument as in the proof of Lemma~\ref{lem:cert_oc}, observations (a) and (b) jointly imply that there exists an active walk in $M$ between $v$ and $w$ with a collider occurrence at $x$. Hence, some such walk has the form $v \Leftarrow_M s_1 \Rightarrow_M x \Leftarrow_M s_2 \Rightarrow_M w$ with $s_1,s_2$ distinct from $x$.
	
	If there exists a path $x \Leftarrow_M y$, then combining it with the above implies there exists a walk $v \Leftarrow_M s_1 \Rightarrow_M x \Leftarrow_M y$, meaning $v$ and $y$ are d-connected given $x$. By Faithfulness, this contradicts observation (c), hence no such path exists, and $x \not\Leftarrow_M y$.
\end{proof}

\begin{remark}[Monotonicity and inoculation]
	The statements in Lemmas \ref{lem:cert_oc} and \ref{lem:cert_noc} concern the joint distribution $P$ restricted to three or four variables. So once the certificate variables ($z$ or $v,w$) are disclosed, the conclusion $x \not\Leftarrow_M y$ persists under any further disclosure by any future sender. This is the inoculation property discussed informally in Section \ref{sec:applications}. 
\end{remark}

\begin{proof}[Proof of Theorem~\ref{thm:debunk_simple}]
	\textbf{Case 1.} We verify the two data facts required by Lemma~\ref{lem:cert_oc} for the triple $(x,y,z)$ and then apply it. The first fact, $x \perp z$, follows immediately from the definition of an obvious cause.
	
	To show $(x \not\perp z \mid y)$, proceed as follows. 
	Since $x,y$ are adjacent in $M_r$, $(x \not\perp y \mid S)$ for any $S \subseteq \Omega_r\setminus\{x,y\}$ and in particular, $x \not\perp y$. The Markov property implies $x$ and $y$ must be d-connected in $M_t$, so there must be a collider-free walk between $x$ and $y$ in $M_t$: $x \Leftarrow_t s \Rightarrow_t y$. Note that $s \neq y$, since the link $x \leftarrow_r y$ is assumed defective.	
	Since $z \in \bar{C}_t(y)$, there then exists a walk $x \Leftarrow_t s \Rightarrow_t y \Leftarrow_t z$ with $s \neq y$. This walk has only one collider occurrence, at $y$, so Definition~\ref{def:dsep2} implies $x$ and $z$ are d-connected given $y$. Faithfulness then implies $(x \not\perp z \mid y)$.
	
	Now let $\Omega_{s} \equiv \Omega_{r}\cup\{z\}$. By Lemma~\ref{lem:cert_oc}, $x \not\Leftarrow_M y$ for any $M$ consistent with $P|\Omega_s$. So any such model debunks the link $x \leftarrow_r y$ and the receiver's model $M_r$, while revealing just one new variable. (This also implies that $z \notin \Omega_r$.)
	
	\medskip
	\textbf{Case 2.} By Lemma~\ref{lem:noc_to_unconf}, the existence of a non-obvious cause implies $\tilde{C}_{t}(x,y)=\emptyset$. If an obvious cause of $y$ given $x$ exists, Case~1 applies. Assume therefore that no obvious cause exists. 
	Then Lemma~\ref{lem:Vstr_upstream_direct} applies, so there exists a direct V-structure $a\to_{t} b \leftarrow_{t} c$ with $a,c\in\bar C_t(x)$ and $b\in\bar C_t(x)\cup\{x\}$. We verify the three data facts required by Lemma~\ref{lem:cert_noc} for the pair $(v,w)=(a,c)$.
	\begin{description}[nosep]
		\item\emph{Fact (a):} $a \perp c$. Simplicity of $M_t$ means the V-parents $a$ and $c$ are not correlated in $M_{t}$, so consistency of $M_t$ implies that $a \perp c$ in the data. 
		
		\item\emph{Fact (b):} $(a\not\perp c\mid x)$. Consider the path $a\to_t b\leftarrow_t c$ in $M_{t}$. Its only interior vertex is the collider $b$, and either $b=x$, or $b\Rightarrow_{t} x$. In either case, the collider condition of Definition~\ref{def:dsep} fails for $S=\{x\}$, so $x$ does not d-separate $a$ and $c$ in $M_t$. Faithfulness implies $(a\not\perp c\mid x)$.
		
		\item\emph{Fact (c):} $(a\perp y\mid x)$. Since $a\in\bar C_t(x)$, $x\Rightarrow_{t} y$, and $\tilde{C}_{t}(x,y)=\emptyset$, Lemma~\ref{lem:unconf_to_noc} applies.
	\end{description}
	
	Now let $\Omega_{s}=\Omega_{r}\cup\{a,c\}$. By Lemma~\ref{lem:cert_noc}, $x \not\Leftarrow_M y$ for any $M$ consistent with $P|\Omega_s$. So any such model debunks the link $x \leftarrow_r y$ and the receiver's model $M_r$ by revealing at most two new variables. (At least one of $a$ and $c$ is new to the receiver.)
\end{proof}

\subsection{Proofs for Section \ref{sec:microfound}} \label{sec:ambig_pf}

\begin{proof}[Proof of Proposition~\ref{prop:ambig_optactions}]
	Fix some consistent $M_s$.
	Observe that the receiver's value from inaction $x=0$ is always zero: $V_R(x=0 \mid M_s) = 0$.
	
	Let $\mathcal{N}(\Omega_s) \equiv \left\{ M \in \mathcal{M}(\Omega_s) \mid x \not\Rightarrow_M y \right\}$ denote the set of models that say action $x$ does not affect outcome $y$ (and note this set may be empty). Then
	\begin{align*}
		V_R(1 \mid M_s) = \min \Bigg\{ 
		&\min_{M \in \mathcal{N}(\Omega_s)} \left\{-1 + \frac{1}{\theta} \psi(M \mid M_s) \right\},
		\\
		&\min_{M \in \mathcal{M}(\Omega_s) \setminus \mathcal{N}(\Omega_s)} \left\{ b(M)-1 + \frac{1}{\theta} \psi(M \mid M_s) \right\}
		\Bigg\}.
	\end{align*}

	\paragraph{Case 1.}
	If $\mathcal{N}(\Omega_s)$ is empty then $V_R(1 \mid M_s) > 0$, since by assumption, $b(M)>1$ for all $M \in \mathcal{M}(\Omega_s) \setminus \mathcal{N}(\Omega_s)$. Therefore, action $x=1$ is optimal for the receiver, $x^*(M_s) = 1$, and $V_S(M_s)=1-\kappa(|\Omega_s|)$.

	\paragraph{Case 2.}
	Suppose that $\mathcal{N}(\Omega_s)$ is nonempty and the receiver is skeptical ($\theta > \bar{\psi}$). Then 
	\begin{align*}
		-1 + \frac{1}{\theta} \psi(M \mid M_s) &< -1 + \frac{1}{\bar{\psi}} \psi(M \mid M_s)
		\leq -1 + \frac{1}{\bar{\psi}} \bar{\psi} = 0.
	\end{align*}
	Hence, $V_R(1 \mid M_s) < 0 = V_R(0 \mid M_s)$, so $x^*(M_s) = 0$ and $V_S(M_s)=0-\kappa(|\Omega_s|)$.

	\paragraph{Case 3.}
	Suppose that $\mathcal{N}(\Omega_s)$ is nonempty and the receiver is credulous ($\theta < \underline{\psi}$).
	Then for all $M \in \mathcal{N}(\Omega_s) \setminus \{M_s\}$:
	\begin{align*}
		-1 + \frac{1}{\theta} \psi(M \mid M_s) & > -1 + \frac{1}{\underline{\psi}} \psi(M \mid M_s)
		\geq -1 + \frac{1}{\underline{\psi}} \underline{\psi} = 0,
	\end{align*}
	and for all $M \in \mathcal{M}(\Omega_s) \setminus \left( \mathcal{N}(\Omega_s) \cup \{M_s\} \right)$:
	\begin{align*}
		b(M)-1 + \frac{1}{\theta} \psi(M \mid M_s) > \frac{1}{\theta} \psi(M \mid M_s) \geq 0
	\end{align*}
	because $b(M) > 1$ for all such $M$. It follows that the value in $V_R(1 \mid M_s)$ is strictly positive given all $M \neq M_s$. At $M_s$, the penalty is zero. The corresponding value is $-1$ if $M_s\in \mathcal{N}(\Omega_s)$, and $b(M_s)-1>0$ otherwise. Hence $x^*(M_s)=1$ if and only if $x\Rightarrow_s y$, in which case $V_S(M_s)=1-\kappa(|\Omega_s|)$. 
\end{proof}

\begin{proof}[Proof of Proposition~\ref{prop:ambig_equiv_str}.]
The sender's strategy space is the same in both settings, and silence yields a payoff of zero in both. Fix any consistent proposal $M_s = (\Omega_s, C_s)$. By Proposition~\ref{prop:ambig_optactions}, a credulous receiver chooses $x^*(M_s) = 1$ if and only if $x \Rightarrow_s y$, which is exactly the
acceptance rule of a na{\"i}ve receiver, so $x^*(M_s) = v(M_s)$. Likewise, a skeptical receiver chooses $x^*(M_s) = 1$ if and only if $x \Rightarrow_M y$ for all $M \in \mathcal{M}(\Omega_s)$, which is the
acceptance rule of a sophisticated receiver (note that this condition implies $x \Rightarrow_s y$, since $M_s \in \mathcal{M}(\Omega_s)$), so
again $x^*(M_s) = v(M_s)$. In either case,
$V_S(M_s) = x^*(M_s) - \kappa(|\Omega_s|) = v(M_s) - \kappa(|\Omega_s|) = U_S(M_s)$: the payoff functions coincide on the common strategy space,
so the sets of optimal strategies coincide.
\end{proof}

\begin{corollary}\label{cor:thm2-game}
If $M_t$ is simple, then the sender's optimal proposal in the game of Section~\ref{sec:microfound}, with either a credulous or a skeptical receiver, is either
$M_s = \varnothing$ or a model in which $x$ and $y$ are adjacent.
\end{corollary}

\begin{proof}
By Proposition~\ref{prop:ambig_equiv_str}, the sender's optimal strategies in the game with a credulous (skeptical) receiver coincide with those in the baseline model with a na{\"i}ve (sophisticated) receiver, which by Theorem~\ref{thm:direct} are either silence or a direct model.
\end{proof}

\bibliographystyle{apalike}
\bibliography{literature}

\newpage
\section{Supplementary Appendix: Additional Examples} \label{sec:examples}

\subsection{Consistent models, IC algorithm}

\begin{figure}
	\begin{center}
		\begin{minipage}[b]{0.24\textwidth}
			\centering 
			\includegraphics[scale=0.25]{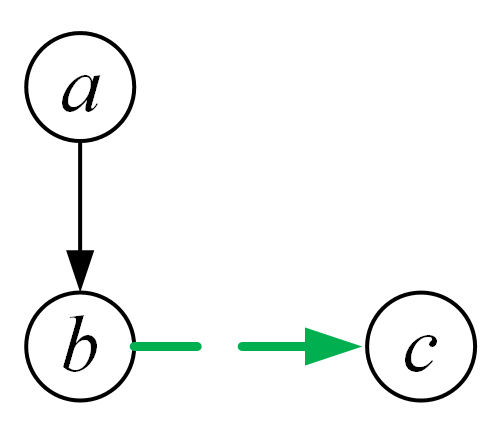}\\
			R1\\
		\end{minipage}
		\begin{minipage}[b]{0.24\textwidth}
			\centering 
			\includegraphics[scale=0.25]{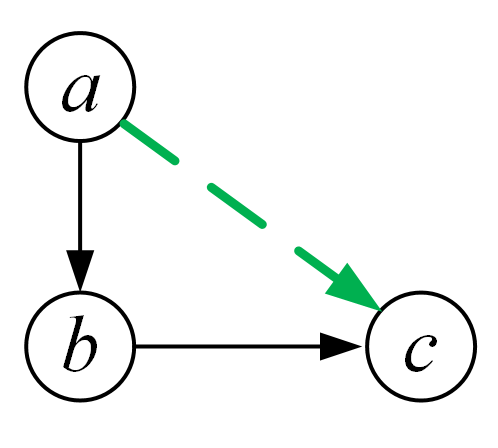}\\
			R2\\
		\end{minipage}
		\begin{minipage}[b]{0.24\textwidth}
			\centering	
			\includegraphics[scale=0.25]{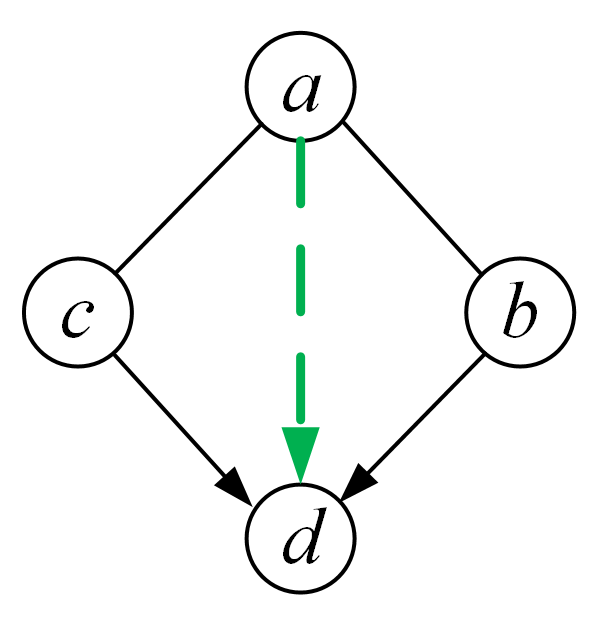}\\
			R3\\
		\end{minipage}
		\begin{minipage}[b]{0.24\textwidth}
			\centering	
			\includegraphics[scale=0.25]{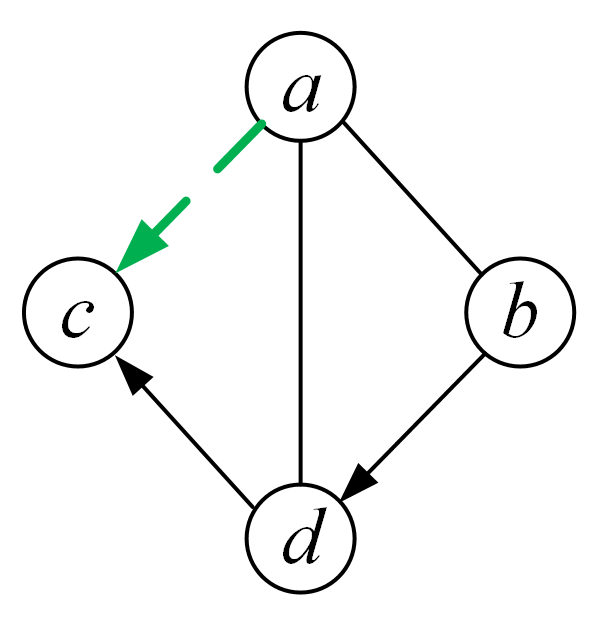}\\
			R4\\
		\end{minipage}
		
		\caption{\citet{meek_causal_1995} rules. \label{fig:Meek}}
	\end{center}
	\footnotesize
	\textsc{Notes:} Green long-dashed arrows represent the undirected links oriented by the respective rules.
\end{figure}

\citet{meek_causal_1995} shows that Step 3 of the IC algorithm is equivalent to the iterative application of four orientation rules depicted in Figure~\ref{fig:Meek}. Rule \textbf{R2} follows immediately from step 3(b) of the IC algorithm: link $a \leftarrow c$ would create a cycle, which is ruled out by acyclicity, hence it must be that $a \to c$. The other three rules follow from step 3(a). For \textbf{R1}: if the true model was such that $b \leftarrow c$, then $(a,b,c)$ form a collider which would create a V-structure in the data that should have been identified in step 2. Since it was not, $(a,b,c)$ cannot form a collider, so it must be that $b \to c$. For \textbf{R3}: if $a \leftarrow d$, then R2 orients the remaining links as $c \to a \leftarrow b$, which would be a V-structure. For \textbf{R4}: if $c \to a$, then R2 orients the remaining links as $d \to a \leftarrow b$, meaning $(c,a,b)$ would be a V-structure.

\begin{figure}
	\centering
	\begin{minipage}[b]{0.45\textwidth}
		\centering 
		\includegraphics[scale=0.275]{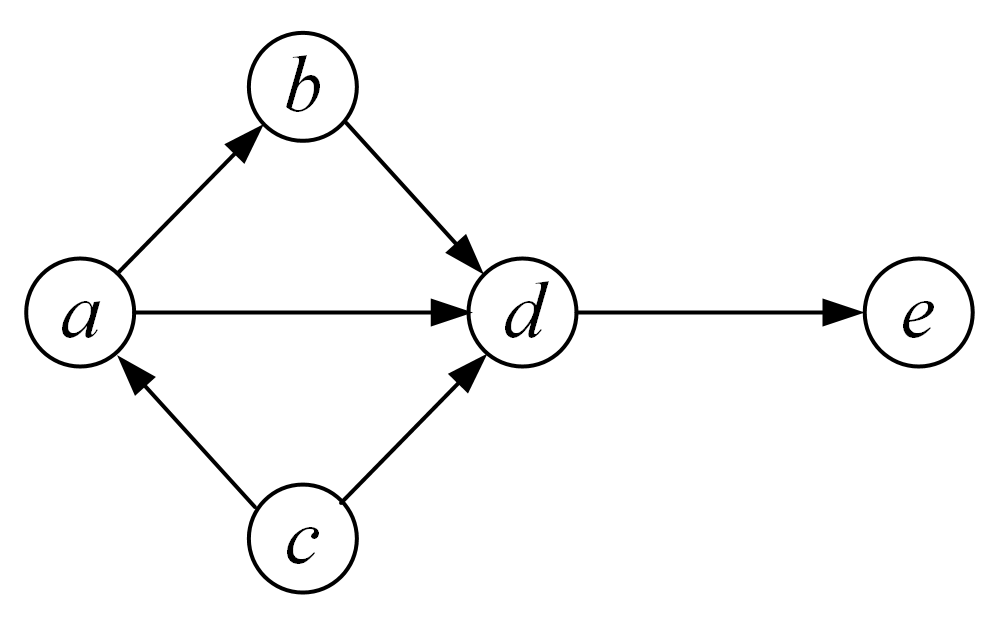}\\
		(a) true model\\
	\end{minipage}
	\begin{minipage}[b]{0.45\textwidth}
		\centering 
		\includegraphics[scale=0.275]{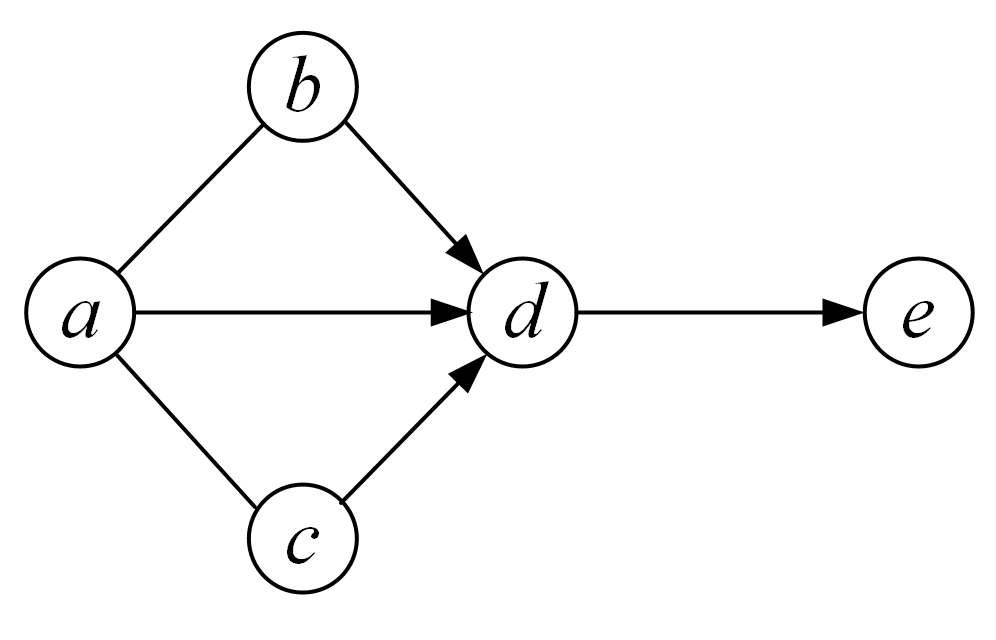}\\
		(b) IC algorithm output\\
	\end{minipage}
	
	\caption{Applying the IC algorithm. \label{fig:IC}}
\end{figure}
\begin{example}[Applying the IC algorithm] \label{exp:IC}
	Suppose the true model is as in Figure~\ref{fig:IC}(a), and we use the IC algorithm on data generated by it. In Step 1, the algorithm identifies pairs of variables that are never conditionally independent. So it connects the pairs $ab$, $ac$, $ad$, $bd$, $cd$, and $de$ with undirected links. In Step 2, it identifies the only V-structure $c\rightarrow d\leftarrow b$ with control $a$, since $(c\perp b\mid a)$ and $(c\not\perp b \mid a,d)$. In Step 3: (1) we apply R1 to $b\rightarrow d - e$ and direct $d\rightarrow e$; and (2) we apply R3 to the rectangle $\{a,b,d,c\}$ and direct the link $a\rightarrow d$. No further links can be oriented. The resulting partially directed graph is shown in Figure~\ref{fig:IC}(b). 
\end{example}

\begin{figure}
	\centering
	\begin{minipage}[b]{0.45\textwidth}
		\centering 
		\includegraphics[scale=0.275]{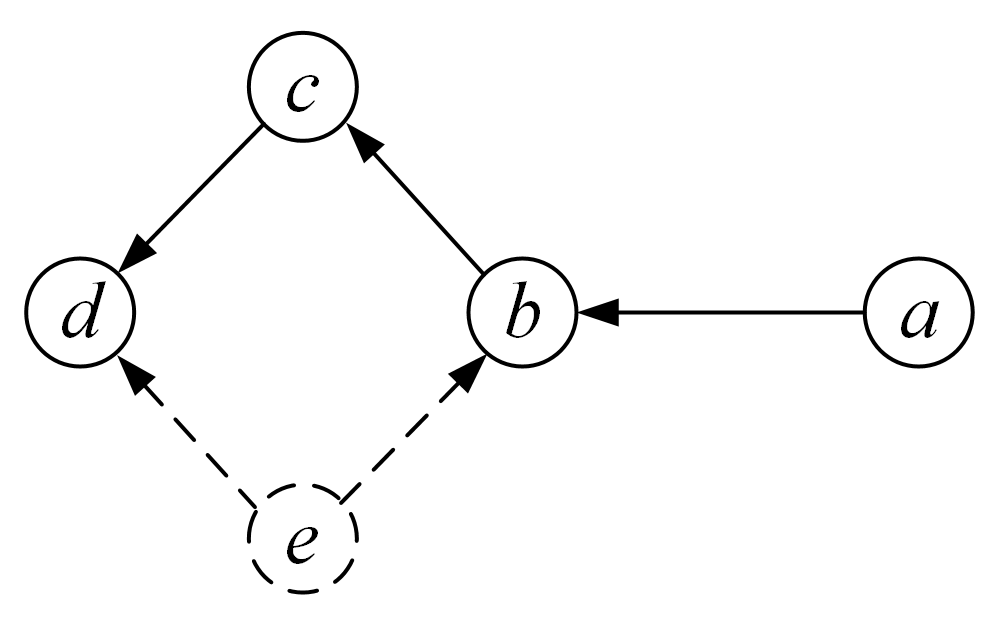}\\
		(a) true model\\
	\end{minipage}
	\begin{minipage}[b]{0.45\textwidth}
		\centering 
		\includegraphics[scale=0.275]{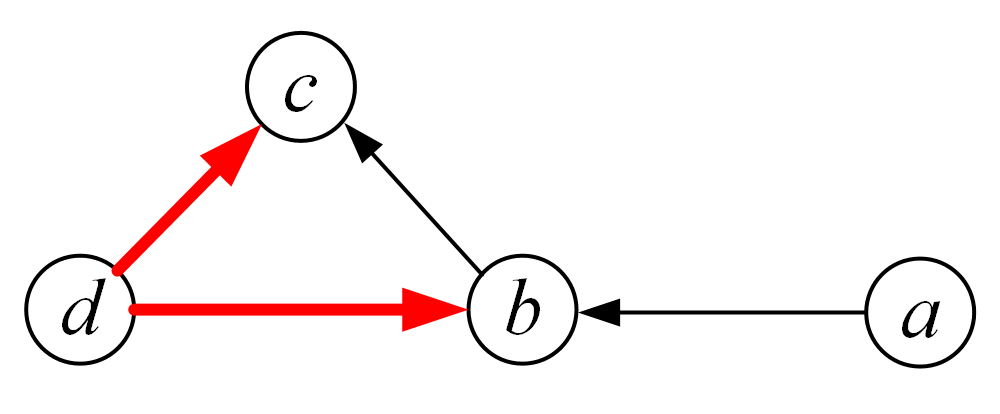}\\
		(b) IC algorithm output\\
	\end{minipage}
	
	\caption{No models are consistent with the data on $\Omega=\{a,b,c,d\}$. \label{fig:no}}
\end{figure}
\begin{example}[IC algorithm may produce an inconsistent model] \label{exp:IC_incons}
	Suppose the true model is as in Figure~\ref{fig:no}(a). Running the IC algorithm with variables $\Omega = \{a,b,c,d\}$ (so $e \notin \Omega$) produces the model in Figure~\ref{fig:no}(b). Specifically, in Step 1, the IC algorithm identifies which variables are linked. There is no $S \subseteq \{a,c\}$ such that $\left(d \perp b \mid S\right)$, hence $d$ and $b$ are adjacent. In Step 2, the algorithm identifies the V-structure $d\rightarrow b \leftarrow a$ with control $c$. In Step 3, the algorithm: (1) orients link $b\rightarrow c$ using rule R1 (which avoids creating a V-structure $a,b,c$ that would have been identified); and (2) orients $d\rightarrow c$ using R2 (which avoids a cycle). However, the resulting model is not consistent with the data, since it violates the Markov property: the model suggests that $a$ and $d$ should be independent, but in the data they are not, $(a\not\perp d)$. This is happening because there does not exist any consistent model for this set of variables $\Omega$. 
\end{example}

\subsection{Persuasion without a pre-existing model}

\begin{example}[Persuading a na{\"i}ve receiver---Vaccination and autism] \label{exp:vaccines}
	Suppose vaccination rates for some disease are positively correlated with the number of confirmed cases of that disease, due to more vaccinations in response to (or in anticipation of) an outbreak. Presenting the data on these two correlated variables to a na{\"i}ve receiver with a claim that vaccines cause the disease would be sufficient to persuade them in this narrative.
\end{example}

%
%

\begin{figure}
	\centering
	\includegraphics[scale=0.275]{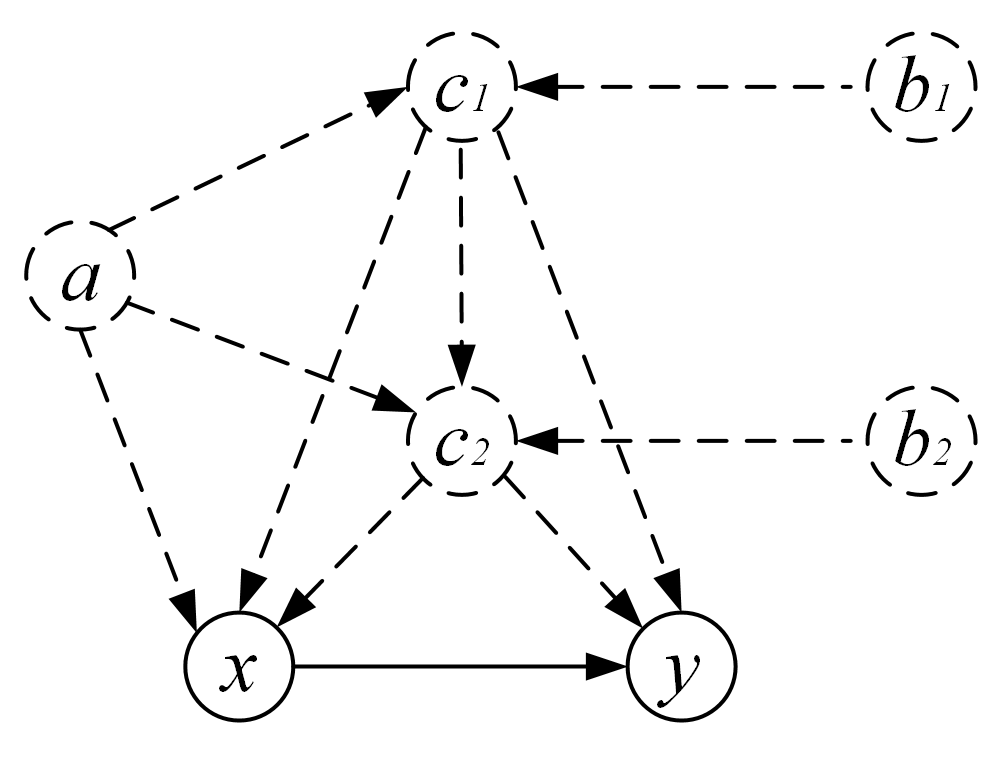}	
	\caption{Persuasion without a cause. \label{fig:nocauses}}
\end{figure}
\begin{example}[Persuasion without a cause can be easy, difficult, or impossible] \label{exp:nocauses}
	Suppose the true model is such that $\Omega_t = \{x,y,a,b_1,...,b_n,c_1,...,c_n\}$, and the causal graph is given by $\{a,b_i\} \to_t c_i \to_t \{x,y\}$ for all $i \in \{1,...,n\}$, $a \to_t x$, $x \to_t y$, and $c_i \to_t c_j$ for any $i>j$. Such a model for $n=2$ is depicted in Figure~\ref{fig:nocauses}. This model contains no obvious or non-obvious causes of $y$ given $x$ (all other variables are confounding for $x,y$). Further, the model is rich, so disclosing the true model would persuade a sophisticated receiver. To see this, run the IC algorithm on the full distribution $P$. Then all undirected links are identified in Step 1. Step 2 identifies V-structures $a \to c_i \leftarrow b_i$ for all $i$, and $c_i \to c_j \leftarrow b_j$ for all $i,j$ with $i>j$.\footnote{In all of the respective colliders in the true model, pairs $(a,b_i)$ and $(c_i,b_j)$ are not correlated, meaning that the model is also simple.} Step 3 of the IC algorithm then orients, for all $i$: 
	\begin{align*}
		b_i \to c_i &\overset{R1}{\Rightarrow} c_i \to \{x,y\},
		\\
		a \to c_i \to x &\overset{R2}{\Rightarrow} a \to x,
		\\
		a \to x &\overset{R1}{\Rightarrow} x \to y.
	\end{align*} 
	
	Persuading a sophisticated receiver that $x \to_s y$ in this case \emph{requires} revealing all variables in $\{c_1, ..., c_n\}$. This is because orienting $x \to_s y$ can only be done using R1 as above, which requires that $a$ is not adjacent to $y$ in $M_s$, meaning that $\Omega_s \supset \{c_1, ..., c_n\}$.\footnote{Orienting $x \to_s y$ further requires that $a \to_s x$, which can also only be oriented using R2 as in the true model and thus requires that $a \to_s c_i \to_s x$ for some $i$ were oriented. The latter requires that $b_i \in \Omega_s$ for some $i$.}
	
	Since $n$ is arbitrary, this means that arbitrarily many variables may be required to persuade a sophisticated receiver in this case---i.e., persuasion can be arbitrarily difficult. Conversely, if $n=1$, persuasion can be ``easy-ish'', requiring only three variables $(a,b_1,c_1)$. Finally, if variable $a$ did not exist, then $M_t$ would not be rich, and it would be impossible to persuade the receiver that $x \to_s y$.
\end{example}

\subsection{Debunking and persuasion with a pre-existing model}

We now present a few applied examples of how causal misperceptions can be debunked, according to Theorem~\ref{thm:debunk_simple}. For the examples that follow, we assume that the true model $M_t$ is as in Figure~\ref{fig:causes_copy}, the receiver's initial model is given by $x \leftarrow_r y$, and all the respective variables are as in different rows of Table \ref{tab:exp}.
A disclaimer is needed that the problems discussed in these examples are complex and multi-faceted. The authors do not claim to have any easy answers and do not claim that the ``true models'' in these examples capture reality accurately. Instead, we assume for sake of the argument that the true causal links in every model go in a certain way, and demonstrate how these links could be proven by causal narratives.

\begin{example}[Election results and opinion polls] \label{exp:polls}
	Suppose an autocrat attempts to suppress the results of the pre-election public opinion polls under the pretense that poll results $y$ affect voters' preferences $x$ and thereby bias the election outcomes. The opposition (sender) wants to debunk this narrative to persuade the voters (receivers) to mobilize and protect electoral transparency and free speech. Suppose further the autocrat's narrative is indeed false (defective) in this example, and voters' preferences are correlated with poll results simply because the latter accurately represent the former ($x \to_t y$).
	The opposition can then debunk it by revealing that weather on the day of the poll biases the poll results without affecting the voters' actual preferences.\footnote{\citet{damsbo-svendsen_when_2023} show in a meta-analysis that weather on the voting day affects voting behavior. By a similar logic, it is reasonable to argue that weather on a polling day affects the poll results.} 
	This makes weather an obvious cause $z$ in this example.
	Alternatively, the opposition could reveal two non-obvious causes $v,w$, which is any pair of independent variables that affect public opinion but have no effect on poll results, except \emph{through} changes in public opinion. Such variables can include state of the economy (growth, inflation, or unemployment) and foreign policy events (wars and conflicts or trade agreements and expressions of support).
\end{example}
\begin{table}
	\centering 
	\begin{tabular}{c | c | c | c}
		$v,w$ 	& $x$ 	& $y$	& $z$
		\\ \hline
		$\begin{array}{c} 
			\text{economy},\\ \text{foreign policy}
		\end{array}$ & voters preferences & poll results & poll-day weather
		\\ \hline
		$\begin{array}{c}
			\text{party lines}, \\ \text{demographics, culture}, \\ \text{voting system} 
		\end{array}$ & polarization & SoMe algorithms 
		& platform design
		\\ \hline
		$\begin{array}{c}
			\text{supply shocks}\\ \text{(wars, nat. disasters)}\\
			\text{demand shocks} \\ \text{(mark-up, foreign demand)}
		\end{array}$ & inflation & interest rates & $\begin{array}{c}
			\beta\text{-shock},\\ \text{demographics}
		\end{array}$
		\\ \hline 
	\end{tabular}
	\caption{Examples of obvious and non-obvious causes. \label{tab:exp}}
\end{table}
\begin{example}[Social media and polarization] \label{exp:polarization}
	Consider now a politician (receiver) who believes that the recommendation algorithms of social media platforms $y$ lead to social polarization $x$ which is undesirable and warrants regulating platforms. Suppose the truth is the converse: that polarized individuals simply prefer polarized media sources, and by engaging with such content ``teach'' the recommendation algorithm \citep{robertson_users_2023}. 
	The platform (sender) would like to debunk the receiver's view and to show the truth. One way to do so is to reveal an obvious cause $z$, such as some aspect of platform design that affects the composition of the users' feed (together with the data on an internal experiment varying this aspect and relating it to which items users engage with). 
	An alternative approach would be to reveal two non-obvious causes $v,w$ that do not affect social media recommendations except through their effect on social polarization. These could be the sentiment of politicians' statements and cross-sectional variation from demographics or culture. 
\end{example}
\begin{example}[Interest rate and inflation] \label{exp:inflation}
	Suppose a voter (receiver) believes that the interest rate $y$ set by the central bank drives inflation $x$, seeing how the two are correlated. The conservative central bank (sender) would like to explain that its monetary policy actually responds to inflation in an attempt to curb the business cycles. To do so with data, the sender can show another factor $z$ that affects interest rates but is not correlated with inflation, which would be the obvious cause in this case. Any variable affecting the savings rate, such as ageing population\footnote{``The impact of population ageing on monetary policy.'' Vox EU, Mar 05, 2019. URL: \url{https://cepr.org/voxeu/columns/impact-population-ageing-monetary-policy}} or a shock to the intertemporal discount factor $\beta$, would be the natural candidate.
	Alternatively, the central bank could disclose independent, non-obvious determinants $v,w$ of inflation, such as aggregate supply and demand shocks, including geopolitical disruptions, pandemics, natural disasters, as well as mark-up shocks and foreign demand shocks. These affect inflation without entering the policy rule directly, and therefore influence interest rates only through the central bank's endogenous response.
\end{example}

\subsection{Debunking: special cases}

\begin{figure}
	\centering
	
	\begin{minipage}[b]{0.24\textwidth}
		\centering
		
		\includegraphics[scale=0.225]{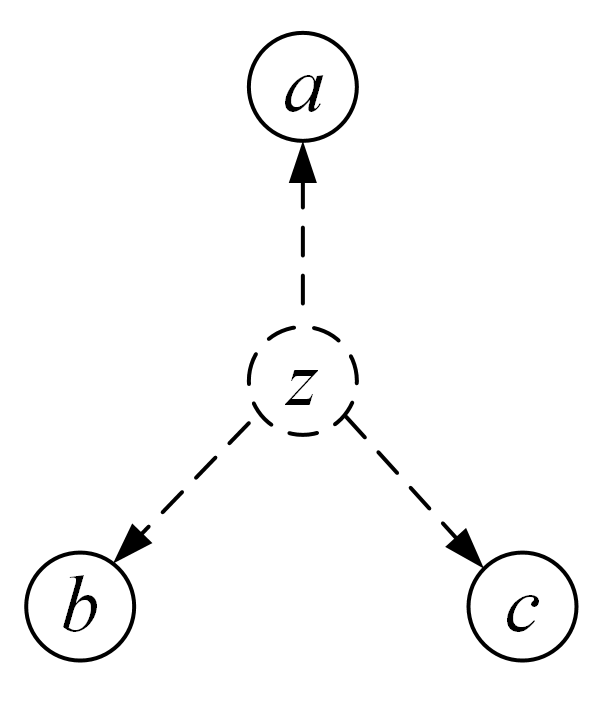}\\
		
		(a) $M_t$\\
	\end{minipage}
	\hfill
	\begin{minipage}[b]{0.24\textwidth}
		\centering
		
		\includegraphics[scale=0.225]{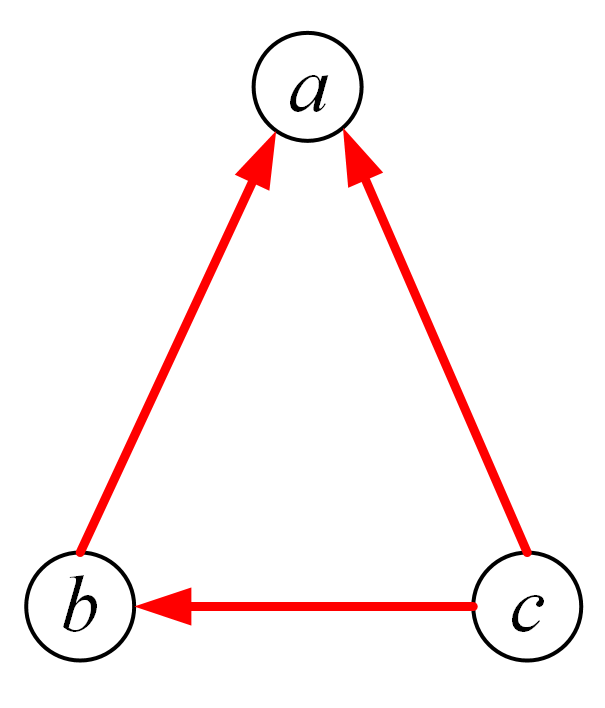}\\
		
		(b) $M_r$\\
	\end{minipage}
	\hfill
	\begin{minipage}[b]{0.24\textwidth}
		\centering
		
		\includegraphics[scale=0.225]{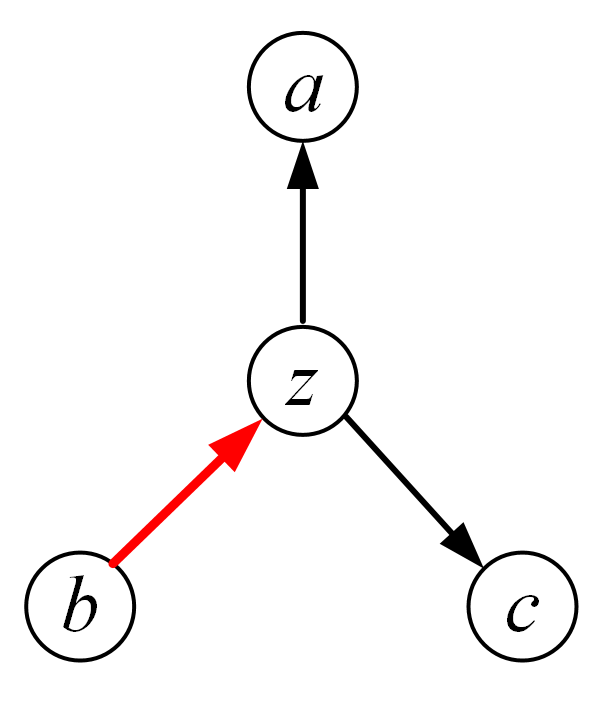}\\
		
		(c) consistent $M$\\
	\end{minipage}
	\hfill
	\begin{minipage}[b]{0.24\textwidth}
		\centering
		
		\includegraphics[scale=0.225]{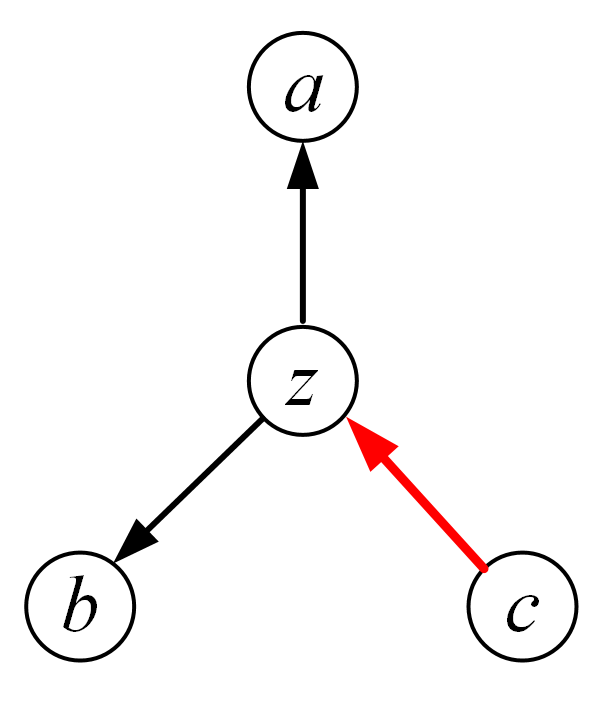}\\
		
		(d) consistent $M$\\
	\end{minipage}
	\caption{Debunking of a link is sufficient but not necessary. \label{fig:debunk_nolink}}
\end{figure}
\begin{example}[A model may be debunked without any individual link being debunked]\label{ex:debunk_nolink}
	Consider a true model with four variables, $z\to_t a,b,c$, depicted in Figure~\ref{fig:debunk_nolink}(a). Suppose the receiver is initially aware of $\Omega_r=\{a,b,c\}$. All three variables are pairwise correlated, regardless of whether we condition on the remaining one, hence they must all be adjacent in any consistent model. Suppose then that $M_r$ is given by $c \to_r b \to_r a$, $c \to_r a$, as in Figure~\ref{fig:debunk_nolink}(b).
	
	Suppose the sender reveals $z$, so that $\Omega_s=\{a,b,c,z\}$. The expanded data identify the skeleton of a star centered at $z$, with no collider at $z$. Several orientations of this star are consistent with $P|\Omega_s$. In particular, a model $b \to z \to a,c$ from Figure~\ref{fig:debunk_nolink}(c) is consistent with $P|\Omega_s$ and preserves the receiver's link $b\to_r a$ as the causal path $b\Rightarrow a$. Similarly, model $c \to z \to a,b$ from Figure~\ref{fig:debunk_nolink}(d) is consistent with $P|\Omega_s$ and preserves both $c\to_r a$ and $c\to_r b$. Thus, each link in $M_r$ is individually compatible with $P|\Omega_s$, although requiring different models.
	
	Furthermore, no single model consistent with $P|\Omega_s$ preserves all three links jointly. In particular, preserving $b\Rightarrow a$ requires $b \to z \to a$, whereas preserving $c\Rightarrow b$ requires $c \to z \to b$. The former requires $b\to z$, while the latter requires $z\to b$, hence it is impossible to attain both in a single DAG. Hence $M_r$ has no consistent extension to $\Omega_s$ and is therefore debunked, even though none of its links are debunked.
	
	The example illustrates that link-level debunking is sufficient but not necessary for model-level debunking. A link is individually debunked only if every model consistent with the expanded data rules it out. By contrast, the receiver's entire model is debunked whenever no single consistent model preserves all of its links simultaneously.
\end{example}

\begin{example}[Debunking with an inconsistent model] \label{exp:decept}
	Suppose the true model is as in Figure~\ref{fig:decept}(a), and the receiver's model is $w \to_r x \leftarrow_r y$, depicted in Figure~\ref{fig:decept}(b). This model can be debunked by revealing just one variable, $z$, which is an obvious cause of $y$ given $x$. However, then both $w,x,y$ and $x,y,z$ constitute V-structures, with the former implying $x \leftarrow y$ and the latter implying $x \to y$, as depicted in Figure~\ref{fig:decept}(c). These contradicting implications mean that there is no consistent model for $\Omega = \{w,x,y,z\}$. If the sender is required to propose a consistent model, then to debunk the receiver's model they must also disclose confounder $c$ in addition to the obvious cause $z$ and propose the true model from Figure~\ref{fig:decept}(a).
\end{example}
\begin{figure}
	\centering
	
	\begin{minipage}[b]{0.32\textwidth}
		\centering
		
		\includegraphics[scale=0.27, trim={15pt 0 10pt 0}, clip]{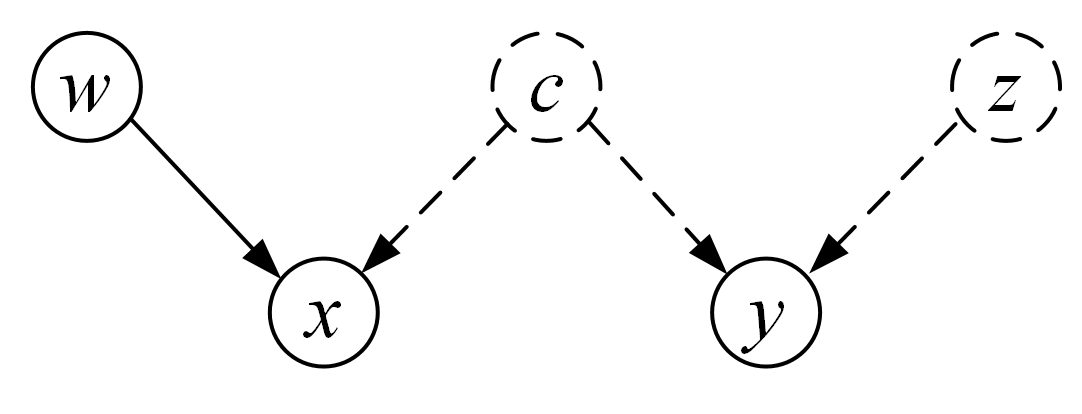}\\
		(a) $M_t$\\
	\end{minipage}
	\hfill
	\begin{minipage}[b]{0.32\textwidth}
		\centering
		
		\includegraphics[scale=0.27, trim={15pt 0 10pt 0}, clip]{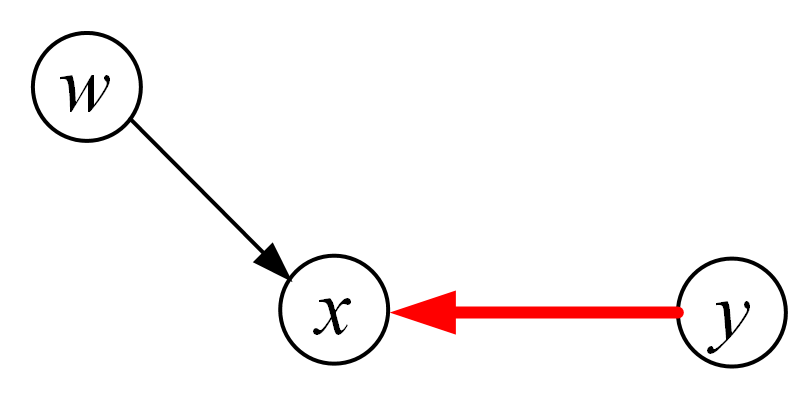}\\
		(b) $M_r$\\
	\end{minipage}
	\hfill
	\begin{minipage}[b]{0.32\textwidth}
		\centering
		
		\includegraphics[scale=0.27, trim={15pt 0 10pt 0}, clip]{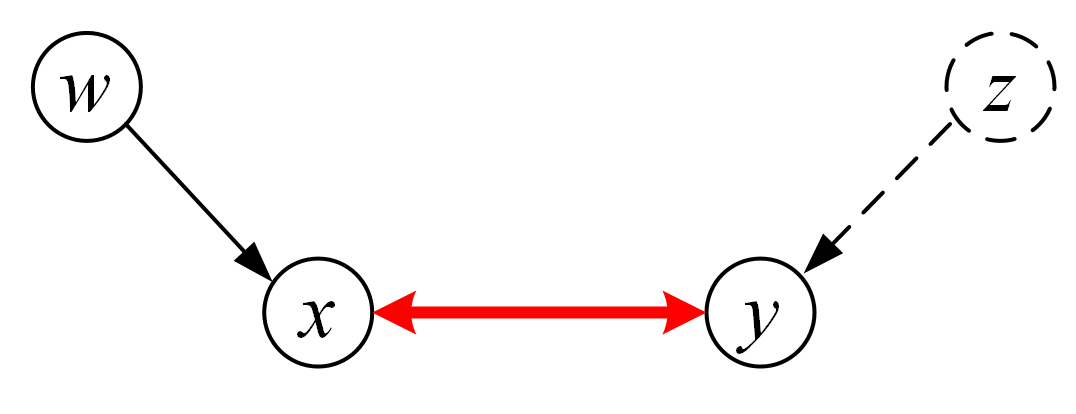}\\
		(c) $M_s$\\
	\end{minipage}
	
	\caption{Debunking may not lead to persuasion. \label{fig:decept}}
\end{figure}

\begin{example}[Nitpicking not always possible] \label{exp:nitpick_no}
	Suppose the true model is as shown in Figure~\ref{fig:nitpick_no}(a), and the na{\"i}ve receiver's subjective model is as in Figure~\ref{fig:nitpick_no}(b). The receiver correctly believes $x \leftarrow_r y$ but incorrectly believes $a\to_r y$. The sender would like to persuade the receiver that $x \to_s y$. However, in this case, the targeted link $a \leftarrow_t y$ and the link of interest $x \leftarrow_t y$ both trace back to the same V-structure $c,y,b$. Thus, to debunk $a\to_r y$, the sender has to reveal $b$ and $c$. But revealing $b$ and $c$ also makes $x \to_s y$ inconsistent with the data, so deception via nitpicking is not possible in this case.
\end{example}
\begin{figure}
	\centering
	
	\begin{minipage}[b]{0.45\textwidth}
		\centering 
		\includegraphics[scale=0.3]{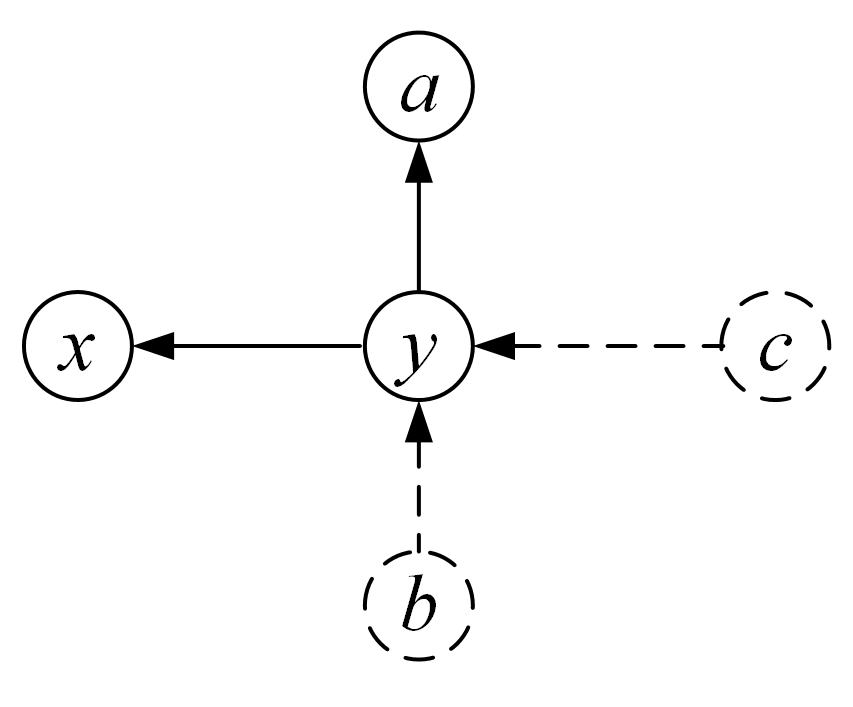}\\
		(a) $M_t$ \\
	\end{minipage}
	\begin{minipage}[b]{0.45\textwidth}
		\centering 
		\includegraphics[scale=0.3]{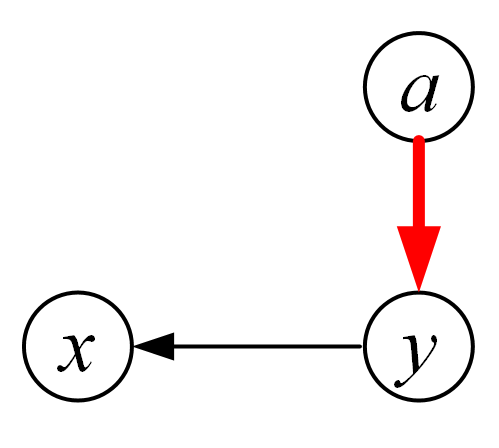}\\
		(b) $M_r$ \\
	\end{minipage}
	
	\caption{Example: persuasion by nitpicking not possible. \label{fig:nitpick_no}}
\end{figure}

\subsection{Dissuasion from a pre-existing model}

\begin{example}[Polio and sugar---dissuading of a link] \label{exp:polio}
	In the beginning of the 20th century, a major outbreak of localized paralytic polio epidemics led to over 27,000 cases and more than 6,000 deaths in the US. Due to recurring summer epidemics, by the 1940s-1950s polio was widely feared as a devastating threat \citep{mehn14}.
	During that period, Dr. Benjamin Sandler was a strong proponent of the theory that refined sugars cause polio in children. His advocacy resulted in a substantial drop in sales of ice cream, cookies and Coca Cola in Asheville, NC, where he practiced \citep{williams2013}. 
	His theory was based on the observation that polio epidemics occur in summers, when children increase their sugar intake (ice cream in particular). In reality, as it was later discovered, polio spread primarily through fecal-oral transmission facilitated by close contact and poor hygiene.\footnote{``Clinical Overview of Poliomyelitis'' CDC, May 9, 2024. URL: \url{https://www.cdc.gov/polio/hcp/clinical-overview/index.html}, retrieved April 3, 2026.} 
	The correlation observed by Sandler had warm weather as the confounding variable---it led to increased social interactions, particularly in swimming pools and rivers (which fostered polio transmission), as well as increased ice cream consumption.
\end{example}

\begin{example}[Dissuading a misled na{\"i}ve receiver] \label{exp:dissuade_2}
	To see how the sender can rule out a link between $x$ and $y$ when $x \leftarrow_r y$ and $x \Rightarrow_t y$, see Figure~\ref{fig:dissuade_2}. The true model is presented in panel (a). The na{\"i}ve receiver's model in panel (b) has defective link $x \leftarrow_r y$, whereas in the true model $x \Rightarrow_t y$. The sender can propose the model in Figure~\ref{fig:dissuade_2}(c), debunking the na{\"i}ve receiver's defective model: V-structure $b,y,a$ immediately implies $b \to_s y \leftarrow_s a$, which is incompatible with $x \Leftarrow_s y$ and $b$ being only a proxy for this channel. The sender can then instill the belief that $x \leftarrow_s b \to_s y$, contrary to the true model, since this model is consistent with $P|\Omega_s$ despite being defective.
\end{example}
\begin{example}[Dissuading a correct na{\"i}ve receiver] \label{exp:dissuade_1}
	To see how the sender can rule out a link between $x$ and $y$ when $x \leftarrow_r y$ and $x \Leftarrow_t y$, see Figure~\ref{fig:dissuade_1}. The true model is presented in panel (a) and is simple and rich. The na{\"i}ve receiver's model in panel (b) has defective link $x \to_r a$, whereas in the true model $x \leftarrow_t a$. The sender can exploit this and reveal variables $b,d$, proposing the model in panel (c), which is also defective but compatible with the data $P|\Omega_s$. It debunks the na{\"i}ve receiver's defective model and instills the belief that $x \leftarrow_s d \to_s y$, contrary to the true model.
\end{example}
\begin{figure}
	\centering
	
	\begin{minipage}[b]{0.32\textwidth}
		\centering
		
		\includegraphics[scale=0.225]{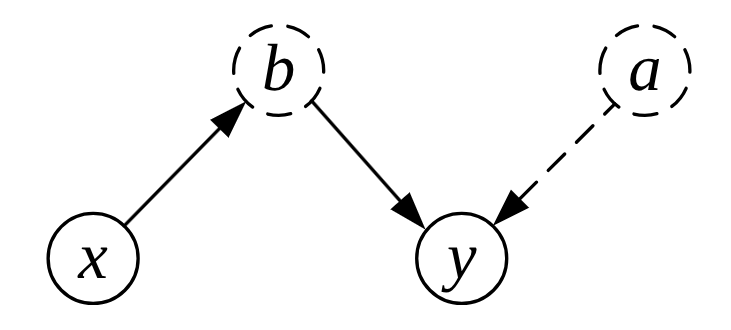}\\
		(a) $M_t$\\
	\end{minipage}
	\hfill
	\begin{minipage}[b]{0.32\textwidth}
		\centering
		
		\includegraphics[scale=0.225]{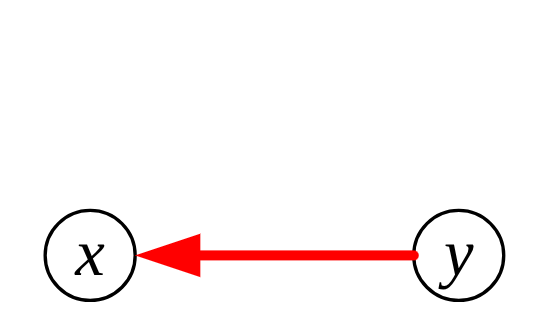}\\
		(b) $M_r$\\
	\end{minipage}
	\hfill
	\begin{minipage}[b]{0.32\textwidth}
		\centering
		
		\includegraphics[scale=0.225]{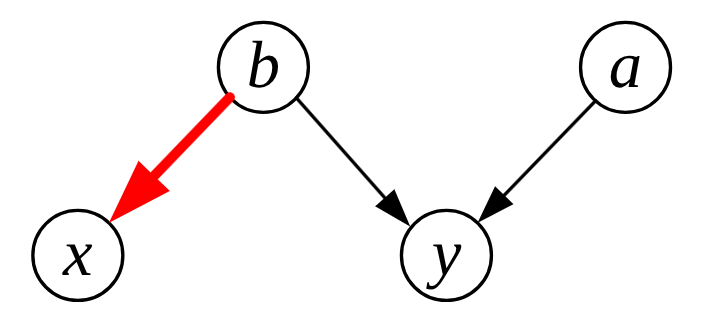}\\
		(c) $M_s$\\
	\end{minipage}
	
	\caption{Ruling out a defective link: $x \leftarrow_r y$, $x \Rightarrow_t y$. \label{fig:dissuade_2}}
\end{figure}
\begin{figure}
	\centering
	
	\begin{minipage}[b]{0.32\textwidth}
		\centering
		
		\includegraphics[scale=0.3]{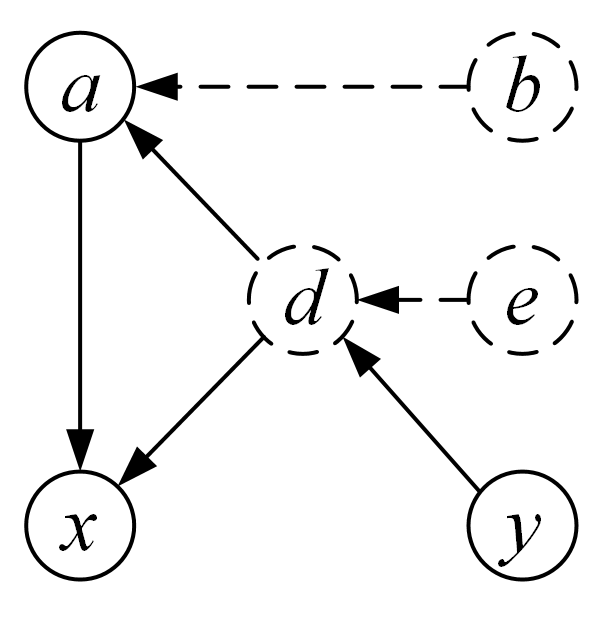}\\
		(a) $M_t$\\
	\end{minipage}
	\hfill
	\begin{minipage}[b]{0.32\textwidth}
		\centering
		
		\includegraphics[scale=0.3]{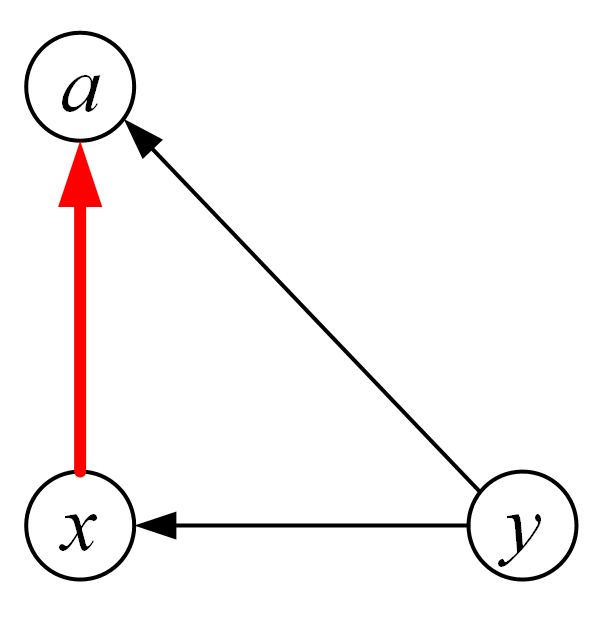}\\
		(b) $M_r$\\
	\end{minipage}
	\hfill
	\begin{minipage}[b]{0.32\textwidth}
		\centering
		
		\includegraphics[scale=0.3]{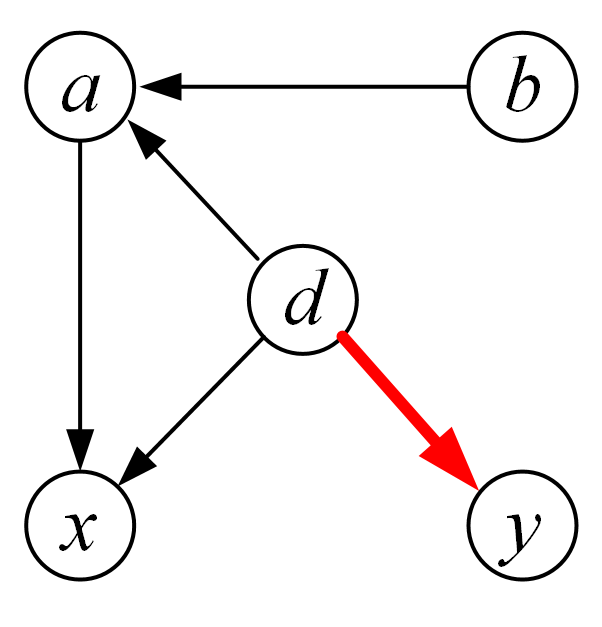}\\
		(c) $M_s$\\
	\end{minipage}
	
	\caption{Ruling out a correct link by nitpicking: $x \leftarrow_r y$, $x \Leftarrow_t y$. \label{fig:dissuade_1}}
\end{figure}

\end{document}